\documentclass[11pt]{article}

\usepackage{amsmath,amsthm,latexsym,amssymb,amsfonts,epsfig}

\newtheorem{Definition}{Definition}[section]

\newtheorem{Proposition}{Proposition}[section]

\newcommand{\be}{\begin{equation}}
\newcommand{\ee}{\end{equation}}
\newcommand{\ba}{\begin{eqnarray}}
\newcommand{\ea}{\end{eqnarray}}

\title{{\sf 
Observations on representations of the spatial diffeomorphism}\\
 {\sf group and algebra in all dimensions}} 
\author{
{\sf T. Thiemann}$^1$\thanks{{\sf 
thomas.thiemann@gravity.fau.de}}\\
\\
{\sf $^1$ Inst. for Quantum Gravity, FAU Erlangen -- N\"urnberg,}\\
{\sf Staudtstr. 7, 91058 Erlangen, Germany}\\
}
\date{{\small\sf \today}}

\makeatletter
\@addtoreset{equation}{section}
\makeatother

\begin{document} 

\maketitle

{\sf

\begin{abstract}
The canonical quantisation of General Relativity including matter 
on a spacetime manifold in the globally 
hyperbolic setting involves in particular the representation theory of the 
spatial diffeomorphism group (SDG), and/or its Lie algebra (SDA),
of the underlying spatial submanifold. 
There are well known Fock representations of the SDA in one spatial 
dimension and non-Fock representations of the SDG in all dimensions.
The latter are not strongly continuous and do not descend to representations 
of the SDA.

In this work we report some partial results on non anomalous 
representations of the SDA for both geometry and matter:
1. Background independent Fock representations of the SDA by operators 
exist in all dimensions.
2. Infinitely many unitary equivalence classes
of background dependent Fock representations of the SDA 
by operators exist in one dimension but these do 
not extend to higher dimensions.
3. Infinitely many unitary equivalence classes
of background dependent Fock representations 
of the SDA  
of volume preserving diffeomorphisms by operators exist in all dimensions. 
4. Infinitely many unitary equivalence classes 
of background dependent Fock representations of the 
SDA by quadratic forms exist in all dimensions.    
 
Except for 1. these representations do not descend from an invariant state 
of the Weyl algebra and 4. points to a new strategy for solving 
the quantum constraints.
\end{abstract}

\section{Introduction}
\label{s1}

To our best knowledge, matter is described by the Quantum Field Theory 
(QFT) of the standard model and geometry by General Relativity (GR). 
When it comes to taking the interaction between geometry and matter into 
account which in the classical regime is dictated by the Einstein equations,
one necessarily must also quantise GR and enter the realm of Quantum Gravity
(GR). The task is now to construct a QFT for both matter and geometry and 
to properly deal with the gauge symmetries of the combined theory. These 
are the Yang-Mills type gauge symmetries which are already present when 
the interaction between geometry and matter can be neglected and the 
additional gauge symmetries that arise because the full theory is now 
generally covariant. These latter gauge symmetries form a group, the 
group of diffeomorphisms (all smooth invertible coordinate transformations)
Diff(M) of the (D+1)-manifold $M$ over which the theory is constructed. 

If one insists 
that the classical theory be predictive, i.e. the initial value problem 
well posed, then $M$ is necessarilly diffeomorphic to 
$\mathbb{R}\times \sigma$ where $\sigma$ is a fixed D-manifold \cite{1}
and the spacetime $(M,g)$ is globally hyperbolic where $g$ is the spacetime
metric field on $M$ of Lorentzian signature. 
This means that $M$ can be foliated by a one parameter family of leaves
$\Sigma_t$, i.e. D-manifolds, each of which is diffeomorphic to $\sigma$. 
Accordingly, there exist diffeomorphisms $F:\; \mathbb{R}\times \sigma
\to M;\; (t,x)\mapsto F(t,x)$ such that $\Sigma_t=F(t,\sigma)$ is a spacelike
hypersurace in $(M,g)$, i.e. the vectors tangent to the leaves are 
spacelike. We call such a foliation therefore ``spacelike''.
 
Given a spacelike foliation $F$, Diff(M) contains the subgroup $Diff_F(M)$ of 
diffeomorphisms that preserve the leaves. These are evidently of the 
form $\Phi(X)=F(t,\varphi(x))_{X=F(t,x)}$ where $\varphi\in$Diff$(\sigma)$
is a diffeomorphism of $\sigma$.
Thus while Diff$_F(M)$ depends of the foliation, Diff$(\sigma)$ is a universal 
object, i.e. independent of the foliation. 
We call this the spatial diffeomorphism group (SDG) $\mathfrak{G}$ 
and its Lie algebra 
the spatial diffeomorphism algebra (SDA) $\mathfrak{a}$. Here the Lie algebra 
is identified with the tangent space of $\mathfrak{G}$ at the identity and 
thus its generators are obtained by differentiation at $s=0$ of 1-parameter 
subgroups $s\mapsto \varphi_s,\; \varphi_r\circ\varphi_s=\varphi_{r+s},\;
\varphi_0={\rm id}_{\mathfrak{G}}={\rm id}_\sigma$. These 1-parameter
subgroups $\varphi^u_t$ are in 1-1 correspondence with 
the integral curves $c^u_x(s)$ through $x\in \sigma$ of vector fields 
$u$ on $\sigma$ defined by 
$c^u_x(0)=x,\;\frac{d}{ds} c^u_x(s)=u(c^u_x(s))$ via 
$\varphi^u_s(x)=c^u_x(s)$ \cite{2}. Thus the Lie algebra 
$\mathfrak{g}=$diff($\sigma$) is isomorphic with the Lie algebra of vector 
fields on $\sigma$.

The remaining diffeomorphisms could therefore be called ``temporal'' as they 
do not preserve the leaves and thus map points $X$ in ``one instant of time'' 
$\Sigma_t$ generically to a point $X'$ in a different instant of time 
$\Sigma_{t'},\;t'\not=t$.

The globally hyperbolic structure delivers a natural starting point for a 
canonical
approach to quantisation of the entire system. One considers arbitrary
spacelike  
foliations $F$ and pulls back all fields by $F$ upon which they become 
fields on the universal manifold $\mathbb{R}\times \sigma$. The trace 
that the auxiliary structure $F$ leaves on the system is that the Legendre 
transform between the Lagrangian and the Hamiltonian is singular, after all
the action (integral over $M$ of the Lagrangian) is independent of $F$.
This has the effect that the Hamiltonian formulation is subject to initial 
value constraints. These are the so-called spatial diffeomorphism constraints
(SDC) and Hamiltonian constraints (HC). Their Poisson bracket algebra is 
isomorphic to the algebra generated by spatial and temporal spacetime 
diffeomorphisms mentioned above when the fields are ``on shell'', i.e. solve
the classical Einstein equations \cite{3}.

However, ``off-shell'' the Poisson algebra
in contrast to the algebra of diffeomorphisms, is no longer a Lie algebra
but rather an algebroid \cite{4}. That is, the Poisson bracket of constraints
is a linear combination of constraints but the coefficients of that linear 
combination are not structure constants but rather structure functions of
the fields. This happens because a general diffeomorphism in Diff(M) maps 
a spacelike foliation $F$ to another foliation but not necessarily 
a spacelike one while the transformations generated by the constraints via
Poisson brackets never leave the realm of spacelike foliations \cite{5}.
The algebraic structure of this algebroid $\mathfrak{h}$
is again universal, that is, theory 
and in particular matter content independent \cite{6} and has a purely 
geometric origin, the so-called spacelike hypersurface deformations 
which is why $\mathfrak{h}$ is denoted as hypersurface deformation 
algebroid (HDA). While $\mathfrak{h}$ is not Lie algebra in a natural way,
we may nevertheless compute all multiple Poisson brackes among its 
generators and identify algebraically independent generators 
in the set of objects so 
obtained. Their exponentiation is then a group $\mathfrak{H}$, the hypersurface 
deformation group (HDG) or Bergmann-Komar group \cite{7}. 
Very little is known about this group but again, one can 
at most expect that Diff(M) and $\mathfrak{H}$ are isomorphic at most on shell.

Thus, the classical initial value formulation of the standard model
coupled to GR automatically leads to a representation of $\mathfrak{h}$ and 
not of diff($\mathbb{R}\times \sigma$), the Lie algebra of 
Diff($\mathbb{R}\times \sigma$), given by the Poisson algebroid 
of constraints. The algebroid $\mathfrak{h}$ however contains a true 
Lie subalgebra which is isomorphic to diff($\sigma$) the Lie algebra of 
Diff($\sigma$). This Lie algebra is generated by the SDC. The SDC take 
a universal form while the HC depend on the Lagrangian that one started from. 
The SDC are much easier to quantise than the HC. The latter describe a 
complicated mixture of pure gauge (i.e. unobservable) 
transformations and physical (i.e. observable) motion \cite{8}          
while the SDC describe pure gauge transformations 
corresponding to the freedom in the choice of 
coordinates of $\sigma$. Thus a thorough understanding of the QFT of 
$\mathfrak{h}$ is an important step towards constructing a theory of 
quantum gravity.

In this paper we focus on the quantum representation theory of the SDC which 
is classically isomorphic to diff($\sigma$)$=\mathfrak{g}$. For attempts 
to rigorously represent the full quantum algebroid $\mathfrak{h}$ see 
\cite{9}. It is well known that for the case of $\sigma=S^1$ or 
$\sigma=[0,1]$ there exist Fock representations of $\mathfrak{g}$, albeit
with an anomaly. That anomaly however can however be absorbed in a central
extension of diff($\sigma$) called the Virasoro algebra \cite{10}. 
Surprisingly, already much less is known in the non-compact case or higher 
dimensions \cite{11}. By contrast, less well known are the unitary, non-anomalous 
representations of 
$\mathfrak{G}$=Diff($\sigma$) in any dimension and for both geometry and 
matter \cite{12,13}. These are however not strongly continuous with 
respect to one parameter subgroups as the representation is of 
Narnhofer-Thirring type rather than Fock type and thus does not provide 
a representation of $\mathfrak{g}$ by differentiation. The discontinuity 
of the representation is an obstacle to many standard constructions 
and leads in particular to non separable Hilbert spaces.

In this paper we thus explore the possibility of represenations of 
$\mathfrak{g}$ or strongly continuous representations of $\mathfrak{G}$
in any dimension. In particular, we wish the understand the possibility or 
not of Fock representations thereof.\\
\\
The architecture of this paper is as follows:\\
\\

In section \ref{s2} we introduce the notation and define the classical 
phase space and the universal generators of the Poisson subalgebra generated 
by the SDC. This is done for any D and $\sigma$ and for both bosonic and 
fermionic matter. The structure of the SDC assumes the form 
$C(u)=\sum_I\; C_I(u)$ where $I$ labels the field species, $u$ 
is a vector field and $\{C_I(u),C_J(v)\}=-\delta_{IJ}\;C_I([u,v])$.

In section \ref{s3} we consider general representations of the canonical 
(anti) commutation relations and associated representations of $\mathfrak{G}$
and give an overview over the various possibilities. In particular we discuss
representations based on $\mathfrak{G}$ 
quasi-invariant measures. If the measure is Gaussian, 
these represtations are isomorphic to Fock representations and the question 
boils down to asking whether Gaussian quasi-invariant measures exist.

In section \ref{s4} we show that a partial 
solution to the 
Fock representation problem consists in performing, prior to 
quantisation, a canonical transformation 
of the classical phase space which accomplishes to turn all matter fields into 
scalar or spinor densities of weight 1/2 while the geometry fields keep 
their natural tensor structure and weight. Given suitable 
matter content, by tyhe same method even a complete solution also 
for the geometry fields can be obtained.
Then all the matter $C_I(u)$ 
for the matter indices $I$ (and under the the above assumtion even 
for geometry) represent $\mathfrak{g}$ without anomaly 
in any dimension. For the geometry part, in the absence 
of suitable matter, one cannot proceed like this because 
turning the metric into a scalar using canonical transformations on 
its own phase space would yield a trivial theory. However, in this 
case for the 
geometry part one could use a Narnhofer Thirring type 
representation as mentioned 
above.

In section \ref{s5} we consider general Fock representations 
where all fields keep their natural density weight and try to 
define a representation of $\mathfrak{g}$ in it 
(possibly with an anomaly and an associated 
central extension) using a quantisation of the $C_I(u)$. 
Here we find an obstruction: For D$>1$ such 
representation does not exist while for D$=1$ it does and yields, among many 
others, a 
representation of the well known Virasoro central extension. The reason for 
why D=1 is special is due to the fact in $D=1$ the direction of a 
unit co-vector is encoded in a sign, i.e. a zero dimensional 
sphere, while in higher dimensions it requires a non-zero dimensional 
sphere. It is that sign which in $D=1$ leads to additional cancellation 
effects when computing the norm of the action of generators of the 
the SDA on Fock states. 
 
In the same section 
we show that the obstruction disappears if we fix a 
volume element $\nu$ and ask for a representation of the subalgebra of the 
$C_I(u)$ where $u$ is restricted to generate only   
volume preserving diffeomorphisms, i.e. div$_\nu(u)=0$. Then one would
have a partial solution to the problem in the form of gauge fixing part of 
the SDA. This observation leads to a natural split of the generators 
$C_I(u)$ into scaling transformations (volume non preserving) and volume 
preserving ones 
and the obstruction observed before lies entirely in the scaling part.

In section \ref{s6} we explore the possibility of using methods from 
constructive QFT (CQFT) in order to define a new representation of 
the CCR aor CAR and thereby of $\mathfrak{g}$ using a dressing 
transformation accompanied with infinite renormalisation which lead to 
a successful construction of $\Phi^4_3$ theory \cite{15}. We start 
this programme choosing a squeezed state Ansatz which implements 
some of the desired properties but not all. This part needs therefore 
a substantial extension of ideas. 

In section \ref{s7} we confine ourselves to the question whether the 
exponentiated dilation operators can at least be constructed as densely 
defined quadratic forms in Fock representations (that their generators 
do is trivial using normal ordering). The answer is in the affirmative but 
requires a non trivial central renormalisation of the dilation generator 
using the counter term method.

In section \ref{s8} we change to a new viewpoint. The difficulties 
in trying to implement the generators of the 
SDA as {\it operators} that we encountered 
in the previous sections can be circumvented if one merely implements 
them as {\it quadratic forms}. We show that in order to find solutions 
to the constraint equations it is entirely sufficient to define 
the SDA generators as quadratic forms and that, very surprisingly, while 
quadratic forms cannot be multiplied there exists a mathematically well
defined notion of the SDA in terms of quadratic forms. We then show that  
with this weakened viewpoint {\it any Fock representation} of the Weyl algebra 
supports an {\it anomaly free representation} of the SDA. This is
quite surprising in two aspects: First, since Fock representations 
are heavily background dependent, the usual intuition is that they are 
not suitable to construct representations of the SDA. We draw an analogy 
to the phenomenon of spontaneous symmetry breaking to eplain why this 
works. Secondly, within this weakened notion of representation, there is 
definitely no uniqueness result concerning the choice of the Fock 
representation which demonstrantes that insistence on having a 
diffeomorphism invariant state on the Weyl algebra is much more stringent
(none of the Fock states for which our result holds is invariant) \cite{32}.
We also make some comments about the structure of the solutions to the 
constraint equations.  

In section \ref{s9} we summarise our findings and conclude with a list 
of implications of this work and directions for future research building on
it.

In the appendix \ref{sa} we collect some tools concerning the normal 
ordering of exponentials of operators bilinear in creation and 
annihilation operators which we need in sections \ref{s6}, \ref{s8}.

\section{Classical phase space and SDC}
\label{s2}
             
This section is just to recall some well known facts from the classical 
formulation, see e.g. \cite{3,13} and references therein for derivations.
This gives us alos the opportunity to fix the notation. \\
\\
We denote spatial tensor indices by $a,b,c,..=1,..,D$ and coordinates on 
$\sigma$ by $x^a$. Tensor density fields of type $M\in \mathbb{N}_0,
N\in \mathbb{N}_0,w\in \mathbb{R}$ are denoted by 
$T^{a_1 .. a_M}\;_{b_1..b_N}$. They transform under spatial diffeomorphisms 
by the pull-back formula $T\mapsto \varphi^\ast T$ where
\be \label{2.1}
[\varphi^\ast \;T]^{a_1..a_A}_{b_1 .. b_B}(x)
=|\det(J(x))|^w\;
\prod_{k=1}^A\;[J^{-1}]^{a_k}_{c_k}(x)\; 
\prod_{l=1}^B\;J^{d_l}_{b_l}(x)\;
T^{c_1..c_A}_{d_1 .. d_B}(\varphi(x)) 
\ee
where $J^a_b(x):=\frac{\partial \varphi^a}{\partial x^b}(x)$ is the Jacobi 
matrix. 

We will require that $\mathbb{R}\times \sigma$ carries a spin structure 
\cite{2}.
A spinor field will be denoted by $\rho_A$ or $\eta^A$ where $A$
is an irreducible 
representation label for the universal covering group of SO(D) under 
consideration. One may wonder why we consider SO(D) instead of SO(1,D).
This is because in the canonical treatment the Lorentz group SO(1,D) 
becomes a local gauge group and one gauge fixes the boost part of that 
group using the time gauge \cite{13}. The residual gauge group is then 
the local rotation group SO(D).
For instance in $D=3$ the covering group is SU(2) and the 
indices $A,B,C,..=1,2$. In general the index range is $1,..,2^{[D/2]}$ where 
$[.]$ is the Gauss bracket. Spinor fields transform naturally in the 
representation $(0,0,1/2)$ with respect to Diff($\sigma$)=$\mathfrak{G}$ 
i.e. as if they were scalars of density weight 1/2 \cite{13}. This 
is because the diffeomorphism group does not possess multi-valued 
representations \cite{16} so that the flat space Lorentz transformations 
must be treated as ``internal'' local gauge transformations of Yang-Mills
type. 

In that respect the metric field $q_{ab}$ of type $(0,2,0)$ plays a 
distinguished role: The presence of spinor fields forces us to introduce 
D-Bein fields $e_a^j\; j=1,..,D$ that transform in the defining representation
of SO(D) with respect to the index $j,k,l,..=1,..,D$. It is related to 
the D-metric on $\sigma$ of Euclidian signature 
by $q_{ab}=\delta_{jk}\; e^j_a\; e^k_b$. Using the D-Bein $e$ and its 
spin connection $\Gamma$ one can construct SO(D) covariant derivatives for 
spinor fields on which the SDC and the HC depend. The basic field content 
of GR and standard matter, wrt to their transformation behaviour under 
spatial diffeomorphisms, in the Hamiltonian formulation then consist 
in the following canonical pairs:\\
1. D-Bein and conjugate momentum $(e_a^j,\; P^a_j)$.\\
2. scalar field (e.g. Higgs) and conjugate momentum $(\phi,\pi)$.\\
3. vector fields (Yang Mills potentials) and conjugate momentum 
$(A_a^\alpha,E^a_\alpha)$. Here the index $\alpha$ corresponds to the adjoint 
representation of the Lie algebra of the corresponding compact 
Yang-Mills gauge group.\\
4. spinor fields and conjugate momentum $(\rho_A,\;\eta^A)$.\\
The fields $e,\phi,A$ carry density weight zero, the fields 
$P,\pi,E$ carry density weight unity and $\psi,\xi$ both carry density weight
1/2. In principle one can extend this list arbitrarily by higher rank tensor 
fields or Rarita-Schwinger fields as they appear in supersymmetric theories.
The developments that follow can easily be extended to such fields. 

The canonical conjugacy means that we have the following non-vanishing 
Poisson (anti-)brackets \cite{17} 
\ba \label{2.2}
&& \{P^a_j(x),e_b^k(y)\}=\delta^a_b\;\delta_j^k\;\delta(x,y),\;\;
\{\pi(x),\phi(y)\}=\delta(x,y),\;\;
\nonumber\\
&& 
\{E^a_\alpha(x),A_b^\beta(y)\}=\delta^a_b\;\delta_\alpha^\beta\;\delta(x,y),\;\;
\{\eta^A(x),\rho_B(y)\}=\delta^A_B\;\delta(x,y)
\ea
and the following $\ast-$relations: The fields $e,P,\phi,\pi,A,E$ are 
self-conjugate while $(\eta^A)=i\delta^{AB}\;(\rho_B)^\ast$. 
Note that $\rho,\eta$ are 
Grassmann valued so that the last bracket is symmetric under exchange 
of entries while the others are anti-symmetric.

The respective contributions to the smeared SDC
\be \label{2.3}
C(u)=\int_\sigma\; d^Dx\; u^a\; C_a,\;C_a=C_a^M+C_a^S+C_a^V+C_a^F
\ee
where the superscript stands for ``Metric'', ``Scalar'', ``Vector boson'' and 
``Fermion'' are respectively with $(.)_{,a}:=\frac{\partial}{\partial x^a} (.)$ 
are      
\ba \label{2.4}
C_a^M &=& P^b_j\;e^j_{b,a}-(P^b_j e^j_a)_{,b}
\nonumber\\
C_a^S &=& \pi\; \phi_{,a}
\nonumber\\
C_a^M &=& E^b_\alpha\;A^\alpha_{b,a}-(E^b_\alpha A^\alpha_a)_{,b}
\nonumber\\
C_a^F &=& -\frac{1}{2}[\eta^A\; \rho_{A,a}-\eta^A_{,a}\; \rho_A]
\ea
Comparing the first and third line we see that the gravity contribution 
is the same as for the gauge boson contribution for the gauge group SO(D) 
(local frame rotations). Moreover, in view of (\ref{2.2})
the contributions for different values of
$j,\alpha,A$ decouple so that as far as representations of $\mathfrak{G}$ are 
concerned we can focus on a single vector pair $(p^a,q_a)$ and a single 
Grassman pair $\rho,\eta$.

It is not difficult to check that with 
\be \label{2.5}
C^V_a=p^b\;q_{b,a}-(p^b q_a)_{,b}\;
C^G_a=-\frac{1}{2}[\eta\; \rho_{,a}-\eta_{,a}\; \rho]
\ee
with $\eta=i\rho^\ast$ 
and using (\ref{2.2}) we have the following non-vanishing brackets 
\be \label{2.6}
\{C^S(u),C^S(v)\}=-C^S([u,v]),\;\;
\{C^V(u),C^V(v)\}=-C^V([u,v]),\;\;
\{C^G(u),C^G(v)\}=-C^G([u,v]),\;\;
\ee
where 
\be \label{2.7}
[u,v]^a=u^b\; v^a_{,b}-v^b\; u^a_{,b}
\ee
In case that $\sigma$ has a boundary, we assume that suitable boundary 
conditions are imposed on all fields and the vector fields $u,v$ so that 
(\ref{2.6}) holds (e.g. $u,v$ could be of compact support). 

The relations (\ref{2.6}) establish that each contribution by itself is 
a Poisson bracket representation of the SDA. In particular, the SDC generate 
the following canonical transformation on the fields via the Poisson bracket
\ba \label{2.8}
&& \{C(u),\phi\}=u^a\;\;\phi_{,a},\;  
\{C(u),\pi\}=[u^a\;\;\pi]_{,a},\;
\{C(u),q_a\}=u^b q_{a,b}+u^b_{,a} \;q_b,\;
\nonumber\\
&&\{C(u),p^a\}=[u^b p^a]_{,b}-u^a_{,b} \;p^b,\;
\{C(u),\rho\}=u^a\; \rho_{,a}+\frac{1}{2}\; u^a_{,a}\;\rho,\;  
\{C(u),\eta\}=u^a\; \eta_{,a}+\frac{1}{2}\; u^a_{,a}\;\eta,\;  
\ea
which shows that these canonical transformations exactly match the 
geometrical tensor and weight structure of the fields mentioned above.

\section{Representations of the CCR, CAR, $\mathfrak{g}$ and $\mathfrak{G}$}
\label{s3}

We want to represent the CCR or CAR following from (\ref{2.2}) and the $\ast$
relations on a Hilbert space $\cal H$. Thus we ask that the fields be 
promoted to operator valued distributions satisfying the following 
non-vanishing canonical (anti-)commutation relations  
\be \label{3.1}
[\pi(x),\phi(y)]=i\delta(x,y),\;
[p^a(x),q_b(y)]=i\delta^a_b\;\delta(x,y),\;
[\eta(x),\rho(y)]_+=i\delta(x,y)\;\;\Leftrightarrow\;\;
[\rho^\ast(x),\rho(y)]_+=\delta(x,y)
\ee
Here $[A,B]=AB-BA,\;[A,B]_+=AB+BA$ denote commutator and anti-commutator 
respectively. Moreover we require the $\ast-$relations 
\be \label{3.2}
\phi^\ast-\phi=
\pi^\ast-\pi
=q_a^\ast-q_a=[p^a]^\ast-p^a=0
\ee
There are no $\ast$ relations on $\rho$ because the 
momentum $\eta=i\rho^\ast$ 
of $\rho$ is proportional to its conjugate. We can turn the operator valued 
distributions into operators by smearing them with test functions 
$f,g,F,G,\zeta$
\ba \label{3.3}
\phi(f) &=& \int_\sigma\; d^D x\; f(x)\;\phi(x),\; 
\pi(g)=\int_\sigma\; d^D x\; g(x)\;\pi(x),\; 
q(F)=\int_\sigma\; d^D x\; f^a(x)\;q_a(x),\; 
\nonumber\\
p(G) &=& \int_\sigma\; d^D x\; g_a(x)\;p^a(x),\; 
\rho(\zeta)=\int_\sigma\; d^D x\; \zeta(x)^\ast\;rho(x),\; 
\ea
where $f,g,F,G$ are real valued while $\zeta$ is complex valued. It follows 
for any $\psi\in {\cal H}$
\be \label{3.4}
<\psi,\;[\rho(\zeta),[\rho(\zeta)]^\ast]_+\;\psi>_{{\cal H}}
=||[\rho(\zeta)]^\ast\; \psi||_{{\cal H}}^2+
||\rho(\zeta)\; \psi||_{{\cal H}}^2=||\zeta||_{L_2}^2 \; 
||\psi||_{{\cal H}}^2
\ee
where
$||\zeta||_{L_2}:=\int_\sigma\; d^Dx\; |\zeta(x)|^2$.
This means that the fermionic 
operators $\rho(\zeta),\rho(\zeta)^\ast$ are automatically 
bounded in any representation.

The remaining bosonic operators are not automatically bounded. 
We pass to the corresponding Weyl elements 
\be \label{3.5}
W(f,g)=\exp(i[\phi(f)+\pi(g)]),\; 
W(F,G)=\exp(i[q(F)+p(G)])
\ee
which satisfy the usual Weyl relations 
\be \label{3.6}
W(f,g)\; W(f',g')=W(f+f',g+g')\;\exp(-\frac{i}{2}[
<g,f'>_{L_2}-<g',f>_{L_2}]),\;\;
W(f,g)^\ast=W(-f,-g)
\ee
and similar for $W(F,G)$. Thus the Weyl elements are automatically unitary 
in any representation and thus bounded.

Thus we look for $\ast-$representations $r$ of the above abstract algebra 
$\mathfrak{A}$ 
as bounded operators on a Hilbert space $\cal H$. Without loss of generality
we can consider a cyclic representation (i.e. $r(\mathfrak{A})\Omega$ is dense 
in ${\cal H}$) based on a cyclic vector 
$\Omega\in {\cal H}$ as every (non-degenerate i.e. 
$r(\mathfrak{A}){\cal H}$ is dense in $\cal H$ which is automatic for 
unital $\mathfrak{A}$ for which $r(1)={\rm id}_{{\cal H}}$) 
representation may be decomposed into cyclic ones \cite{18} and the CCR and 
CAR algebras have no degenerate representations because they are unital 
(i.e. contain a unit element which is represented as the identity operator).
Cyclic representations are in 1-1 correspondence with states $\omega$ 
(positive, linear, normalised functionals) on 
$\mathfrak{A}$ via the GNS construction, i.e. 
$<\Omega,r(a)\Omega>_{{\cal H}}=\omega(a)$. 
If $\cal H$ is separable or has
an ONB labelled by $\mathbb{R}$ then there exists an Abelian 
sub$^\ast-$algebra $\mathfrak{B}$ of the $C^\ast-$algebra 
${\cal B}({\cal H})$ of bounded 
operators on $\cal H$ for which $\Omega$ is still cyclic \cite{19}.
If the GNS represengtation is surjective as a map 
$r:\;\mathfrak{A}\to {\cal B}({\cal H})$ (a $C^\ast-$algebra is uniformly
closed but $r(\mathfrak{A})$ could be a proper $C^\ast-$subalgebra
of ${\cal B}({\cal H})$) then 
for each $B\in \mathfrak{B}$ there exists $a_B\in \mathfrak{A}$
such that $B=r(a_B)$. Then $\mathfrak{A}_{\mathfrak{B}}=\{a_B:\;
B\in  \mathfrak{B}\}$ is an Abelian $C^\ast-$subalgebra of 
$\mathfrak{A}$ if $r$ is faithful (i.e. $r(a)=0$ implies $a=0$):  
Since $[B,B']=r([a_B,a_{B'}])=0$ and $B\mapsto a_B$ is a 
$^\ast-$representation 
of $\mathfrak{B}$ in $\mathfrak{A}$ as 
$B^\dagger=r(a_{B^\dagger})=r(a_B)^\dagger=r(a_{B}^\ast)$ and 
$B B'=r(a_{BB'})=r(a_B a_{B'})$. Then $\mathfrak{A}_{\mathfrak{B}}$ is 
generated from a set of Abelian Weyl elements $W(f,g_f)$ i.e. we have 
$S((f,g_f),(f',g'_{f'})):=<g_f,f'>-<g'_{f'},f>=0$. 
As this is the symplectic 
structure of the classical Poisson algebra it follows that 
the corresponding Poisson subalgebra corresponds to choosing $g_f=0$ 
up to a canonical transformation. Thus the Abelian Weyl subalgebra 
generators are represented by $U\;r(W(f,0))\;U^{-1}$ 
if the canonical transformation is implementable as unitary $U$.
Since the GNS data are unique only up to a unitary trsnaformation we set 
$U=1$ and assume that $\mathfrak{B}$ is generated by the $r(W(f,0))$ and 
by the $r(W(F,0))$ by the same argument.

We can then think of $\mathfrak{B}$ as $C^0(\Delta)$, the continuous functions 
on the Gel'fand spectrum $\Delta$ of $\mathfrak{B}$, i.e. the characters 
$\chi:\;\mathfrak{B}\to \mathbb{C}$ satisfying 
$\chi(BB')=\chi(B)\chi(B'),\chi(B+B')=\chi(B)+\chi(B'),\;[\chi(B)]^\ast=
\chi(B^\ast)$ \cite{18}. Then by the Riesz Markov theorem \cite{20}
we can think 
of $\omega$ restricted to $\mathfrak{B}$ as a regular Borel probability
measure on 
$\Delta$ via 
$<\Omega,B\Omega>=\mu(\check{B})$ where $\check(B)(\chi)=\chi(B)$ is the 
Gel'fand transform \cite{18}. 

The upshot of these considerations is that we may think of ${\cal H}$ as
$L_2(\Delta,d\mu)$ as far as the bosonic degrees of freedom are concerned.
The total Hilbert space is a tensor product 
${\cal H}={\cal H}_B\otimes {\cal H}_F$ where the factors correspond to 
the bosonic and fermionic degrees of freedom respectively. Non-trivial 
examples of such representations exist, e.g. Fock representations.

Among the many possible representations of the CCR and CAR we now 
wish to select representations which are also representations of 
$\mathfrak{g}$ or $\mathfrak{G}$. Surprisingly, this easily
accomplished in Fock representations of the fermions but rather non-trivial
in Fock representations of the bosons. The geometric origin of this difference
is as follows:\\
A Fock representation for a bosonic field such as the scalar field $\phi$
is defined by an annihilation operator
\be \label{3.7}
A:=2^{-1/2}[\kappa\cdot \phi-i\kappa^{-1}\pi]
\ee
where $\kappa$ is an invertible, positive 
operator on $L_2(\sigma,d^Dx)$. However,
due to the fact that $\phi$ is a scalar of density weight zero while 
$\pi$ is a scalar of density weight one, it follows that (\ref{3.8}) does 
not get mapped to an annihilation operator under a diffeomorphism but rather 
becomes a linear combination of annihilation and creation operators. 
This means that the cyclic vector $\Omega$ cannot be invariant under 
diffeomorphisms. In particular it is itself not annihilated by a potential
generator of spatial diffeomorphisms, even when normal ordered. This 
means that the expression for $C^S(u)$, which will be a homogeneous 
quadratic polynomial in $A,A^\ast$, contains an $(A^\ast)^2$ contribution.
That $(A^\ast)^2$ contribution typically maps $\Omega$ out of the Fock space 
depending on the choice of $\kappa$. The situation is similar for 
$C^V(u)$.

By contrast, the natural Fock vacuum 
for the fermions is annihilated by $\rho$ which itself is the annihilator
$A$. 
As $\rho$ is a scalar density of weight 1/2, $\rho$ is mapped to another 
annihilator under diffeomorphisms and therefore the Fock vacuum is invariant
under diffeomophisms. Indeed, the normal ordered expression for $C^F(u)$ 
does not contain an $(A^\ast)^2$ term as we will see.

Thus, the source of the trouble to represent $C^S(u),C^V(u)$ in Fock 
representations is due to the fact that for bosons the configuration 
and momentum degrees of freedom have different density weight and different 
(namely dual) tensorial structure while for fermions the density weight 
and tensorial structure 
for configuration and momentum degrees of freedom is the same. 
This observation directs towards a possible at least partial 
solution: Perform a canonical 
transformation on the bosonic phase space such that at least the matter 
fields and their conjugate momenta become scalar densities of weight 1/2
so that it makes geometrical sense to form complex linear combinations that 
qualify as annihilation operators that are mapped to themselves under 
diffeomorphisms. We will show that this is indeed possible but that the 
problem persists in the gravitational sector.

The disadvantage of this canonical transformation is that it has a non-trivial 
effect on the gravitational momentum and thus renders the already complicated
HC even more complicated. Thus one may wonder whether by a judicious choice 
of $\kappa$ above one can avoid that canonical transformation and in addition 
achieve a Fock representation also in the gravitational sector. The surprising
answer is that this is indeed possible in D=1 as was already known but that 
it is not possible for general vector fields $u$ in higher dimensions.
One can find Fock representations if one makes an additional restriction 
on the allowed choice of vector fields which corresponds to a partial 
gauge fixing of the SDG. That restriction comes with a background structure, 
namely a volume element. Let $g$ be a fixed Euclidian D-metric of Euclidian 
signature. Define the divergence of the vector field $u$ with respect to 
$g$ by 
\be \label{3.8}
{\sf  div}_\nu(u):=[\det(g)]^{-1/2}\;\partial_a [\det(g)]^{1/2}\; u^a]   
\ee
which actually does not depend on the full metric $g$ but only on its volume 
element $\nu:=[\det(g)]^{1/2}$. A vector field is calles $\nu-$divergence free 
iff (\ref{3.8}) vanishes. These vector fields form a subalgebra of the 
SDA. To see this, given $\nu$ pick any metric $g$ such that 
$\nu=[\det(g)]^{1/2}$. Then ${\sf div}_\nu(u)=\nabla_a u^a$ where $\nabla$ 
is the unique torsion free and $g$ compatible covariant differential. 
It follows for $\nu-$divergence free vector fields $u,v$
\ba \label{3.9}   
{\sf div}_\nu([u,v]) &=&
\nabla_a[\nabla_u \; v^a-\nabla_v \; u^a]
\nonumber\\
&=& [\nabla_a u^b]\;[\nabla_b v^a]+u^b\;\nabla_a\nabla_b\; v^a
-u\leftrightarrow v
\nonumber\\
&=& u^b\;[\nabla_a,\nabla_b]\; v^a -u\leftrightarrow v
\nonumber\\
&=& u^b\;R_{bc}\; v^c -u\leftrightarrow v=0
\ea
where $R_{ab}$ is the symmetric Ricci tensor of $g$.

If one does not make that restriction we are forced to step outside the realm 
of Fock representations for D$>1$. A possibility to do this arises from the 
measure theoretic framework sketched above. Let $\mu$ be a 
probability measure on a space 
$X$ of tensor fields $\phi$ of a certain type and suppose that the 
natural action $\varphi^\ast$ of $\mathfrak{G}$ on elements $\phi\in X$ 
reviewed in section \ref{s2}
is measurable, i.e. preserves the $\sigma-$algebra of measurable sets in $X$
upon which $\mu$ is based (that is, if $B$ is measurable, then is 
$\varphi^\ast(B)=\{\varphi^\ast(\phi);\; \phi \in B\}$ for all $\varphi\in 
\mathfrak{G}$). 
Then we may try to define a unitary representation of $\mathfrak{G}$ on 
${\cal H}=L_2(X,d\mu)$ by 
\be \label{3.10} 
[U(\varphi)\psi](\phi):=m_\varphi(\phi)\;\psi(\varphi^\ast \pi)
\ee 
for a certain functional $m_\varphi:\; X\to \mathbb{C}$ to be determined
which may depend parametrically on $\varphi$ but not on the vector $\psi$.
Suppose that the measure $\mu$ is quasi-invariant for the action 
$\varphi^\ast$, that is, the measure $\mu$ and the measure 
$\mu\circ \varphi^\ast$ have the same sets of measure zero. Then by the 
Radon Nikodym theorem \cite{20} there exists a non-negative $L_1(X,d\mu)$ 
function such that $\mu-$a.e.
\be \label{3.11}
d\mu(\varphi^\ast\phi)=f_\varphi(\phi)\; d\mu(\phi)
\ee
It follows that for the choice 
\be \label{3.12}
m_\varphi(\phi):=\sqrt{f_\varphi(\phi)}
\ee
we achieve unitarity
\be \label{3.12a}
<U(\varphi)\psi,U(\varphi)\tilde{\psi}>
=\int_X\; d\mu(\phi)\;|m_\varphi(\phi)|^2\; 
[\psi^\ast\;\tilde{\psi}](\varphi^\ast \phi)            
=\int_X\; d\mu(\varphi^\ast\phi)\;
[\psi^\ast\;\tilde{\psi}](\varphi^\ast \phi)            
=<\psi,\tilde{\psi}>
\ee
Moreover, we have using the identity $[\varphi_1\circ \varphi_2]^\ast=
\varphi_2^\ast\circ\varphi_1^\ast$
\be \label{3.13}
f_{\varphi_1\circ\varphi_2}(\phi)
=\frac{d\mu(\varphi_2^\ast\circ\varphi_1^\ast\phi)}{d\mu(\phi)}
=\frac{d\mu(\varphi_2^\ast\circ\varphi_1^\ast\phi)}{d\mu(\varphi_1^\ast\phi)}
\frac{d\mu(\varphi_1^\ast\phi)}{d\mu(\phi)}
=f_{\varphi_2}(\varphi_1^\ast\phi)\;f_{\varphi_1}(\phi)
\ee
so that
\ba \label{3.14}
&&[U(\varphi_1\circ\varphi_2)\psi](\phi)
=
m_{\varphi_1\circ\varphi_2}(\phi)\;\psi([\varphi_1\circ\varphi_2]^\ast \phi)
=m_{\varphi_2}(\varphi_1^\ast\phi)\;
m_{\varphi_1}(\phi)\;
\psi(\varphi_2^\ast\circ\varphi_1^\ast\phi)
\nonumber\\
&=& m_{\varphi_1}(\phi)\; [U(\varphi_2)\psi](\varphi_1^\ast\phi)
=[U(\varphi_1)\;U(\varphi_2)\psi](\phi)
\ea

Thus we obtain a non-anomalous representation of $\mathfrak{G}$ if 
quasi-invariant measures exist. That this is non-trivial is due to the 
fact that formally the Radon-Nikodym derivative is the determinant of the 
Jacobi Matrix when one changes variables $\phi$ to variables $\varphi^\ast \phi$
and such determinants are typically ill-defined when an infinite number 
of degrees of freedom is involved. An exception is provided again in the 
case that $\phi$ is a scalar of density weight 1/2. In this case an example 
for a quasi-invariant (in fact precisely invariant) measure is the Gaussian
measure with white noise covariance. This measure is defined via its 
Fourier transform 
\be \label{3.15}
\hat{\mu}(f):=\int_X\;d\mu(\phi)\;\exp(i\phi(f))
:=\exp(-\frac{1}{2}<f,f>_{L_2})
\ee 
and the span of the vectors $\psi(\phi):=w_f(\phi):=e^{i\phi(f)}$ is dense in 
$L_2(X,d\mu)$. We have the identity 
\be \label{3.16}
\varphi^\ast\phi(f)
=\int_\sigma\; d^Dx\; f(x)
|\det(\partial\varphi)/\partial x)|^{1/2}\;\phi(\varphi(x))
=\int_\sigma\; d^Dy\; f(\varphi^{-1}(y))
|\det(\partial\varphi^{-1})/\partial y)|^{1/2}\;\phi(y)
=\phi([\varphi^{-1}]^\ast f)
\ee
where we assign $f$ to be a scalar density of weight 1/2 as well. 
Since 
\be \label{3.17}
<w_f,w_{\tilde{f}}>=\hat{\mu}(\tilde{f}-f)
\ee
we find with 
\be \label{3.18}
U(\varphi)w_f:=w_{[\varphi^{-1}]^\ast\;f}
\ee
that unitarity holds
\ba \label{3.18a}
&& <U(\varphi)\;w_f, U(\varphi)\;w_{\tilde{f}}>_{L_2(X)}=
\hat{\mu}([\varphi^{-1}]^\ast(\tilde{f}-f)]>
=
\exp(-\frac{1}{2}||[\varphi^{-1}]^\ast(\tilde{f}-f)||^2_{L_2(\sigma)}) 
=\exp(-\frac{1}{2}||\tilde{f}-f||^2_{L_2(\sigma)}) 
\nonumber\\
& = &<w_f, w_{\tilde{f}}>_{L_2(X)}
\ea
and $U(\varphi_1)\;U(\varphi_2)=U(\varphi_1\circ \varphi_2)$ thanks to 
$([\varphi_1\circ\varphi_2]^{-1})^\ast= 
([\varphi_1]^{-1})\ast \circ ([\varphi_2]^{-1})^\ast$. The Radon-Nikodym 
derivative in this case equals unity. It appears that otherwise little 
is known about the existence of such quasi-invariant measures, see
\cite{21} for an overview.

Note that our discussion is in agreement with the obstructions pointed out in 
\cite{21a}. In general, given a state $\omega$ on a $C^\ast-$algebra on which 
a group $G$ acts by $^\ast-$automorphisms $\alpha_g$, that is, 
$\alpha_g\circ \alpha_{g'}=\alpha_{gg'},\;\alpha_{1_G}=1_{\mathfrak{A}}$ 
and $\alpha_g(a)^\ast=\alpha_g(a^\ast)$ (besides being compatible with
algebraic operations) such that $\omega$ is invariant 
$\omega=\omega\circ\alpha_g$ for all $g\in G$ this induces a unitary 
representation of $G$ in the GNS representation by 
$U(g)r(a)\Omega=r(\alpha(g))\Omega$. Applied to the SDG we thus ask for 
a SD invariant state $\omega$, i.e. since the state $\omega$ is completely 
fixed if we know the expectation value functional 
$(f,g)\mapsto \chi(f,g):=\omega(W(f,g))$ we require that 
$\chi(f,g)=\chi(\varphi^\ast f,\varphi^\ast g)$ where $\varphi^\ast$ denotes 
the corresponding dual pull-back representation on the smearing 
fields and $\chi$ must define a positive linear functional. If $\omega$ 
is continuous in both $f,g$ (regular state) 
this invariance property transcends to the 
all $N-$point functions, e.g. $\omega(q(f_1)..q(f_N))$ and   
$\omega(p(g_1)..p(g_N))$ which in turn define distributions in $N$ 
points. However, if $f$ is a non-trivial tensor of type $(A,B,w)$ with 
$A\not=B$ or $w\not=1/2$ there are no invariant such distributions except 
the trivial one.
Thus $\omega(q(f_1)..q(f_N))=0$. Thus by $q(f)^\ast=q(f)$ we have
$||r(q(f_1)..q(f_N))\Omega||^2=\omega(q(f_N)..q(f_1)q(f_1)..q(f_N))=0$
i.e. $r(q(f_1)..q(f_N))\Omega=0$ and similarly
$r(p(g_1)..p(g_N))\Omega=0$. Thus we get the contradiction 
$0=[r(p(g)),r(q(f)]\Omega=i<g,f>\Omega$. This obstruction can be evaded
e.g. in the following cases:\\
1.\\ 
the state is irregular in either $f$ or $g$ or both, 
this is the option followed in LQG.\\
2.\\
one achieves $A=B,w=1/2$, this is what we will discuss in the next 
section\\
3.\\ 
One does not insist on an invariant state but merely asks for 
quasi-invariance \cite{21b}, that is, for each $g\in \mathfrak{G}$ we 
find a generalised Radon-Nikodym derivative $b_g\in \mathfrak{A}$ such 
$\omega(\alpha_g(a))=\omega(b_g a)$ which generalises the above 
measure theoretic definition. The above N-point distributions then need not 
be SD invariant.\\
4.\\ 
One allows for (central) extensions of the SDA. This again 
means that the corresponding state is not SD invariant.
Now additional complications 
arise, namely in solving the constraints, the anomalous term in the 
constraint algebra prevents the existence of generalised zero 
eigenvectors. Suppose that the anomaly vanishes when the vector 
fields $u$ are restricted to a subspace $U_1^\mathbb{C}$ 
with $[U_1^\mathbb{C},U_1^\mathbb{C}]\subset U_1^\mathbb{C}$
of the 
complexification $U^\mathbb{C}$ of the space of real vector fields 
$U$ such that $U_1^\mathbb{C}\cup U_2^\mathbb{C}=U^\mathbb{C}$ where 
$U_2^\mathbb{C}$ is the space of complex conjugate 
fields $u_1^\ast,\; u_1\in U_1^\mathbb{C}$. Let $\psi_1,\psi-1'$ 
be joint generalised 
zero eigenvectors of the $C(u_1)$, i.e. 
$<\psi_1, C(u_1) \psi>=0$, $<\psi_1', C(u_1) \psi>=0$, 
for all
$u_1\in U_1^\mathbb{C},\;\psi\in D$ where $D$ is a common dense domain for 
all the $C(u_1)$. Then for $U_2^\mathbb{C}\ni u_2=u_1^\ast,\; u_1\in 
U_1^\mathbb{C}$ we formally have $C(u_2)=C(u_1)^\ast$ as $C(u)$ is symmetric
for real valued $u$ and thus 
$<\psi_1,C(u_2)\psi_1'>=<C(u_1)\;\psi_1,\psi_1'>=<\psi_1',C(u_1)\psi_1>^\ast 
=0$ because $\psi_1$ is a (non-normalisable) linear combination of elements 
$\psi\in D$. In that sense one has solved all $C(u),\; u=u_1+u_1^\ast\in U$
in the sense of matrix elements between generalised 
eigenvectors $\psi_1$. The space
of solutions $\psi_1$ ist then to be equipped with a new inner product, 
often of the form $<\psi_1,\psi_1'>_1=<\psi_1,\psi_1'>/<\psi_1^0,\psi_1^0>$
where $\psi_1^0$ is a chosen reference solution. \\
5.\\
The above obstruction argument relies on phase spaces which are 
cotangent spaces over spaces tensor fields which themselves are vector 
spaces over each point $x\in \sigma$. 
One may however allow more general phase spaces such as cotangent 
spaces over groups or cones over each point $x\in \sigma$. This 
can be argued to be the case for the gravitational 
D-metrics which are classically 
of Eucldian signature \cite{21c} and thus do not form a vector space
(although one may argue that deviations from Euclidian signature 
may be allowed in the quantum theory, i.e. ``tunneling'' into the 
classically forbidden regime is allowed).\\
6.\\ 
One gives up the requirement that $C(u)$ be a densely defined 
{\it operator} altogether and contents oneself by asking that it 
is at least a densely defined {\it quadratic form}. Then while 
$C(u)C(v)$ is ill defined there is a chance to define $[C(u),C(v)]$ as 
a quadratic form using standard regularisation techniques exploiting 
the fact that due to the minus sign involved certain divergences can 
cancel out. This strategy is perhaps the minimal physical requirement 
because when solving the constraints $C(u)=0$ in the quantum theory we 
ask for generalised zero eigenvectors $\psi$ that satisfy 
$<\psi,C(u) \psi'>=0$ for all $\psi'$ in the dense quadratic form 
domain $\cal D$. 
Thus we are actually
only using the constraint as a quadratic form in the search for suitable 
$\psi$ which are not normalisable linear combinations of elements in $\cal D$,
hence all coefficients in that linear combination are in fact well defined.
We will see that this is in fact possible.\\
\\
In the search for quasi-invariant states or measures resulting in 
well defined constraint operators one    
may try to approach the problem using methods of 
constructive QFT (CQFT). First of all we can e.g. split 
$C^S(u)=C^S_D(u)+C^S_H(u)$ into two 
pieces where $C^S_D(u)$ generates dilatations and $C^S_H(u)$ is the 
``half-density piece'' of $C^S(u)$. The second piece presents no problem 
in any Fock representation and thus we may focus on the dilatation piece 
$C^S_D(u)$. The idea, following \cite{15}, is now to introduce both IR and 
UV cut-offs. E.g. we work on a large torus $\sigma=T^D$ of radius $R$
so that momentum modes are discrete labelled by an integer $m\in\mathbb{Z}$
and we use a mode cut-off $M$. 
Then  
we expand the exact $C^S_D(u)$ into those modes and drop all terms involving 
modes $|m|>M$. 
Suppose now that we find a 
an invertible ``dressing operator'' $T_M$ on ${\cal H}$ acting 
only on the modes $|m|\le M$ such 
that $\hat{C}^S_D(u)_M:=T_M^{-1}\; C^S_D(u)_M\;T_M$ does not 
contain any $[A^\ast]^2$ terms. We define 
${\cal H}_M:=T_M\;{\cal H}$ and for any $T_M \psi,T_M \tilde{\psi}\in 
{\cal H}_M$ the ``renormalised'' inner product
\be \label{3.19}
<T_M \psi,T_M \tilde{\psi}>_M
:=\frac{<T_M\psi,\;T_M\tilde{\psi}>}{<T_M\Omega,T_M\Omega>}    
\ee
The operators $\hat{C}^S_D(u)_M$ are now well defined on $\cal H$ with 
range therein. 
Equivalently $C^S_D(u)_M$ is well defined on ${\cal H}_M$ with range in 
${\cal H}_M$ as 
\be \label{3.20}
C^S_D(u)_M\; T_M\;\psi=T_M\;\hat{C}^S_D(u)_M\psi
\ee
The idea is now to try to take the limit $M\to \infty$ in order to obtain
a representation of $C^S_D(u)$ on ${\cal H}_\infty$. Whether this works 
depends on the convergence of (\ref{3.19}), in particular whether the rate 
of divergence of the vectors $T_M\psi$ is the same as that of $T_M\Omega$.
In $\Phi^4$ theory in $D=3$ this works for the Hamiltonian operator. 
In our case we will show that this fails, among other things because 
the dressing operator acquires a non-trivial dependence on $u$ which 
it must not in order that (\ref{3.19}) is well defined.

We may therefore try to obtain a more modest result and consider the 
objects $\exp(i C^S_D(u))$ on a chosen Fock space. Even if $C^S_D(u)$
is normal ordered with respect to the chosen annihilation/creation operator 
structure, this object is ill defined as it stands and thus does not 
even qualify as a quadratic form on that Fock space. We therefore compute 
the normal ordered object $:\exp(i C^S_D(u)):$ which is surprisingly 
complicated. We find that $:\exp(i C^S_D(u)):=\lim_{M\to\infty}
\exp(i C^S_D(u)_M)\; e^{\gamma_M(u)}$ where $\gamma_M(u)$ is a number that 
diverges as we remove the UV regulator $M$. It is still 
remarkable that one compute this in closed form as one proceeds just as 
in ordinary QFT perturbation theory: We formally Taylor expand 
$\exp(i C^S_D(u))$ and normal order the powers of $C^S(u)$ that one obtains 
this way. The correction terms are collected into a counter term. One would 
not have guessed that one can do this to all orders and that the counter 
terms organise themselves into a simple exponential. The normal ordered 
quadratic forms 
$:\exp(i C^S_D(u)):$ are now well defined on the Fock space and while 
their products are therefore generically 
ill defined one may ask whether in the sense 
of quadratic forms 
\be \label{3.21}
:\exp(i C^S_D(u)):\;
:\exp(i C^S_D(v)):\;
:\exp(i C^S_D(-u)):\;=\;
:\exp(i C^S_D(v)):
\ee
which serves as a renormalised, exponentiated substitute of the classical 
identity $\{C^S_D(u),C^S_D(v)\}=0$.
Alterantively, at the level of the algebra ony may ask whether in the sense 
of quadratic forms $[C^S_D(u),C^S_D(v)]=0$. In both cases one works at 
finite cut-off to compute the commutators and then takes the limit. We will 
see that the answer is in the affirmative.

\section{Half-density transformations and Fock representations}
\label{s4}

In this section we show that Fock representations of the SDG or SDA are 
easy to construct if the underlying fields are scalar half-densities 
of weight 1/2 and if one uses the white noise covariance of the corresponding
Gaussian measure in the bosonic case. Fermions naturally come with density 
weight 1/2 and are already expressed as annihilation and creation operators
therefore we start this section with fermions. For bosonic scalar fields
one can perform an intrinsic canonical transformation from the density zero 
and one fields $\phi,\pi$ to halt density fields $\hat{\phi},\hat{\pi}$.
For higher rank bosonic tensor fields an intrinsic canonical transformation
is no longer possible, but if a D-Bein field is present one can find
a canonical transformation such that all matter fields and their 
canonical momenta are scalar densities of weight 1/2.

All representations constructed in this section are background independent
and are very different from the background independent 
representation used in LQG. One may wonder how that can be true in view 
of the uniqueness result \cite{32}. The answer is quite simple: In 
\cite{32} one uses completely different Weyl algebras. The fields 
are not scalar densities of weight $1/2$ but keep their natural tensorial 
and weight character. Furthermore, in order to solve the non-Abelian 
Gauss constraint one uses less than $D$ smearing dimensions of the fields 
(here we use $D$ smearing dimensions). Therefore the representations 
considered in \cite{32} and the present section are for different algebras 
and therefore cannot be compared, there is no contradiction.

\subsection{Fermions}
\label{s4.1}

We start this subsection by defining the natural Fock representation of 
the CAR for fermions using $\rho(x)\Omega=0$. The fermionic contribution 
to the SDC reads in normal ordered form 
\be \label{4.1} 
C^G(u)=-\frac{i}{2}\int_\sigma\; d^Dx\; u^a(x)\; 
[\rho^\ast(x)\;\rho_{,a}(x) -\rho^\ast_{,a}(x)\;\rho(x)] 
\ee 
This 
operator is well defined on the corresponding Fock space because it does 
not contain a term of the form $\rho^\ast(x)\rho^\ast_{,a}(x)$, in 
particular it annihilates the vacuum. It is useful to integrate by parts 
and write it in the alternative form 
\be \label{4.2} 
C^G(u)=-i\int_\sigma\; d^Dx\;[ u^a(x)\; 
[\rho^\ast(x)\;\rho_{,a}(x)+\frac{1}{2}\; u^a_{,a}(x)\; 
\rho^\ast(x)\;\rho(x)] 
\ee 
where we have assumed that $u^a$ has compact 
support if $\sigma$ has a boundary. Using the abbreviations 
\ba \label{4.3} 
&& B_{ax}=\rho^\ast(x)\rho_{,a}(x),\; 
B_x=\rho^\ast(x)\rho_{,a}(x),\; 
\rho_{ax}=\rho_{,a}(x),\;\rho_x=\rho(x),\;
u^a_x=u^a(x),\;v^a_x=v^a(x),\;
\\
&& \mu_x=\frac{1}{2}\;u^a_{,a}(x),\; 
\nu_x=\frac{1}{2}\;v^a_{,a}(x),\;
\delta=\delta(x,y),
\;\delta_{ax}=\delta_{,x^a}(x,y)
\;\delta_{by}=\delta_{,y^b}(x,y)
\nonumber
\ea 
and the well known identity among commutators $[.,.]$ and anti-commutators
$[.]_+$ given by 
$[AB,CD]=[A,CD]B+A[B,CD]=([A,C]_+ D-C[A,D]_+)B+A([B,C]_+ D -C[B,D]_+)$ 
as well as the CAR we find 
\ba \label{4.4}
&& -[C^G(u),C^G(v)]
=\int \;d^D x\;d^D y\;
\{
u^a_x\; v^b_y [B_{ax},B_{by}]+
u^a_x\; \nu_y [B_{ax},B_y]+
\mu_x\; v^b_y [B_x,B_{by}]+
\mu_x\; \nu_y [B_x,B_y]
\}
\nonumber\\
&=&\int \;d^D x\;d^D y\;
\{
u^a_x\; v^b_y [B_{ax},B_{by}]+
(u^a_x\; \nu_y-v^a_x \mu_y)  [B_{ax},B_y]+
\frac{1}{2}\;(\mu_x\; \nu_y-\mu_y \nu_x) [B_x,B_y]
\}
\nonumber\\
&=&\int \;d^D x\;d^D y\;
\{
u^a_x\; v^b_y (\rho^\ast_x\; [\rho_{ax},\rho^\ast_y]_+ \;\rho_{by}
-\rho^\ast_y\;[\rho^\ast_x,\rho_{by} ]_+ \;\rho_{ax})
+(u^a_x\; \nu_y-v^a_x \mu_y)\;
(\rho^\ast_x\; [\rho_{ax},\rho^\ast_y]_+ \;\rho_y
-\rho^\ast_y\; [\rho^\ast_x,\rho_y]_+ \;\rho_{ax})
\}
\nonumber\\
&=&\int \;d^D x\;d^D y\;
\{
(u^a_x\; v^b_y-u^b_y v^a_x) 
\rho^\ast_x\; \delta_{ax} \;\rho_{by}
+(u^a_x\; \nu_y-v^a_x \mu_y)\;
(\rho^\ast_x\; \delta_{ax} \;\rho_y
-\rho^\ast_y\; \delta \;\rho_{ax})
\}
\nonumber\\
&=&\int \;d^D x\;
\{-((u^a \rho^\ast)_{,a}\; v^b - u^b (v^a \rho^\ast)_{,a})\;\rho_{,b}
-((u^a \rho^\ast)_{,a}\; \nu-(v^a \rho^\ast)_{,a} \mu)\;\rho
-(u^a\; \nu-v^a \mu_)\;\rho^\ast\;\rho_{,a})
\}
\nonumber\\
&=&\int \;d^D x\;\rho^\ast
\{
(u^a \;v^b_{,a}-v^a \;u^b_{,a}) \rho_{,b}\;
+(u^a\nu_{,a}-v^a\mu_{,a}) \rho
\}
\nonumber\\
&=&\int \;d^D x\;
\{[u,v]^a B_a+\frac{1}{2} [u,v]^a_{,a} B
\}
=i C^G([u,v])
\ea
In the first step we expanded out all terms, in the second step we relabelled 
$a\leftrightarrow b,\;x\leftrightarrow y$ in the third and fourth term,
in the third step we dropped the fourth term which is proportional to 
$\delta(x,y)$ and used above identity, in the fourth step we interchanged
labels in the second term and used the CAR, in the fifth step we integrated 
by parts and carried out the delta distribution so that we can drop the 
$x,y$ labels, in the sixth step we integrated by parts so that no derivative 
acts on $\rho^\ast$, in the seventh step we could drop the $\rho_{(ab)}$ term 
which is multiplied by $2 u^{[a} v^{b]}$ and could drop the third term 
which cancels against the term proportional to $\rho_{,a}$ from the 
second term and in the eighth step we used 
\be \label{4.5}
\partial_b\;[u,v]^b=
u^a_{,b}\;v^b_{,a}+u^a\;v^b_{,ab}
-v^a_{,b}\;u^b_{,a}+v^a\;u^b_{,ab}=u^a\nu_{,a}-v^a\mu_{,a}
\ee
Accordingly
\be \label{4.6a}
[C^G(u),C^G(v)]=i\; C^G(-[u,v])
\ee
i.e. the commutator equals $i$ times the result of the  classical Poisson 
bracket.

\subsection{Intrinsic Half Density Transformation for Scalar Fields}
\label{s4.2}

The scalar $\phi$ has density weight zero while $\pi$ has density weight
one. Therefore the apparently natural annihilation operator for white 
noise covariance 
\be \label{4.6}
A:=2^{-1/2}[\phi-i\pi]
\ee
behaves awkwardly under diffeomorphisms 
\be \label{4.7}
A\mapsto A_\varphi:=2^{-1/2}\;[\phi\circ\varphi-i\;J\;\pi\circ\varphi]
=\frac{1+J}{2}\;A\circ\varphi 
+\frac{1-J}{2}\;A^\ast\circ\varphi
\ee
where $J=|\det(\partial \phi/\partial x)|$. Thus the Fock vacuum $\Omega$ 
defined by $A\Omega=0$ cannot be left invariant by a potential unitary 
operator $U(\varphi)$ that accomplishes $U(\varphi) A U(\varphi)^{-1}
=A_\varphi$ as otherwise we get the contradiction 
$0=U(\varphi)A\Omega=A_\varphi \Omega
=\frac{1-J}{2}\;A^\ast\circ\varphi\Omega$. There is no contradiction 
if and only if $J=1$ i.e. for {\it unimodular diffeomorphisms} or
{\it volume preserving diffeomorphisms}. The 
set of unimodular diffeomorphisms forms a subgroup of the diffeomorphism 
group where the product is defined by concatenation of maps. Its Lie algebra 
consists of divergence free vecor fields with respect to the chosen volume 
form as follows from the expansion $\varphi^u_t(x)=x+t u(x)+o(t^2)$ in a
given chart and
$J=1+tu^a_{,a}+o(t^2)$ 
using the volume form of the Euclidian metric of that chart. 
 
Accordingly, for general diffeormorphisms we expect significant trouble when
using the Fock representation based on $A$. Indeed, when expanding 
$C^S(u)$ in terms of $A,A^\ast$ and normal ordering we find 
\ba \label{4.8}
C^S(u) &=&
\int\;d^D x\;u^a\;:\pi\phi_{,a}:
=\frac{i}{2}
\int\;d^D x\;u^a\;:(A-A^\ast)\;(A+A^\ast)_{,a}:
\nonumber\\
&=&
\frac{i}{2}
\int\;d^D x\;u^a
(A A_{,a}-A^\ast A_{,a}-A^\ast A_{,a}+A^\ast_{,a} A)
\nonumber\\
&=& -\frac{i}{2}
\int\;d^D x\;
\{\frac{1}{2}\;u^a_{,a}\;[A^2-(A^\ast)^2]+u^a\; [A^\ast A_{,a}-A^\ast_{,a} A]\}
\ea
where we integrated by parts.
As expected, the troublesome term $[A^\ast]^2$ is proportional to $u^a_{,a}$.
Only for $u^a_{,a}=0$ is (\ref{4.8}) densely defined on the Fock space,
otherwise it maps out of it which can be seen from the computation 
\be \label{4.9}
||C^S(u)\Omega||^2=\frac{1}{16}\delta(0,0)\;||u^a_{,a}||^2_{L_2(\sigma)}
\ee
no matter whether $\sigma$ is compact or not. 

This does not yet show that there are no Fock representations for scalar 
fields in which $\phi,\pi$ have density weight zero and one respectively 
because we can use a different covariance $\kappa$ (positive,
invertible integral kernel)
and set 
\be \label{4.10}
A_\kappa=2^{-1/2}[\kappa\cdot \phi-i\kappa^{-1}\cdot\pi]
\ee
and analyse whether there exist $\kappa$ such that $C^S(u)$ is densely 
defined in the Fock representation based on $A_\kappa \Omega_\kappa=0$. 
This analysis is non-trivial and reserved for the next section. We will
find that the answer depends on the dimension: For $D=1$ there do exist 
representations which however do not extend to $D>1$.

For the remainder of this subsection we consider a salar field 
$\hat{\phi}$ of density weight 1/2 so that its conjugate momentum 
$\hat{\pi}$ also 
has density weight 1/2 consistent with the Poisson bracket 
$\{\hat{\pi}(x),\hat{\phi}(y)\}=\delta(x,y)$. For such a scalar field 
the SDC would be 
\be \label{4.10a}
\hat{C}^S(u)=\frac{1}{2}\int\; d^Dx\;
[\hat{\pi}\hat{\phi}_{,a}-\hat{\pi}_{,a}\hat{\phi}]
\ee
since it correctly generates 
\be \label{4.11}
\{\hat{C}^S(u),\hat{\phi}=u^a\phi_{,a}+\frac{1}{2}\;u^a_{,a}\;\hat{\phi}
\ee
Since both $\hat{\phi},\hat{\pi}$ have equal density weight the annihilator
\be \label{4.12}
\hat{A}=2^{-1/2}[\hat{\phi}-i\hat{\pi}]
\ee
is still of the same density weight and is mapped to another annihilator 
under diffeomorphisms 
\be \label{4.13}
\hat{A}\to \hat{A}_\phi=J^{1/2}\;\hat{A}\circ\varphi
\ee
and the Fock vacuum $\hat{A}\hat{\Omega}=0$ is preserved under 
diffeomorphisms. This can also be confirmed by calculation the 
corresponding normal ordered operator 
\ba \label{4.14}
\hat{C}^S(u) &=& 
\frac{i}{4}\int\;d^D\;u^a\;
:[(\hat{A}-\hat{A}^\ast)\;(\hat{A}+\hat{A}^\ast)_{,a}
-(\hat{A}-\hat{A}^\ast)_{,a}\;(\hat{A}+\hat{A}^\ast)]:
\nonumber\\
&=&-\frac{i}{2}\int\;d^D\;u^a\;
[\hat{A}^\ast \hat{A}_{,a}-\hat{A}^\ast_{,a} \hat{A}]
\ea
in which the $\hat{A}^2,\; [\hat{A}^\ast]^2$ terms drop out using the CCR. 
The operator (\ref{4.14}) annihilates the vacuum. Moreover, comparing
(\ref{4.14}) with (\ref{4.1}) we see that they are formally identical 
if we formally substitute $\rho$ by $\hat{A}$. In verifying the SDA
we must use bosonic statistics and the corresponding identity
$[AB,CD]=A[B,CD]+[A,CD]B=A([B,C]D+C[B,D])+([A,C]D+C[A,D])B$ which differs
by some minus signs from the above identity in terms of 
anticommutators but by following the  
caculation (\ref{4.4}) step by step we see that 
$+[\rho_x,\rho^\ast_y]_+$ is replaced by $+[\hat{A}_x,\hat{A}^\ast_y]$ and   
$-[\rho^\ast_,\rho_y]_+$ is replaced by $+[\hat{A}^\ast_x,\hat{A}_y]$ and the
same for (anti)commutators with derivatives.
But these two distributions eactly coincide, thus without additional 
computation we immediately conclude 
\be \label{4.15}    
[\hat{C}^S(u),\hat{C}^S(v)]=
i\hat{C}^S(-[u,v]))
\ee
~\\
The question remains how $\hat{\phi},\hat{\pi}$ should come into existence
given that the canonical formulation uniquely leads to $\phi,\pi$ with
their unequal density weight. The observation is that we may use 
the scalar field momentum $\pi$ to define the density of weight 1/2
\be \label{4.16}
\sigma:=|\pi|^{1/2}
\ee
We define a momentum of density weight 1/2 by 
\be \label{4.17}
\hat{\pi}:=\frac{\pi}{|\pi|^{1/2}}\;\;
\Rightarrow\;\; |\hat{\pi}|=\sigma\;\;
\Rightarrow\;\; \pi=\sigma\;\;\hat{\pi}=|\hat{\pi}|\hat{\pi}
\ee
We thus have for the symplectic potential, dropping total functional 
exterior differentials as usual
\be \label{4.18}
\Theta
=-\int\;d^D x\; \phi\;[\delta\pi]
=-\int\;d^D x\; \phi\;[\sigma [\delta\hat{\pi}]
+{\rm sgn}(\hat{\pi})\hat{\pi}\;
[\delta\hat{\pi}]]
=-2\int\;d^D x\; \phi\;\sigma\; [\delta\hat{\pi}]
=\int\;d^D x\;  \hat{\pi}[\delta \hat{\phi}]
\ee
where 
\be \label{4.19}
\hat{\phi}=2|\pi|^{1/2} \; \phi\;\; \Rightarrow \phi=\frac{1}{2}
|\hat{\pi}|^{-1/2}\;\hat{\phi}
\ee
is canonically conjugate to $\hat{\pi}$. It follows 
\ba \label{4.20}
C^S(u) &=& \int\; d^Dx\; u^a\pi\phi_{,a}
=\frac{1}{2}\int\; d^Dx\; u^a \hat{\pi}\sigma (\frac{\hat{\phi}}{\sigma})_{,a}
\nonumber\\
&=& \frac{1}{2}\int\; d^Dx\; u^a (\hat{\pi}\hat{\phi}_{,a}
-\sigma^{-1}\hat{\pi}
{\rm sgn}(\hat{\pi})\hat{\pi}_{,a}\hat{\phi})
\nonumber\\
&=& \frac{1}{2}\int\; d^Dx\; u^a (\hat{\pi}\hat{\phi}_{,a}
-\hat{\pi}_{,a}\hat{\phi})
\ea
which is exactly (\ref{4.10}).\\
\\
Thus it is possible to obtain a Fock represntation of the SDA and SDG 
intrinsically, i.e. without using any other fields, for scalar fields 
after performing a canonical transformation on the classical 
phase space prior to quantisation. Note that this representation is 
manifestly background independent. Of course this reformulation is quite 
non trivial, e.g. in the HC the kinetic 
term $\pi^2$ becomes $\hat{\pi}^4$ and a mass term 
term $\phi^2$ becomes $\hat{\phi}^2/\hat{\pi}^2$. Thus the canonical 
transformation generates additional non-linearities and in particular
raises difficult 
questions about the divergences of the inverse of the operator 
$\hat{\pi}$. In that respect the transformations 
of the next subsection are more natural if one can keep the D-metric
sufficiently non degenerate in 
the quantum theory.

\subsection{Extrinsic half density transformation for all bosonic matter
fields}
\label{s4.3}

As fermions naturally come with a Fock representation based on density 
weight 1/2 as discussed in 
section \ref{s4.1}, we only have to construct Fock representations 
for bosonic tensor fields and the corresponding SDA or SDG. It is already 
difficult to do this for scalars with density weight zero, hence one would 
expect this to be even harder for higher rank tensors of density weight zero.
Consider for instance a co-vector field $q_a$ of density weight 
zero with conjugate momentum vector field $p^a$ of density weight unity.
The experience with scalar fields of the previous subsection has shown 
that that it is benefitial to turn both objects into tensor fields of the 
same type. One may wonder whether this is possible intrinsically. 
The first objective would be to construct a density out of $q_a,p^a$ and 
then pass to density 1/2 valued fields $\hat{q}_a,\hat{p}^a$ 
by a canonical transformation following the 
pattern of the previous subsection. It is not possible to do this 
algebraically, however, a natural scalar density of weight 1/2 is 
$\sigma=|\partial_a p^a|^{1/2}$. This is background independently defined 
because $p^a$ has density weight one. Thus one would begin the canonical 
transformation using $\hat{p}^a:=p^a\;\sigma^{-1}$. However, this is 
already difficult to invert since $\hat{p}^a_{,a}
=\frac{p^b_{,b}}{|p^c_{,c}|^{1/2}}-\frac{p^a\sigma_{,a}}{\sigma^2}$
cannot be solved algebraically for $p^a_{,a}$, it is a non-trivial first 
order PDE of the form (for $p^b_{,b}>0$)
\be \label{4.20a}
p^a \sigma_{,a}=\sigma^3-\sigma^2 \hat{p}^a_{,a}
\ee
to be solved for $\sigma$. Even if this can be solved in any useful way,   
the next even more difficult task would then be to lower the index of
$\hat{p}^a$ so that it be of same tensor type as 
$\hat{q}_a\propto \sigma q_a$. This would seem to require a {\it 
background metric} if done intrinsically
which then would make an annihilation constructed out of $\hat{q}_a,\hat{p}_a$
transform into a linear combination of annihilation and creation operators 
which comes with the difficulties discussed in the previous subsection.\\
\\
The discussion suggests that we use the {\it dynamic metric} or rather the 
D-Bein $e_a^j$ to accomplish the task {\it extrinsically} to turn 
canonical pairs of tensor fields 
$q^{a_1..a_A}\;_{b_1..b_B},\;p_{a_1..a_A}\;^{b_1 .. b_B}$ 
into scalar half densities as follows.
It will be sufficient to consider $A=B=1$, the general case follows 
analogously. We define
\ba \label{4.21}  
\hat{q}^j\;_k &:=& |\det(e)|^{1/2}\; e_a^j\;e^b_k\; q^a_b\;\;
\Leftrightarrow\;\; q^a_b=|\det(e)|^{-1/2}\; e^a_j\;e_b^k\; \hat{q}^j_k
\nonumber\\
\hat{p}_j\;^k &:=& |\det(e)|^{-1/2}\; e^a_j\;e_b^k\; p_a^b\;\;
\Leftrightarrow\;\; p_a^b=|\det(e)|^{1/2}\; e_a^j\;e^b_k\; \hat{p}_j^kk
\ea
and 
where $e^a_j e_a^k=\delta_j^k$.
Then the symplectic potential can be written with $P^a_j$ the momentum 
conjugate to $e_a^j$
\ba \label{4.22}
\Theta 
&=& \int \; d^Dx\; 
\{
P^a_j\;[\delta e_a^j]+p_a^b\;[\delta q^a_b]
\}
\nonumber\\
&=& \int \; d^Dx\; 
\{
P^a_j\;[\delta e_a^j]+\hat{p}_j^k \;[\delta \hat{q}^j_k]+
p_a^b \hat{q}^j_k [\delta (e^a_j e_b^k |\det(e)|^{-1/2})]
\}
\nonumber\\
&=& \int \; d^Dx\; 
\{
P^a_j\;[\delta e_a^j]+\hat{p}_j^k \;[\delta \hat{q}^j_k]+
p_a^b q^c_d  e_c^j e^d_k |\det(e)|^{1/2}
([\delta e^a_j] e_b^k |\det(e)|^{-1/2}
\nonumber\\
&& +e^a_j [\delta e_b^k] |\det(e)|^{-1/2}
+e^a_j e_b^k [\delta |\det(e)|^{-1/2}])
\}
\nonumber\\
&=& \int \; d^Dx\; 
\{
P^a_j\;[\delta e_a^j]+\hat{p}_j^k \;[\delta \hat{q}^j_k]+
p_a^b q^c_d  (
[\delta e^a_j] e_c^j \delta_b^d 
+\delta^a_c e^d_k [\delta e_b^k] 
+\delta^a_c \delta_b^d |\det(e)|^{1/2} [\delta |\det(e)|^{-1/2}])
\}
\nonumber\\
&=& \int \; d^Dx\; 
\{
P^a_j\;[\delta e_a^j]+\hat{p}_j^k \;[\delta \hat{q}^j_k]+
p_a^b (-e^a_j] [\delta e_c^j] q^c_b 
+q^a_d e^d_k [\delta e_b^k] 
-\frac{1}{2} q^a_b e^c_j [\delta e_c^j])
\}
\nonumber\\
&=& \int \; d^Dx\; 
\{
[P^a_j
- p_c^b q^a_b e^c_j   
+p_c^a q^c_b e^b_j  
-\frac{1}{2} p_c^b q^c_b e^a_j]\;
[\delta e_a^j]+\hat{p}_j^k \;[\delta \hat{q}^j_k]
\}
\nonumber\\
&=& \int \; d^Dx\; 
\{
\hat{P}^a_j\; [\delta \hat{e}_a^j]
+\hat{p}_j^k \;[\delta \hat{q}^j_k]
\}
\ea
where 
\be \label{4.23}
\hat{e}_a^j:=e^a_j,\;
\hat{P}^a_j:=
P^a_j
- p_c^b q^a_b e^c_j   
+p_c^a q^c_b e^b_j  
-\frac{1}{2} p_c^b q^c_b e^a_j
\ee
The calculation reveals that the map from 
$(e_a^j,P^a_j),\;(q^a_b,p_a^b)$ to
$(\hat{e}_a^j,\hat{P}^a_j),\;(\hat{q}^j_k,\hat{p}_j^k)$ 
is a canonical transformation. The first term in $\hat{P}^a_j-P^a_j$ 
is due to the contra-variant index $a$ in $q^a_b$, 
the second due to the co-variant index $b$ and the third due to the 
switch from density weight zero to 1/2. For a tensor field $q$ of type
$(A,B,w=0)$ we would have correspondingly $A$ terms of the contra-variant
type, $B$ terms of the co-variant type and one term of the density type.
The contributions from all tensor fields involved simply add up 
in $\hat{P}^a_j-P^a_j$. Thus, we have also treated also the general case.
In the case above, inverse map of the canonical transformation is given
by  
\ba \label{4.24}
q^a_b &=& \hat{e}^a_j\;\hat{e}_b^k \;|\det(\hat{e})|^{-1/2}\; \hat{q}^j_k
\nonumber\\
p_a^b &=& \hat{e}_a^j\;\hat{e}^b_k \;|\det(\hat{e})|^{1/2}\; \hat{p}_j^k
\nonumber\\
e_a^j &=& \hat{e}_a^j
\nonumber\\
P^a_j &=& \hat{P}^a_j
+p_c^b q^a_b e^c_j   
-p_c^a q^c_b e^b_j  
+\frac{1}{2} p_c^b q^c_b e^a_j
\nonumber\\&=& \hat{P}^a_j
+\hat{p}_j^k \hat{q}^l_k\;\hat{e}^a_l    
-\hat{p}_k^l \hat{q}^k_j \hat{e}^a_l   
+\frac{1}{2} \hat{p}_l^k \hat{q}^l_k \hat{e}^a_j
\ea
These expressions have to be substituted into the HC 
and the various Gauss constraints (GC) which of course make them 
algebraically more complicated but it poses no problem in principle 
because (\ref{4.24})
are just algebraic relations that do not involve solving PDE's. In particular,
since in the HC the $P^a_j, p_a^b$ appear without derivatives, there are no 
new derivative terms appearing there. 

After these transformations all tensor fields and fermion fields 
together with their conjugate momenta are scalar densities of weight 1/2,
except for the D-Bein conjugate pair $(\hat{e}_a^j,\hat{P}^a_j)$. The 
transformation cannot applied to $q_a^k:=e_a^k$ viewed as D co-vector fields
labelled index $k$ itself because $\hat{q}_j^k=e^a_j q_a^k |\det(e)|^{1/2}
=\delta_j^k |\det(e)|^{1/2}$ cannot provide a canonical transformation 
as $\hat{q}_j^k$ captures only one degree of freedom rather than $D^2$.
Note also that if $q,p$ is a tensor of type $(A,B,w=0), (B,A,w=1)$ 
and transforms 
in some representation $r,r^\ast$ 
of a Yang-Mills gauge group or (the covering group
of) SO(D), then $\hat{q},\hat{p}$ are tensors of type 
$(A=0,B=0,w=1/2),(B=0,A=0,w=1/2)$ transforming
in 
$[\otimes^A R]\otimes [\otimes^B R^\ast]\otimes r,\;
[\otimes^A R^\ast]\otimes [\otimes^B R]\otimes r^\ast$
where $R$ is the 
defining representation of SO(D) and $R^\ast$ its dual. Because of this 
additional gauge symmetries of Yang-Mills type and their corresponding 
Gauss constrains remain unbroken and one simply has to to substitute 
by the new fields.  

More in detail, for the gravitational Gauss constraint
for the frame rotation group SO(D) there are also no new derivative terms 
appearing because it also involves all fields only algebraically
\cite{13} when written in terms of $e_a^j, P^a_j$. When substituting  
$\hat{e},\hat{P}$ new terms depending on $\hat{q},\hat{p}$ appear that 
take their behaviour under frame rotations correctly into account
\ba \label{4.25}
G_{kj}
&=& 
2\delta_{l[k} e_a^l \;(P^a_{j]}
=2\delta_{l[k} \hat{e}_a^l \;(\hat{P}^a_{j]}
+\hat{p}_{j]}^m \hat{q}^n_m\;\hat{e}^a_n    
-\hat{p}_m^n \hat{q}^m_{j]} \hat{e}^a_n   
+\frac{1}{2} \hat{p}_n^m \hat{q}^n_m \hat{e}^a_{j]})
\nonumber\\
&=&
2\delta_{l[k} \hat{e}_a^l \;\hat{P}^a_{j]}
+2\delta_{l[k} \hat{e}_a^l
(
\hat{p}_{j]}^m \hat{q}^l_m    
-\hat{p}_m^l \hat{q}^m_{j]}    
+\frac{1}{2} \hat{p}_n^m \hat{q}^n_m \delta^l_{j]})
\nonumber\\
&=&
2\delta_{l[k} \hat{e}_a^l \;\hat{P}^a_{j]}
+2\delta_{l[k} \hat{e}_a^l
(\hat{p}_{j]}^m \hat{q}^l_m-\hat{p}_m^l \hat{q}^m_{j]})    
\ea
For YM type 
constraints we have only algebraic dependence on the fermion fields which 
are not affected by above transformation, only algebraic dependence for 
the tensor fields and their momenta which are different from the YM 
connection and its canonical momentum while the connection contribution
has a specific form. Consider e.g. a G Gauss constraint
\be \label{4.26} 
G_\alpha=(\partial_a E^a_\alpha+f_{\alpha\beta}\;^\gamma A_a^\beta E^a_\gamma
+p_a^{bI}\;[\tau^r_\alpha]_I\;^J q^a_{bJ}
\ee
where $\tau_\alpha$ are generators of the Lie algebra $L(G)$ with stucture
constants $[\tau_\alpha,\tau_\beta]=f_{\alpha\beta}\;^\gamma \tau_\gamma$,
$\tau^r_\alpha=[\frac{d}{dt} r(\exp(t\tau_\alpha))]_{t=0}$ representation 
matrices of $L(G)$ in the representation $r$ in which $q^a_b$ transforms 
while $p_a^b$ transforms in the dual (or contra-gredient) 
representation $r^\ast$ with $[r^\ast(g)]^I\;_J:=[r(g^{-1})]_J\;^I$. 
Substituting 
$A_a^\alpha=\hat{A}_j^\alpha e_a^j |\det(e)|^{-1/2},\;
E^a_\alpha=\hat{E}^j_\alpha e^a_j |\det(e)|^{1/2}$
and (\ref{4.21}) we find for (\ref{4.26})
\be \label{4.27}
G_\alpha=\partial_a(|\det(\hat{e})|^{1/2} \hat{e}^a_j 
\hat{E}^j_\alpha)+f_{\alpha\beta}\;^\gamma 
\hat{A}_j^\beta \hat{E}^j_\gamma
+\hat{p}_j^{kI}\;[\tau^r_\alpha]_I\;^J \hat{q}^j_{kJ}
\ee
The interesting part is now the SDC for the benefit of 
which we invented this canonical 
transformation. The fermion fields are unaffected by it, hence it suffices 
to consider ($L_u$ denotes the Lie derivative with respect to $u$ and 
$\sigma:=|\det(e)|^{1/2}|$) 
\ba \label{4.28}
C(u) 
&=& \int\; d^Dx\; \{P^a_j \;[L_u e^j]_a
+p_a^b\;[L_u q]^a_b\}
\nonumber\\
&=& \int\; d^Dx\; \{P^a_j \;[L_u e^j]_a
+p_a^b\;([L_u \sigma^{-1} \hat{q}^j_k] e^a_j e_b^k
+\sigma^{-1} \hat{q}^j_k (
[L_u e_j]^a e_b^k + e^a_j [L_u e^k]_b)) 
\}
\nonumber\\
&=& \int\; d^Dx\; \{P^a_j \;[L_u e^j]_a
+\hat{p}_j^k u^a \sigma [\sigma^{-1} \hat{q}^j_k]_{,a} 
+ \hat{q}^j_k (\hat{p}_l^k e_a^l [L_u e_j]^a + \hat{p}_j^l e^b_l [L_u e^k]_b)) 
\}
\nonumber\\
&=& \int\; d^Dx\; \{P^a_j \;[L_u e^j]_a
+\hat{p}_j^k u^a ([\hat{q}^j_k]_{,a}-\frac{1}{2} \hat{q}^j_k e^b_l e^l_{b,a}) 
+\hat{q}^j_k (-\hat{p}_l^k e^a_j [L_u e^l]_a + \hat{p}_j^l e^b_l [L_u e^k]_b)) 
\}
\nonumber\\
&=& \int\; d^Dx\; \{P^a_j \;[L_u e^j]_a
+u^a \hat{p}_j^k [\hat{q}^j_k]_{,a}
-\frac{1}{2} \hat{p}_j^k \hat{q}^j_k e^b_l ([L_u e^l]_b-u^a_{,b} e^l_a) 
-\hat{q}^l_k \hat{p}_j^k e^a_l [L_u e^j]_a + 
+\hat{q}^k_j \hat{p}_k^l e^a_l [L_u e^j]_a)) 
\}
\nonumber\\
&=& \int\; d^Dx\; \{[P^a_j
-\hat{q}^l_k \hat{p}_j^k e^a_l 
+\hat{q}^k_j \hat{p}_k^l e^a_l 
-\frac{1}{2} \hat{p}_k^l \hat{q}^k_l e^a_j]\; [L_u e^j]_a
+u^a (\hat{p}_j^k [\hat{q}^j_k]_{,a}
-\frac{1}{2} [\hat{p}_j^k \hat{q}^j_k]_{,a})
\}
\nonumber\\
&=& \int\; d^Dx\; \{\hat{P}^a_j [L_u \hat{e}^j]_a
+\frac{1}{2}(\hat{p}_j^k [\hat{q}^j_k]_{,a}
-[\hat{p}_j^k]_{,a} \hat{q}^j_k)
\}
\ea
In summary, after above canoncial transformations 
the gravitational field $\hat{e}_a^j,\hat{P}^a_j$ retains its tensorial 
tarnsformation behaviour while all matter fields have been transformed into
scalar densities of weight 1/2. The representations $R,R^\ast$ 
can be identified
because we have the Euclidian metric $\delta_{jk},\delta^{jk}$ to intertwine
the representations $R,R^\ast$. Thus we can identify e.g.
\be \label{4.29}
\hat{q}^j_k\equiv \hat{q}_{jk},\;\;
\hat{p}_j^k\equiv \hat{p}_{jk},\;\;
\ee
which is in fact a canonical transformation.
Thus it is no longer relevant whether $j$ is contra-variant or co-variant,
just the sequence in which these indices appear is relevant. This enables us 
to give geometrical meaning to the annihilation operator
\be \label{4.30}
\hat{A}_{jk}:=2^{-1/2}[\hat{q}_{jk}-i\hat{q}_{jk}]
\ee
By the results of the previous subsection, the resulting Fock representation
is a non-anomalous, background independent and continuous unitary 
representation of the SDG for all matter fields. 

As far as the 
geometrical degrees of freedom are concerned, we may resort to a unitary 
but discontinuous representation of Narnhofer-Thirring type
\cite{22}, e.g. defined by $\hat{e}_a^j(x)\Omega=0$ similar 
as the one used in Loop Quantum Gravity \cite{9,12}. The reason for this 
choice is that the gravitational constraint $C^G(u)=\int\; d^Dx\;
\hat{P}^a_j [L_u \hat{e}]_a^j$ just derived can be ordered such that all 
terms depending on $\hat{e}$ are to the right while it does not allow 
for complex annihilation and creation operators with no $(A^\ast)^2$ terms.
Then in this ordering the constraint formally annihilates the vacuum.  
The representation 
space of that representation has as dense domain the Weyl states 
$w[0,g]\Omega:=\exp(i\int\; d^Dx\; g_a^j \hat{P}^a_j)$ with inner product
\be \label{4.31}
<w[0,g]\Omega,w[0,\tilde{g}]\Omega>=\delta_{g,\tilde{g}}
\ee
where on the right hand side we have the Kronecker symbol, i.e. the 
two vectors have non-vanishing inner product if and only if 
$g_a^j(x)=\tilde{g}_a^j(x)$ for all $a,j,x$. The relation (\ref{4.31}) 
is a simple consequence of the Weyl relations and $\hat{e}_a^j(x)\Omega=0$:
Let $w(f,0)=\exp(i\int\;d^Dx\; f^a_j \hat{e}_a^j)$ then 
\ba \label{4.32}
&&<w[0,g]\Omega,w[0,\tilde{g}]\Omega>=
<\Omega,w[0,\tilde{g}-g]\Omega>
=<w[f,0]\Omega,w[0,\tilde{g}-g]\Omega>
\nonumber\\
&=& <\Omega,w[-f,0]\;w[0,\tilde{g}-g]\;w[f,0]\Omega>
=\exp(-i\int\;dx\;f^a_j[\tilde{g}-g]_a^j)\;
<w[0,g]\Omega,w[0,\tilde{g}]\Omega>
\ea
for all $f$. Picking $f^a_j=s\;[\tilde{g}-g]_a^j$ for any real number $s$
it follows that (\ref{4.32}) must vanish unless $g=\tilde{g}$ almost 
everywhere, thus $g=\tilde{g}$ everywhere when the space of $g$ consists 
of at least continuous functions.

Then $U(\varphi) w[0,g]\Omega=w[0, g\circ \varphi^{-1}]\Omega$ is a 
discontinuous but unitary non-anomalous representation of the SDG on the 
non-separable Hilbert space defined by (\ref{4.31}). 

\subsection{Half-density transformation for all geometry and matter 
fields}
\label{s4.4}

One can avoid the discontinuous representation in the geometry sector 
sketched at the end of the previous subsection provided that the 
matter content consists of at least D scalar fields $\phi^I$ with 
conjugate momentum $\pi_I$ with $I=1,..,D$ and natural density weights 
zero and unity each, provided that the matrix with entries 
\be \label{4.33}
h_a^I:=\phi^I_{,a} 
\ee
is non-degenerate, i.e. the map $x\mapsto \phi(x)$ provides a {\it dynamical
diffeomorphism} in the sense that the fields $\phi^I$ are not components 
of a background coordinate transformation but rather
enter themselves into
the SDC and HC. Then (\ref{4.33}) provides a flat (i.e. the metric 
$h_{ab}:=\delta_{IJ} h^I_a h^J_b$ is flat) co-D-Bein with inverse $h^a_I$.
In that case, in a fist step we use the extrinsic 
strategy of section \ref{s4.3} to 
turn all
fields into scalar densities of weight 1/2 {\it including the gravitational
field} but {\it excluding} the fields $\phi^I,\pi_I$ which keep their 
density weight zero and unity respectively when carrying out the 
corresponding canonical transformation. Then, in a second step, we use 
the intrinsic strategy of section \ref{s4.2} in order to obtain the transformed
scalar fields $\hat{\phi}^I,\hat{\pi}_I$ of of equal density weight 1/2.
Surprisingly this is possible without solving PDE's despite the fact that 
the canonical transformations involve not only algebraic but also 
derivative expressions of the fields.

Proceeding to the details, it will again be sufficient to exemplify this 
for a tensor field pair $q^a_b,p_a^b$. We set with $\sigma=|\det(h)|^{1/2}$
\be \label{4.34}     
\hat{q}^I_J=\sigma\; h_a^I\; h^b_J\;q^a_b,\;\;
\hat{p}_I^J=\sigma^{-1/2}\; h^a_I\; h_b^J\;p_a^b
\ee
which up to the factor $\sigma^{\pm 1}$ is nothing but the pull-back 
map under the dynamical diffeomorphism. 
Then the symplectic potential is 
\ba \label{4.35}
\Theta &=& \int\; d^Dx\;
\{
\pi_I\;[\delta\phi^I]+p_a^b\; [\delta q^a_b]
\}
\nonumber\\
&=& \int\; d^Dx\;
\{
\pi_I\;[\delta\phi^I]+\hat{p}_I^J\;h_a^I h^b_J \sigma [\delta 
(h^a_M h_b^N \sigma^{-1} \hat{q}^M_N)]
\}
\nonumber\\
&=& \int\; d^Dx\;
\{
\pi_I\;[\delta\phi^I]+\hat{p}_I^J \;[\delta \hat{q}^I_J]
+\hat{p}_I^J\;(
[\delta h^a_M] h^a_I \hat{q}^M_J
+h_b^J [\delta h_b^N] \hat{q}^I_N
-\frac{1}{2} h^a_M [\delta h_a^M]\hat{q}^I_J)
\}
\nonumber\\
&=& \int\; d^Dx\;
\{
\pi_I\;[\delta\phi^I]+\hat{p}_I^J \;[\delta \hat{q}^I_J]
- \hat{p}_I^J h^a_L [\delta h^a_I] \hat{q}^L_J
+\hat{p}_L^J h^a_J [\delta h_a^I] \hat{q}^L_I
-\frac{1}{2}\hat{p}_L^J h^a_I [\delta h_a^I]\hat{q}^L_J
\}
\nonumber\\
&=& \int\; d^Dx\;
\{
[\pi_I+(
\hat{p}_I^J \hat{q}^L_J h^a_L
-\hat{p}_L^J  \hat{q}^L_I h^a_J
+\frac{1}{2}\hat{p}_L^J  \hat{q}^L_J h^a_I)_{,a}]
[\delta\phi^I]+\hat{p}_I^J \;[\delta \hat{q}^I_J]
\}
\ea 
We thus obtain the canonical transformation (\ref{4.34}) and 
\be \label{4.36}
\hat{\phi}^I=\phi^I,\;
\hat{\pi}_I=
\pi_I+
(h^c_I p_c^b q^a_b 
-p_c^a  \hat{q}^c_b h^b_I 
+\frac{1}{2} p_b^c  q^b_c h^a_I)_{,a}
\ee
Note that the round bracket in (\ref{4.36}) is a vector density so that 
its divergence is background independent and yields a scalar density. We 
can invert (\ref{4.34}) and (\ref{4.36}) to find with 
$\hat{h}_a^I=\hat{\phi}^I_{,a}$ etc.
\ba \label{4.37}
&& q^a_b=\hat{\sigma}^{-1} \hat{h}^a_I \hat{h}_b^J \hat{q}^I_J,\;\;
p_a^b=\hat{\sigma} \hat{h}_a^I \hat{h}^b_J \hat{p}_I^J
\nonumber\\
&& \phi^I=\hat{\phi}^I,\;\;
\pi_I=\hat{p}_I
-(
\hat{p}_I^J \hat{q}^L_J \hat{h}^a_L
-\hat{p}_L^J  \hat{q}^L_I \hat{h}^a_J
+\frac{1}{2}\hat{p}_L^J  \hat{q}^L_J \hat{h}^a_I)_{,a}
\ea
The reason why this can be done is that the canonical transformation does 
not depend on derivatives of $\pi_I$ and is trivial for $\phi^I$. As 
compared to the previous subsection, substituting this into the 
HC now will yield more complicated expressions due to the derivatives acting 
on the momenta $\hat{p}_I^J$ but can be done explicitly.

The computation of the SDC can be based on the calculation of (\ref{4.28}) 
if substitute $\hat{q}^j_k,\hat{p}_j^k,\hat{e}_a^j,\hat{e}^a_j$ by 
$\hat{q}^I_J,\hat{p}_I^J,\hat{h}_a^I,\hat{h}^a_I$ and gives immediately 
\ba \label{4.38}
C(u) &=& \int\; d^Dx\; 
\{
\pi^I\;[L_u \phi^I] 
+p^a_b\; [L_u q]^a_b
\nonumber\\
&=& 
\int\; d^Dx\; 
\{
\pi^I\;[L_u \phi^I]
+u^a \hat{p}_I^J [\hat{q}^I_J]_{,a}
-\frac{1}{2} \hat{p}_J^L \hat{q}^J_L h^a_I ([L_u h^I]_a-u^b_{,a} h^I_b) 
-\hat{q}^L_J \hat{p}_I^J h^a_L [L_u h^I]_a + 
+\hat{q}^J_I \hat{p}_J^L h^a_L [L_u h^I]_a)
\}
\nonumber\\
&=& 
\int\; d^Dx\; 
\{
[\pi^I\;
+(\frac{1}{2} \hat{p}_J^L \hat{q}^J_L h^a_I 
+\hat{q}^L_J \hat{p}_I^J h^a_L 
-\hat{q}^J_I \hat{p}_J^L h^a_L)_{,a}] [L_u \phi^I]+\frac{1}{2}u^a 
(\hat{p}_I^J [\hat{q}^I_J]_{,a}-[\hat{p}_I^J]_{,a} [\hat{q}^I_J)
\}
\nonumber\\
&=&
\int\; d^Dx\; u^a\;
\{\hat{\pi}_I \hat{\phi}_{,a}+\frac{1}{2}
(\hat{p}_I^J [\hat{q}^I_J]_{,a}-[\hat{p}_I^J]_{,a} [\hat{q}^I_J)
\}
\ea
where in the last step we used $[L_u h^I]_a=[L_u \phi^I]_{,a}$. 

Thus 
$\hat{\phi}_I,\hat{\pi}_I$ are still scalar densities of weight zero and one 
respectively. We now can apply the intrinsic transformation of section 
\ref{s4.2} to each of them separately without involving the 
$\hat{q}^I_J,\hat{p}_I^J$ any more and obtain canonical pairs of scalar
half densities $\tilde{\phi}^I,\tilde{\pi}_I$. Then {\it all geometry
and matter degrees of freedom} are transformed 
into scalar densities of weight 1/2. 
Then the {\it background independent, anomaly free, strongly 
continuous, unitary Fock representation} of the SDG constructed in sections 
(\ref{s4.1}) and (\ref{s4.2}) can be constructed. For instance the scalar,
YM  and geometry contribution to the SDC would be  
\be \label{4.39}
C(u)=\frac{1}{2}\int\;d^Dx\;u^a
\{
\tilde{\pi}_I \;\tilde{\phi}^I_{,a}-\tilde{\pi}_{I,a} \;\tilde{\phi}^I+
+\hat{E}^I_\alpha \hat{A}_{I,a}^\alpha-\hat{E}^I_{\alpha,a} \hat{A}_I^\alpha
+\hat{P}^I_j \hat{e}_{I,a}^j-\hat{P}^I_{j,a} \hat{A}_I^j
\}
\ee
and we obtain the annihilators by ignoring the position of the scalar species 
index (formally moving it with $\delta_{IJ},\delta^{IJ}$)
\be \label{4.40}
\tilde{a}_I=2^{-1/2}[\tilde{\phi}^I-i\tilde{\pi}_I],\;
\hat{a}_I^\alpha=2^{-1/2}[\hat{A}_I^\alpha-i\hat{E}^I_\alpha],\; 
\hat{a}_I^j=2^{-1/2}[\hat{e}_I^j-i\hat{P}^I_j]
\ee
Of course under the assumptions of this subsection, one may also 
use the matter sector given by $\phi^I, \pi_I$ to provide 
a preferred gauge fixing \cite{8} to the coordinates $y^I:=\phi^I(x)$.
Then one just has to substitute 
the solution $\pi_I=-h_a^I C'_a$ of the SDC into the HC to reduce the system
and representations of the SDC no longer have to be constructed.

\section{Fock representations without half density transformations}
\label{s5}

While the Fock representations constructed in the previous section are 
geometrically natural, they are far away from the background dependent 
Fock representations of QFT in CST (curved spacetime) defined using 
e.g. a background Laplacian. This is no constradiction because the Fock
reps. of the SDC are constructed to {\it define} the SDC and of 
{\it non SD invariant operators}, that is, it is a {\it kinematical 
reprsentation}. Its solutions 
(kernel) are expected to be distributions (generalised zero eigenvectors)
and one must construct a new Hilbert space of solutions and a new 
{\it physical representation} of SD invariant operators (observables). 
The physical representation may then be based on that more familiar type   
of Fock representations. This is even expected because one can instead 
of defining a quantum SDC on a kinematical HS reduce the theory classically
imposing gauge fixing conditions that select a coordinate system and 
feed in information about a preferred background that may then be used to
construct those more familiar representations, possibly perturbatively. 

Yet the half density transformations render the HC even more complicated 
than it already is because of the high non-linearity of the necessary 
canonical transformations. Hence one may wonder whether Fock representations
exist, background dependent or not, 
which do not make use of those canonical transformations, keep the density 
weight as it is and do not mix field species. In this section we show 
for the illustrative example of a scalar field on $\sigma=\mathbb{R}^D$ 
that such Fock reps. based on a reflection invariant and translation invariant 
covariance of the corresponding Gaussian measure do not exist in $D>1$ while 
they do exist with central extension in $D=1$, We will also be able to 
pin point the reason for this dimension dependence. We will argue that 
this obstruction is also valid for other topologies, higher tensor rank and 
covariances without symmetries.\\
\\
We begin generally and consider real valued tensor fields $q$ of type 
$(A,B,w=0)$ on $\sigma$ with conjugate momentum 
$p$ of type $(B,A,w=1)$. We use the multi index notation
$q^\mu(x):=q^{a_1..a_A}\;_{b_1..b_B}(x)$ and 
$p_\mu(x):=p_{a_1..a_A}\;^{b_1..b_B}(x)$ and and consider the 1-particle 
HS $L_2(\sigma)$ on complex valued tensor fields of that type 
with inner product $<f,g>=\int\; d^Dx\; 
[f^\mu(x)]^\ast \delta_{\mu\nu} g^\nu(x)$ which 
of course breaks SD invariance. We define likewise 
$p^\mu(x)=\delta^{\mu\nu}p_\nu(x), q_\mu(x)=\delta_{\mu\nu} q^\nu(x)$.
Let
$\kappa_{\mu,\nu}(x,y)=\kappa_{\nu,\mu}(y,x)$ be a positive,
invertible integral kernel on $L_2$, i.e. $<f,\kappa f>\ge 0$ where 
$[\kappa f]_\mu(x):=\int\; d^Dx\; \kappa_{\mu,\nu}(x,y) q^\nu(y)$.

We define a Fock representation of the CCR using the annihilation operator
\be \label{5.1}
A_\mu(x):=2^{-1/2}[\kappa q-i\kappa^{-1}p]_\mu(x),\;
A(f):=<f,A>_{L_2}
\ee
and the vacuum $A_\mu(x)\Omega=0$. The question is now for which $\kappa$ if 
any, the SDC, normal ordered with respect to chosen $\kappa$, becomes 
a well defined operator on the corresponding Fock space ${\cal H}$
which is the completion of the span of vectors 
$A(f_1)^\ast..A(f_N)^\ast \Omega$. We recall that ${\cal H}=L_2(Q,d\mu)$ is 
an $L_2$ space with Gaussian measure $\mu$ over a quantum configuration
space $Q$ of tempered 
distributional, real valued tensor fields of type $(A,B,w=0)$
\cite{23}. The covariance of that Gaussian measure is given by 
$C=(2\kappa)^{-2}$ as can be seen by computing the generating functional
for real valued $f$
\ba \label{5.2}
\mu(\exp(i<f,.>)) &=& \int_Q\; d\mu(q)\;e^{i<f,q>}=
<\Omega,e^{i<f,q>}\Omega>
\\
&=&
<\Omega,e^{i 2^{-1/2}[A(\kappa^{-1}f)^\ast+A(\kappa^{-1} f)]}\Omega>
=\exp(-\frac{1}{4}<f,\kappa^{-2} f>)=:\exp(-<f,C f>)
\nonumber
\ea
where we used the BCH formula.

To examine this question we decompose the SDC into annihilation and creation 
operators and normal order. We have (focussing just on this type of tensor 
field because in this intrinsic construction the contributions from different
field species mutually commute)
\ba \label{5.3}
&& C(u) =\int\; d^Dx\; p^\mu\; [L_u q]_\mu
\\
&=&\int\; d^Dx\; p^\mu\; \{
u^a q_{\mu,a}
-\sum_{k=1}^A\;u^{a_k}_{,e} q^{a_1..\hat{a}_k e ..a_A}\;_{b_1..b_B}
+\sum_{l=1}^B\;u^e_{,b_l} q^{a_1..a_A}\;_{b_1..\hat{b}_l e ..b_B}
\}
\nonumber\\
&=& \frac{1}{2}\int\; d^Dx\; [p^\mu\; q_{\mu,a}-p^\mu_{,a} q_\mu]\;u^a
\nonumber\\
&&+\int\; d^Dx\; p^\mu \{
-\frac{1}{2} u^a_{,a} 
\prod_{m=1}^A \delta^{a_m}_{c_m}\prod_{l=1}^B \delta^{d_n}_{b_n}
-\sum_{k=1}^A\;u^{a_k}_{,e}
\delta^e_{c_k}\prod_{k\not=m=1}^A \delta^{a_m}_{c_m}\prod_{n=1}^B 
\delta^{d_n}_{b_n}
+\sum_{l=1}^b\;u^e_{,b_l}
\delta^{d_l}_e \prod_{m=1}^A \delta^{a_m}_{c_m}\prod_{l\not=n=1}^B 
\delta^{d_n}_{b_n}\} q^{c_1..c_A}_{d_1..d_B}
\}
\nonumber\\
&=:& \frac{1}{2}\int\; d^Dx\; [p^\mu\; q_{\mu,a}-p^\mu_{,a} q_\mu]\;u^a
+\int\; d^Dx\; p_\mu T(u)^\mu_\nu q^\nu
\nonumber\\
&=:& D(u)+S(u)
\ea
Here we have split $C(u)$ into two pieces $D(u), S(u)$. The first 
``differential'' piece 
$D(u)$ contains no derivatives of $u$ and by itself generates 
diffeomorphisms treating $q^\mu,p_\mu$ as if they were 
scalar half densities.
The second ``scaling'' piece $S(u)$ contains no derivatives of $q,p$ but 
contains a matrix valued, linear first order partial 
differential operator $u\mapsto Tu$. 

It is interesting to note that the $D(u)$ themselves represent the SDA because
they generates spatial diffeos treating $p_\mu, q^\mu$ as if they 
were scalar half densities  
\be \label{5.3a}
\{D(u),D(v)\}=-D([u,v])
\ee
Therefore, as $S(u)$ does not contain derivatives of $p_\mu,q^\mu$, $D(u)$ 
acts on $S(v)$ treating $p_\mu \; q^\nu$ as a scalar density of weight one
\be \label{5.3b}
\{D(u),S(v)\}=
\int \; d^Dx\; (T(v)^\mu\;_\nu)\; [u^a p_\mu q_\nu]_{,a} 
=-\int \; d^Dx\; p_\mu\; u^a (T(v)^\mu\;_{\nu,a}\; q_\nu
\ee
while explicit calculation gives immediately
\be \label{5.3c}
\{S(u),S(v)\}=-\int \; d^Dx\; p_\mu\; ([T(u),T(v)])^\mu\;_\nu\; q_\nu
\ee
Since $\{C(u),C(v)\}=-C([u,v])=-D([u,v])-S([u,v])$ we obtain without 
effort the interesting matrix identity
\be \label{5.3d}
T([u,v])=
u^a\;(T(v))_{,a}-v^a\;(T(u))_{,a}+[T(u),T(v)]
\ee
We expect the trouble to come 
from the $S(u)$ part 
while $D(u)$ has a chance to be well defined in the chosen Fock 
representation.   
To confirm this we substitute $q=2^{-1/2}\kappa^{-1}[A+A^\ast]$ and 
$p=i 2^{-1/2} [A-A^\ast]$ into $D(u)$ and normal order. We find 
\ba \label{5.4}
D(u) &=& 
:
\frac{i}{2}\;\int\;d^Dx\;u^a\;\delta^{\mu\nu}\;
\{
(\kappa[A-A^\ast])_\mu\;(\kappa^{-1}[A+A^\ast])_{\nu,a}
-(\kappa[A-A^\ast])_{\mu,a}(\kappa^{-1}[A+A^\ast])_\nu
\}
:
\nonumber\\
&=&
:
\frac{i}{2}\;\int\;d^Dx\;u^a\;\delta^{\mu\nu}\;
\{
-[\kappa A^\ast]_\mu\;[\kappa^{-1} A]_{\nu,a}
+[\kappa^{-1} A^\ast]_{\mu,a} \;[\kappa A]_\nu
-[\kappa A^\ast]_\mu\;[\kappa^{-1} A^\ast]_{\nu,a}
+[\kappa A]_\mu\;[\kappa^{-1} A]_{\nu,a}
\nonumber\\
&&
+[\kappa A^\ast]_{\mu,a}\;[\kappa^{-1} A]_\nu
-[\kappa^{-1} A^\ast]_\mu \;[\kappa A]_{\nu,a}
+[\kappa A^\ast]_{\mu,a}\;[\kappa^{-1} A^\ast]_\nu
-[\kappa A]_{\mu,a}\;[\kappa^{-1} A]_\nu
\}
:
\ea
If $\kappa$ is the identity kernel, the third and fourth term 
respectively would cancel with the the seventh and eighth term respectively
in agreement with the previous section. Then $D(u)$ would not contain 
$(A^\ast)^2$ terms and $D(u)$ would be well defined on the Fock space.
For general $\kappa\not={\rm id}$ we have 
\ba \label{5.6}
&&4\; ||D(u)\Omega||^2= 
\int\; d^D x\; \int\; d^Dy\;u^a(x)\;u^b(y)\; 
\delta^{\mu_1\nu_1}\;\delta^{\mu_2\nu_2}\;
\nonumber\\
&& <\Omega,
([\kappa A]_{\mu_1}\;[\kappa^{-1} A]_{\nu_1,a} 
-[\kappa A]_{\mu_1,a}\;[\kappa^{-1} A]_{\nu_1})(x)\;
([\kappa A^\ast]_{\mu_2}\;[\kappa^{-1} A]_{\nu_2,b} 
-[\kappa A^\ast]_{\mu_2,b}\;[\kappa^{-1} A^\ast]_{\nu_2})(y))\Omega
\nonumber\\
&=& \int\; d^D x\; \int\; d^Dy\;u^a(x)\;u^b(y)\; 
\delta^{\mu_1\nu_1}\;\delta^{\mu_2\nu_2}\;
\int\; d^{2D} u\; \int\; d^{2D}v\;
\nonumber\\
&& \{
\kappa_{\mu_1\rho_1}(x,u_1)\kappa^{-1}_{\nu_1,\sigma_1,x^a}(x,v_1)\;
\kappa_{\mu_2\rho_2}(y,u_2)\kappa^{-1}_{\nu_2\sigma_2,y^b}(y,v_2)
\nonumber\\
&& -\kappa_{\mu_1\rho_1}(x,u_1)\kappa^{-1}_{\nu_1,\sigma_1,x^a}(x,v_1)\;
\kappa_{\mu_2\rho_2,y^b}(y,u_2)\kappa^{-1}_{\nu_2\sigma_2}(y,v_2)
\nonumber\\
&& -\kappa_{\mu_1\rho_1,x^a}(x,u_1)\kappa^{-1}_{\nu_1,\sigma_1}(x,v_1)\;
\kappa_{\mu_2\rho_2}(y,u_2)\kappa^{-1}_{\nu_2\sigma_2,y^b}(y,v_2)
\nonumber\\
&&+\kappa_{\mu_1\rho_1,x^a}(x,u_1)\kappa^{-1}_{\nu_1,\sigma_1}(x,v_1)\;
\kappa_{\mu_2\rho_2,y^b}(y,u_2)\kappa^{-1}_{\nu_2\sigma_2}(y,v_2)]\;
\}
\nonumber\\
&& 
<\Omega,
A_{\rho_1}(u_1)\;A_{\sigma_1}(v_1)
A^\ast_{\rho_2}(u_2)\;A^\ast_{\sigma_2}(v_2)\Omega>
\ea
The last factor is 
\be \label{5.7}
\delta_{\rho_1\rho_2}\delta_{\sigma_1,\sigma_2}
\delta(u_1,u_2)\delta(v_1,v_2)
+
\delta_{\rho_1\sigma_2}\delta_{\sigma_1,\rho_2}
\delta(u_1,v_2)\delta(v_1,u_2)
\ee
We then integrate over $u_1,u_2,v_1,v_2$ and sum over 
$\rho_1,\rho_2,\sigma_1,\sigma_2$ in (\ref{5.6}) keeping in mind 
the symmetry properties of the kernel $\kappa$. We find
\ba \label{5.8}
&& 4 ||D(u)\Omega||^2= 
\int\; d^D x\; \int\; d^Dy\;u^a(x)\;v^b(y)\; 
\delta^{\mu_1\nu_1}\;\delta^{\mu_2\nu_2}\;
\nonumber\\ 
&& \{
\kappa^2_{\mu_1\mu_2}(x,y)\kappa^{-2}_{\nu_1\nu_2,x^a,y^b}(x,y)
+\delta_{\mu_1\nu_2}\delta_{\mu_2\nu_1}\delta_{,x^a}(x,y)\delta_{,y^b}(x,y)
-\kappa^2_{\mu_1\mu_2,y^b}(x,y)\;\kappa^{-2}_{\nu_1\nu_2,x^a}(x,y)
\nonumber\\
&&-\delta_{\mu_1\nu_2} \delta_{\mu_2\nu_1}\delta(x,y)\delta_{,x^a,y^b}(x,y)
-\kappa^2_{\mu_1\mu_2,x^a}(x,y)\kappa^{-2}_{\nu_1\nu_2,y^b}(x,y)
-\delta_{\mu_1\nu_2}\delta_{\mu_2\nu_1}\delta_{,x^a,y^b}(x,y)\delta(x,y)
\nonumber\\
&&+\kappa^2_{\mu_1\mu_2,x^a,y^b}(x,y)\;\kappa^{-2}_{\nu_1\nu_2}(x,y)
+\delta_{\mu_1\nu_2}\delta_{\mu_2\nu_1}\delta_{,x^a}(x,y)\;\delta_{,y_b}(x,y)
\}
\ea
This is the formula for general tensor fields on general $\sigma$ for general 
kernel. We now consider the simplest case of scalar fields on 
$\sigma=\mathbb{R}^D$ for a translation and reflection 
invariant kernel $\kappa(x,y)=\kappa(x-y)=\kappa(y-x)$. Then (\ref{5.8})
simplifies to 
\ba \label{5.9}
&& 4 ||D(u)\Omega||^2
=\int\; d^D x\; \int\; d^Dy\;u^a(x)\;u^b(y)\; 
\nonumber\\ 
&& \{
\kappa^2(x,y)\kappa^{-2}{,x^a,y^b}(x,y)
+\delta_{,x^a}(x,y)\delta_{,y^b}(x,y)
-\kappa^2_{,y^b}(x,y)\;\kappa^{-2}_{,x^a}(x,y)
-\delta(x,y)\delta_{,x^a,y^b}(x,y)
\nonumber\\
&& -\kappa^2_{,x^a}(x,y)\kappa^{-2}_{,y^b}(x,y)
-\delta_{,x^a,y^b}(x,y)\delta(x,y)
+\kappa^2_{,x^a,y^b}(x,y)\;\kappa^{-2}(x,y)
+\delta_{,x^a}(x,y)\;\delta_{,y_b}(x,y)
\}
\ea
The virtue of working on $\mathbb{R}^D$ is that we can invoke Fourier analysis
\be \label{5.10}
\kappa(x,y)=\int\frac{d^D k}{(2\pi)^D} \hat{\kappa(k)} e^{ik(x-y)},\;\;
\hat{\kappa}(k)=\hat{\kappa}(-k)=
\hat{\kappa}(k)^\ast
\ee
to work out the product of distributions in (\ref{5.9}). 
\ba \label{5.11}
&& 4 ||D(u)\Omega||^2
=\int\; d^D x\; \int\; d^Dy\;
(2\pi)^{-2D}\;\int\;d^D k\int\;d^D l\;
u^a(x)\;u^b(y)\; 
\int\;\frac{d^D l}{(2\pi)^D}\; e^{i[(x-y)(k+l)}
\nonumber\\ 
&& \{
\hat{\kappa}^2(k)\hat{\kappa}^{-2}(l)\; l_a\; l_b
+k_a l_b
-\hat{\kappa}^2(k)\hat{\kappa}^{-2}(l)\; k_b\; l_a
-l_a l_b
-\hat{\kappa}^2(k)\hat{\kappa}^{-2}(l)\; k_a\; l_b
-k_a k_b
+\hat{\kappa}^2(k)\hat{\kappa}^{-2}(l)\; k_a\; k_b
+k_a l_b
\}
\nonumber\\
&=& 
(\pi)^{-2D}\;\int\;d^D k\int\;d^D l\;
[\hat{u}^a(k+l)]^\ast\;\hat{u}^b(k+l)\; 
\nonumber\\ 
&& \{[\frac{\hat{\kappa}(k)}{\hat{\kappa}(l)}]^2-1\} [k_a-l_a][k_b-l_b]
\nonumber\\
&=&\frac{1}{2} 
(2\pi)^{-2D}\;\int\;d^D k\int\;d^D l\;
|\hat{u}^a(k+l)(k_a-l_a)|^2
\{
[\frac{\hat{\kappa}(k)}{\hat{\kappa}(l)}]^2+
+[\frac{\hat{\kappa}(l)}{\hat{\kappa}(k)}]^2-2
\}
\nonumber\\
&=&\frac{1}{2} 
(2\pi)^{-2D}\;\int\;d^D k\int\;d^D l\;
|\hat{u}^a(k+l)(k_a-l_a)|^2\;
\;[\frac{\hat{\kappa}(k)}{\hat{\kappa}(l)}-1]^2
\ea
where we introduce the Fourier transform of the vecor field and replaced 
the integral by half the sum of the integral and the integral with $k,l$ 
interchanged. As expected (\ref{5.11}) vanishes for the identity kernel 
$\hat{\kappa}=1$. We introduce $k^\pm=k\pm l$ then
\be \label{5.12}
32\;(2\pi)^{2D}||D(u)\Omega||^2
=\int\;d^D k^+\int\;d^D k^-\;
|\hat{u}^a(k^+)k_a^-|^2\;
\;[\frac{\hat{\kappa}([k^+ + k^-]/2)}{\hat{\kappa}([k^+ - k^-]/2)}-1]^2
\ee
We can use the decay or support properties of $u^a$ to make $\hat{u}^a$ 
decrease rapidly or even choose it 
of compact momentum support. Then the integral over $k^+$ converges. 
To make the 
integral over $k^-$ converge, it is sufficient 
to grant that the bracket term decays faster than 
$||l||^{D+2+1}$. There are an infinite number of kernels that accomplish
this, for example 
\be \label{5.13}
\hat{\kappa}^2(k):=c+\frac{1}{1+||k||^{2N} L^{2N}}
\ee
where $c$ is a positive 
constant and $L$ is some length scale since then the bracket term becomes
\ba \label{5.14}   
&& [\frac{\hat{\kappa}([k^+ + k^-]/2)}{\hat{\kappa}(k^+ - k^-]/2)}-1]^2
=[\hat{\kappa}((k^- -k^-)/2)]^{-2}\;
[\hat{\kappa}((k^- -k^-)/2)+\hat{\kappa}((k^- +k^-)/2)]^{-2}\;
\times 
\nonumber\\
&& [\hat{\kappa}^2((k^- +k^+)/2)-\hat{\kappa}^2((k^- +k^+)/2)]^2
\nonumber\\
&=& [\hat{\kappa}((k^- -k^-)/2)]^{-2}\;
[\hat{\kappa}((k^- -k^-)/2)+\hat{\kappa}((k^- +k^-)/2)]^{-2}\;
\times\nonumber\\
&& [1+[(L/2)^2||k^- - k^+||^2]^N]^{-2}\;
[1+[(L/2)^2||k^- + k^+||^2]^N]^{-2}\;
\times \nonumber\\
&& [(L/2)^{2N} ||k^- + k^+||^{2N}-(L/2)^{2N} ||k^- - k^+||^{2N}]^2
\ea
If $\hat{u}$ has compact momentum support we can assume that $k^+$ is finite 
and consider the limit $||k^-||\to \infty$ of (\ref{5.14}) at fixed 
$k^+$ (formally we integrate first $k^-$ and then  $k^+$). Then the first 
two factors in (\ref{5.14}) approach a constant of order unity, 
the third and fourth decay as 
$||k^-||^{-8N}$ while the fifth factor grows as $||k^-||^{4(N-1)+2}$ 
when $||k^-||\to\infty$. Hence we need $4N+2>D+3$ i.e. $N>(D+1)/4$. For 
$D=1$ we can choose e.g. any integer $N\ge 1$ while for $D=3$ we can choose 
any integer $N\ge 2$.\\ 
\\
This shows that there exist kernels on $\sigma=\mathbb{R}^D$ 
different from the identity for which 
$D(u)$ is a well defined operator on the corresponding Fock space for scalar
fields. For arbitrary tensor fields this still holds if we choose 
the above scalar kernel multiplied by the tensor space identity i.e.
$\kappa_{\mu\,\nu}(x,y)=\kappa(x,y)\delta_{\mu\nu}$. Thus we have established:
\begin{Proposition} \label{prop5.1} ~\\
Let $\kappa$ be a translation and reflection invariant positive, invertible
integral kernel for scalar fields on $\sigma=\mathbb{R}^D$ and 
$\kappa_{\mu\nu}=\kappa\delta_{\mu\nu}$ the corresponding kernel for the 
tensor fields of a given type. Then there are infinitely many choices of 
$\kappa$ such that $D(u)$ is a densely defined operator on the corresponding 
Fock space.
\end{Proposition} 
Note that for different choices of $c$ the representations (\ref{5.13}) are 
unitarily inequivalent because the Hilbert-Schmidt condition asks
that 
$\frac{\hat{\kappa}^2-[\hat{\kappa}']^2}{\hat{\kappa}^2+[\hat{\kappa}']^2}$
decays at least as $||k||^{-(D+1)}$.  

We 
now consider the scaling operator $S(u)$ for the same example in the 
scalar case. Abusing 
the notation we write it as $S(\lambda),\;\lambda=u^a_{,a}$. Then 
its normal ordered form is 
\ba \label{5.15}
S(\lambda) &=&
:
\frac{i}{2}\int \; d^Dx\;
\lambda\;(\kappa[A-A^\ast])(\kappa^{-1}[A+A^\ast]) 
:
\nonumber\\
&=&
\frac{i}{2}\int \; d^Dx\;
\lambda\;\{
-[\kappa A^\ast]\;[\kappa^{-1} A]
+[\kappa^{-1} A^\ast]\;[\kappa A]
+[\kappa A]\;[\kappa^{-1} A]
-[\kappa A^\ast]\;[\kappa^{-1} A^\ast]
\}
\ea
and applied to the vacuum we find 
\ba \label{5.16}
&& 4\; ||S(\lambda)\Omega||^2
=\int\; d^Dx\;\int\; d^Dy \lambda(x)\lambda(y)
\int\; d^{2D}u \int\; d^{2D}v\;
\kappa(x,u_1)\kappa^{-1}(x,v_1)
\kappa(y,u_2)\kappa^{-1}(y,v_2)
\times \nonumber\\
&&
<\Omega,A(u_1) A(v_1) A^\ast(u_2) A^\ast(v_2)\Omega
\nonumber\\
&=&\int\; d^Dx\;\int\; d^Dy \lambda(x)\lambda(y)
\int\; d^{2D}u \int\; d^{2D}v\;
\kappa(x,u_1)\kappa^{-1}(x,v_1)
\kappa(y,u_2)\kappa^{-1}(y,v_2)
\times \nonumber\\
&&
<\Omega,A(u_1) A(v_1) A^\ast(u_2) A^\ast(v_2)\Omega>
\nonumber\\
&=&\int\; d^Dx\;\int\; d^Dy \lambda(x)\lambda(y)
\int\; d^{2D}u \int\; d^{2D}v\;
[\kappa(x,u_1)\kappa^{-1}(x,v_1)
\kappa(y,u_2)\kappa^{-1}(y,v_2)]\;
\times \nonumber\\
&&
[\delta(u_1,u_2)\delta(v_1,v_2)+
\delta(u_1,v_2)\delta(v_1,u_2)]
\nonumber\\
&=&\int\; d^Dx\;\int\; d^Dy \lambda(x)\lambda(y)
[\kappa^2(x,y)\kappa^{-2}(x,y)+\delta(x,y)^2]
\nonumber\\
&\ge& \delta(0,0) ||\lambda||_{L_2}^2
\ea
which shows that the integral over $k^-$ has no chance to converge 
even for a general kernel. Thus we have established:
\begin{Proposition} \label{prop5.1a} ~\\
Let $\kappa$ be any
 positive, invertible
integral kernel for scalar fields on $\sigma=\mathbb{R}^D$ 
Then for no choice of $\kappa$ is 
$S(u)$ a densely defined operator on the corresponding 
Fock space.
\end{Proposition} 
Thus, for translation and reflection invariant kernels on $\mathbb{R}^D$ 
the operators $D(u), S(u)$ are not separately welll defined on the 
corresponding Fock space for scalar fields. 
On the other hand, it is well known that 
the combination $C(u)=D(u)+S(u)$ can become a well defined operator 
for suitable kernels in $D=1$ e.g. in string theory (ST) or parametrised field
theory (PFT) 
\cite{24,25}. To see how this comes about we recall that in these 
D=1 theories one exploits the fact that tensors of type $(A,B,w)$ are 
scalar densities of weight $w+B-A$. Therefore the objects $\pi\pm \phi'$ 
with $(.)'=d/dx(.)$ that one uses there to define the annihilation 
and creation operators are in 
fact densities of 
weight one so that annihilators transform into annihilators under spatial
diffeomorphisms which is one of the reasons why $D=1$ is special:
In no other dimension can one change the density weight background 
independently, that is, only using the differential structure, without 
changing the tensor type $A,B$ of a field. 

We now repeat the calculation above directly for the combination 
$C(u)=D(u)+S(u)$. This gives with exactly the same manipulations 
\ba \label{5.17}
&& 32 (2\pi)^{2D} ||C(u)\Omega||^2
= 8\int\; d^Dk\; \int\; d^Dl\; 
\hat{u}^a(k+l)^\ast
\hat{u}^b(k+l)
[\hat{\kappa}(k)^2\hat{\kappa}(l)^{-2} \;l_a l_b+k_a l_b]
\nonumber\\
&=&4 \int\; d^Dk\; \int\; d^Dl\; 
\hat{u}^a(k+l)^\ast
\hat{u}^b(k+l)
[\hat{\kappa}(k)^2\hat{\kappa}(l)^{-2} \;l_a l_b+k_a l_b
+\hat{\kappa}(l)^2\hat{\kappa}(l)^{-2} \;k_a k_b+\l_a k_b]
\nonumber\\
&=& 4\int\; d^Dk\; \int\; d^Dl\; 
\hat{u}^a(k+l)^\ast
\hat{u}^b(k+l)
[\hat{\kappa}(k)\hat{\kappa}(l)^{-1} \;l_a+
\hat{\kappa}(l)\hat{\kappa}(k)^{-1} \;k_a]\;
[\hat{\kappa}(k)\hat{\kappa}(l)^{-1} \;l_b+
\hat{\kappa}(l)\hat{\kappa}(k)^{-1} \;k_b]\;
\nonumber\\
&=& 4\int\; d^Dk\; \int\; d^Dl\; 
|\hat{u}^a(k+l)\;
[\hat{\kappa}(k)\hat{\kappa}(l)^{-1} \;l_a+
\hat{\kappa}(l)\hat{\kappa}(k)^{-1} \;k_a]|^2
\ea
Consider first the case $D=1$ and the choice $\kappa(k)=|k|^{1/2}$ which 
corresponds to the Fock representation of a massless Klein Gordon field in 
two spacetime dimensions and which corresponds to the choice that is usually
made in ST and PFT.
Then 
\be \label{5.18}
\hat{\kappa}(k)\hat{\kappa}(l)^{-1} \;l+
\hat{\kappa}(l)\hat{\kappa}(k)^{-1} \;k=
\hat{\kappa}(k)\hat{\kappa}(l)[{\rm sgn}(l)+{\rm sgn}(k)]
\ee
which vanishes unless $k,l$ have the same sign. However, if $k,l$ have the 
same sign and $\hat{u}(k+l)$ vanishes for $|k+l|>k_0$ (compact momentum 
support) then the integral is confined to a subset of the 
union of the sets defined by $0\le k,l\le k_0$ or $-k_0\le k,l\le 0$ which 
has finite Lebesgue volume. Therefore (\ref{5.17}) converges 
in this case. 

Consider still in $D=1$ the more general choice 
\be \label{5.18a}
\kappa(k)=|k|^{1/2}(1+f(k)),\; f(k)=\frac{1}{1+(Lk)^{2N}}
\ee
for $N$ to be specified. The inverse kernel has Fourier transform 
$\hat{k}^{-1}(k)=[\hat{\kappa}(k)]^{-1}$ 
which diverges at $k=0$ as $|k|^{-1/2}$. Therefore for 
the corresponding Weyl algebra based on (\ref{5.18a}) to be 
well defined we consider test functions whose Fourier transform vanises 
rapidly both at $|k|=0,\infty$ \cite{27}. In the compact case $\sigma=T^D$
one can take care of the zero mode separately, see e.g. \cite{26}
for an example. Then 
\be \label{5.18b}
[\hat{\kappa}(k)\hat{\kappa}(l)^{-1} \;l
+\hat{\kappa}(l)\hat{\kappa}(k)^{-1} \;k]^2
=|k \;l|\;
[\frac{1+f(k)}{1+f(l)}{\rm sgn}(l)+\frac{1+f(l)}{1+f(k)}{\rm sgn}(k)]^2
\ee
If $k,l$ have the same sign, as the support of $\hat{u}$ enforces 
$|k+l|<k_0$, then $0\le |k|,|l|\le k_0$ and the integral converges in these 
two quadrants of the $k,l$ plane. If $k,l$ have different sign then 
\be \label{5.18b1}
[\hat{\kappa}(k)\hat{\kappa}(l)^{-1} \;l
+\hat{\kappa}(l)\hat{\kappa}(k)^{-1} \;k]^2
=|k \;l|\;
[\frac{1+f(k)}{1+f(l)}-\frac{1+f(l)}{1+f(k)}]^2
\ee
where we used reflection invariance $f(k)=f(-k)$. We now introduce 
$m=k+l\in [-k_0,k_0]$ as a new integration variable, keep $l\in \mathbb{R}$ 
as an 
independent integration variable and thus have $k=m-l$. In the sector 
$k<0,l>0$ this requires $l>m$ and in the sector   
$k>0,l<0$ this requires $l<m$. Then at fixed $m$, the factor $|l k|$ grows 
quadratically with $|l|$. On the other hand 
\ba \label{5.18c}
&& [\frac{1+f(k)}{1+f(l)}-\frac{1+f(l)}{1+f(k)}]^2
=[1+f(k)]^{-2}\;[1+f(l)]^{-2}
[(1+f(k))^2-(1+f(l))^2]^2
\nonumber\\
&=& [1+f(k)]^{-2}\;[1+f(l)]^{-2}
[f(k)+f(l))]^2\;[f(k)-f(l)]^2
\nonumber\\
&=& [\frac{f(k)+f(l))}{[1+f(k)]\;[1+f(l)]}]^2
[\frac{(Ll)^{2N}-(Lk)^{2N}}{[1+(Lk)^{2N}][1+(Ll)^{2N}]}]^2
\ea
The first factor approaches unity as $|l|\to \infty$ while the second decays 
as $|l|^{2(2N-1-4N)}=|l|^{-(4N+2)}$. Hence in $D=1$ the integral converges for 
$2+2\le 4N+2$ i.e. e.g. integer $N\ge 1$.\\  
\\
This mechanism does not exist in higher 
dimensions: Consider in any $D>1$ the support of 
$\hat{u}^a(m),\;a=1,..,D$ to be in the ball $||m||\le k_0$. 
We introduce again $m,l$ as independent integration variables and
thus have $k=m-l$. Furthermore we decompose for $||m||>0$ (the set 
$||m||=0$ has Lebesgue measure zero) $\hat{u}(m)=\hat{u}_1(m)\; n(m)
=\hat{u}_\perp(m)$ and $l=l_1 n(m)+l_\perp$ where $n(m)=m/||m||$ is 
the unit vecor in direction $m$ and 
$m\cdot \hat{u}_\perp(m)=m\cdot l_\perp=0$. This corresponds to a choice 
of frame with $m$ in direction $a=D$ and the corresponding directions in the
plane orthogonal it. Then the Lebesgue measure becomes 
$d^D k \; d^Dl=d^D m\; dl_1\; d^{D-1}l_\perp$ and the integrand 
becomes 
\ba \label{5.19}
&& |\hat{u}^a(m)\;
[\hat{\kappa}(m-l)\hat{\kappa}(l)^{-1} \;l_a+
\hat{\kappa}(l)\hat{\kappa}(m-l)^{-1} \;(m-l)_a]|^2
\\
&=& |\hat{u}^a(m)\;
[(\hat{\kappa}(m-l)\hat{\kappa}(l)^{-1}-\hat{\kappa}(l)\hat{\kappa}(m-l)^{-1})
\;l_a
+\hat{\kappa}(l)\hat{\kappa}(m-l)^{-1} \;m_a]|^2
\nonumber\\
&=& 
|
[(\hat{\kappa}(m-l)\hat{\kappa}(l)^{-1}-\hat{\kappa}(l)\hat{\kappa}(m-l)^{-1})
\;(\hat{u}_1(m) l_1+\hat{u}_\perp(m)\cdot l_\perp)
+||m||\;\hat{\kappa}(l)\hat{\kappa}(m-l)^{-1} \;u_1(m)]
|^2
\nonumber\\
&=& 
|
[(\hat{\kappa}(m-l)\hat{\kappa}(l)^{-1}-\hat{\kappa}(l)\hat{\kappa}(m-l)^{-1})
\;(\hat{u}_1(m) l_1+\hat{u}_\perp(m)\cdot l_\perp)
+||m||\;\hat{\kappa}(l)\hat{\kappa}(m-l)^{-1} \;u_1(m)]
|^2
\nonumber
\ea
We perform the integral at fixed $m$ independently over $l_1,l^a_\perp\in 
\mathbb{R}^D$. It is required to converge for all choices of 
$\hat{u}_1(m),\hat{u}_\perp(m)$. Consider first $u_1(m)=0$ (transversal
vector field). Then we conclude 
that 
\be \label{5.20}
[(\hat{\kappa}(m-l)\hat{\kappa}(l)^{-1}
-\hat{\kappa}(l)\hat{\kappa}(m-l)^{-1})]^2
\;(\hat{u}_\perp(m)\cdot l_\perp)^2
\ee
for all $\hat{u}_\perp(m)$ must decay as faster than $||l||^{-D}$ i.e. 
\be \label{5.21}
[\hat{\kappa}(m-l)\hat{\kappa}(l)^{-1}
-\hat{\kappa}(l)\hat{\kappa}(m-l)^{-1}]^2
\ee
must decay faster than $||l||^{-(D+2)}$. We saw above that such $\kappa$ exist. 
Now choose $\hat{u}_\perp(m)=0$ (longitudinal vector field). Then 
\be \label{5.22}
[(\hat{\kappa}(m-l)\hat{\kappa}(l)^{-1}-\hat{\kappa}(l)\hat{\kappa}(m-l)^{-1})
\;(\hat{u}_1(m) l_1)
+||m||\;\hat{\kappa}(l)\hat{\kappa}(m-l)^{-1} \;u_1(m)]^2
\ee
must decay daster than $||l||^{-D}$ for any $\hat{u}_1(m)$. 
This is already the case for the first term in the square bracket. However
if (\ref{5.21}) decays as $||l||\to \infty$ then 
$\frac{\hat{\kappa}(m-l)}{\hat{\kappa}(l)}$ approaches a constant value. 
By expanding 
\be \label{5.23}
\kappa(m-l)=\kappa(l-m)=\kappa(l)-[\partial_a\kappa](l)\;m^a+..=
\kappa(l)(1+O(||m||/||l||)
\ee
we see that this constant must be unity. Thus we conclude that the integral 
over (\ref{5.19}) diverges. Note that the argument fails for $D=1$ because 
in that case no transverse vector field components exist. 
We also see again that 
$C(u)=D(u)$ is well defined when $\partial_a u^a=0$, that is $S(u)=0$.
\begin{Proposition} \label{prop5.2} ~\\
Let $\kappa$ be a translation and reflection invariant positive invertible
integral kernel on $\sigma=\mathbb{R}^D$ defining a Fock representation
for scalar fields.\\
i. In D=1 there are infinitely many such $\kappa$ such that $C(u)$ is densely 
defined in the Fock representation for all $u$.\\
ii. In $D>1$ there are no such $\kappa$ such that $C(u)$ is densely 
defined in the Fock representation for all $u$.\\
iii. In $D>1$ there are infinitely many such $\kappa$ such that $C(u)$ is 
densely defined in the Fock representation for all transverse $u$.
\end{Proposition}
Note that had we replaced in (\ref{5.18a}) the function $1+f(k)$ by 
$c+f(k)$ for positive $c$ then proposition \ref{prop5.2} is still 
valid but for different $c$ the Fock representations are again unitarily 
inequivalent. 

The examples we have given are kernels which are not only translation and 
reflection invariant but even rotation invariant. These symmetries are 
of course geared to $\sigma=\mathbb{R}^D$ and the question arises what 
happens for more general topologies. First we note that the three 
propositions immediately generalise to the D-torus $T^D$ with periodic
vector fields or 
the solid D-ball with vector fields of compact support since Fourier 
transform is still possible. One just has to replace momentum integrals 
by momentum sums, however, the convergence conditions remain the same.
In other words, IR regulators do not help to make $C(u)$ a densely 
defined operator. Next, on a more general manifold we expect that the 
conclusions are qualitatively unchanged because we could consider vector 
fields with compact support in a coordinate patch that contains a solid 
D-ball. As concerns more general tensor fields one would need to study 
the scaling contribution in more detail but we also do not expect 
qualitatively different results because common to all of them is the 
contribution $u^a_{,a} p_\mu q^\mu$ which is the source of trouble
in the scalar case and whose tensorial structure is different from 
the remaining terms in $S(u)$. Concerning more general kernels $\kappa$ 
without any symmetries the result of proposition \ref{prop5.1} valid for any 
kernel shows that even in that general case at least $C(u),S(u)$ cannot 
be separately densely defined and the special role that $D=1$ has played 
in the case of symmetric $\kappa$ suggests that for $D>2$ the SDC $C(u)$ 
is also not densely defined for general non-symmetric $\kappa$.\\
\\
We close this section by remarking that for the Fock representations of 
the class (\ref{5.18a}) the central charge vanishes and there is 
no central extension of the SDA, i.e. in the D=1 case the 
$C(u)$ generate an anomaly free algebra. This is because the  
central charge even in the standard case $\hat{\kappa}(k)=|k|^{1/2}$ 
vanishes. However, in ST and PFT one considers the Fock rep. not using
$\hat{\kappa}=|k|^{1/2}$ but rather 
$\hat{\kappa}_\pm=|k|^{1/2}(1\pm {\rm sgn}(k)$. In this case one finds 
that $C(u)=C_+(u) - C_-(u)$ with $[C_\delta(u),C_{\delta'}(v)]=
\delta\; \delta_{\delta\delta'}[iC_\delta(-[u,v])+i \sigma(u,v)\cdot 1]$.
It follows that $[C(u),C(v)]=iC(-[u,v])$ i.e. the spatial diffeomorphism 
constraint $C(u)$ of PFT itself is without
anomaly, just $C_\pm(u)$ are anomalous, therefore also the 
Hamiltonian constraint $H(u)=C_+ + C_-$ is anomalous. 
See \cite{26} for an explanation
using the notation used here.

To see this more formally we note
that the normal ordered $C(u)$ differs from the natural ordering 
$\int\;dx\; u\pi\phi'$ by a constant as $C$ is bilinear
in the annihilation and creation operators. As the algebra is non-anomalous 
in the natural ordering (just using algebraic operations keep track of the 
order, the remaing manipulations are the same as in the classical Poisson
bracket calculation) the algebra in the normal ordering can have at most 
a central anomaly i.e. $[C(u),C(v)]=i(C(-[u,v])+\sigma(u,v))$ with 
real valued $\sigma(u,v)=-\sigma(v,u)$ since $C(u)$ is a symmetric operator.
It follows that 
\ba \label{5.24}
&& i\sigma(u,v)=
<\Omega,[C(u),C(v)]\Omega>     
=<C(u)\Omega,C(v)\Omega>-<C(v)\Omega,C(u)\Omega>
\nonumber\\
&=& 
<C'(u)\Omega,C'(v)\Omega>-<C'(v)\Omega,C'(u)\Omega>
\ea
where $C'(u)=-i\int\;\;dx\; u\;[\kappa A^\ast]\;[\kappa^{-1} A^\ast]'$
is the piece of $C(u)$ which conatins no annihilator. Then the computation
is the same as that for $||C(u)\Omega||^2$ with the result
\be \label{5.25}
i\sigma(u,v)=(4\pi)^{-2}\int_{\mathbb{R}}\;dk\;\int_{\mathbb{R}} 
dl\; \omega(k+l)
|\frac{\hat{\kappa}(k)}{\hat{\kappa}(l)} \; l 
+\frac{\hat{\kappa}(l)}{\hat{\kappa}(k)} \; k|^2,\;\;
\omega(m)=\hat{u}(m)^\ast v(m)-v(m)^\ast u^\ast
\ee
Decomposing the integral over $k,l$ into the four quadrants we find 
\ba \label{5.26}
&& i\sigma(u,v)=(4\pi)^{-2}\int_{\mathbb{R}_+}\;dk\;\int_{\mathbb{R}_+} 
dl\; 
\times \nonumber\\
&& \{
[\omega(k+l)+\omega(-[k+l])]
|\frac{\hat{\kappa}(k)}{\hat{\kappa}(l)} \; l 
+\frac{\hat{\kappa}(l)}{\hat{\kappa}(k)} \; k|^2\;\;
+
[\omega(-k+l)+\omega(-[-k+l])]
|\frac{\hat{\kappa}(k)}{\hat{\kappa}(l)} \; l 
-\frac{\hat{\kappa}(l)}{\hat{\kappa}(k)} \; k|^2,\;\;
\}
\ea
Hoever $\hat{u}(-m)=\hat{u}(m)^\ast$ for real $u$ implies $\omega(-k)
=-\omega(k)$. Thus $\sigma(u,v)=0$.

Note that the span of Fock states is not an invariant domain in the case 
$D=1$, even
$C(u)\Omega$ is not a finite linear combination of states of the 
form $A^\ast(f_1)..A^\ast(f_N)\Omega$. Also we have only checked 
that $||C(u)\Omega||<\infty$ for $D=1$ and kernels of the form (\ref{5.18a}).
To check that $C(u)$ is densely defined on all Fock states we check 
the equivalent statement that $||C(u)\;w(f,0)\Omega||<\infty$ with 
$w(f)=\exp(i<f,\phi>)$. We have with $F=2^{-1/2}\kappa^{-1} f$ 
\be \label{5.27}
A\;w(f)=w(f)\;[A+iF],\;
A^\ast\;w(f)=w(f)\;[A^\ast-iF]
\ee
thus
\ba \label{5.28}
&& w(f)^{-1}\;C(u)w(f)\Omega=
\frac{i}{2}\int \; dx\; u\;
\times\nonumber\\
&& \{
[\kappa(A+iF)]\;[\kappa^{-1}(A+iF)]'
-[\kappa(A^\ast-iF)]\;[\kappa^{-1}(A^\ast-iF)]'
+[\kappa^{-1}(A^\ast-iF)]'\;[\kappa(A+iF)]
\nonumber\\
&& -[\kappa(A^\ast-iF)]\;[\kappa^{-1}(A+iF)]'
\}\;\Omega
\nonumber\\
&=&
\frac{i}{2}\int \; dx\; u\;
\{
-[\kappa F]\;[\kappa^{-1} F]'
-[\kappa A^\ast]\;[\kappa^{-1}A^\ast]'
+i[\kappa A^\ast]\;[\kappa^{-1} F]'
+i[\kappa F]\;[\kappa^{-1}A^\ast]'
\nonumber\\
&& +[\kappa F]\;[\kappa^{-1}F ]'
+i[\kappa^{-1}A^\ast]'\;[\kappa F]
+[\kappa^{-1}F]'\;[\kappa F]
-i[\kappa A^\ast]\;[\kappa^{-1} F]'
-[\kappa F]\;[\kappa^{-1} F]'
\}\;\Omega
\nonumber\\
&=&
\frac{i}{2}\int \; dx\; u\;
\{
-[\kappa A^\ast]\;[\kappa^{-1}A^\ast]'
+i[\kappa A^\ast]\;[\kappa^{-1} F]'
+i[\kappa F]\;[\kappa^{-1}A^\ast]'
+i[\kappa^{-1}A^\ast]'\;[\kappa F]
-i[\kappa A^\ast]\;[\kappa^{-1} F]'
\}\;\Omega
\nonumber\\
&=&
\frac{i}{2}\int \; dx\; u\;
\{
-[\kappa A^\ast]\;[\kappa^{-1}A^\ast]'
+2i[\kappa F]\;[\kappa^{-1}A^\ast]'
\}\;\Omega
\nonumber\\
&=&
C(u)\Omega-2^{-1/2}\;<\kappa^{-1}[u\cdot f]',A>^\ast \Omega=:C(u)\Omega
+\psi
\ea
We have for the Fourier transform 
\be \label{5.29}
(\kappa^{-1}[u\cdot f]')^\wedge(k)=ik\hat{\kappa}^{-1}(k)\int\;dl\;
\hat{u}(k-l)\hat{f}(l)
\ee
If both $u,f$ have compact momentum support in say $[-k_0,k_0]$ then 
the integral in (\ref{5.24}) has momentum support in $[-2k_0,2k_0]$ and 
since $k \hat{\kappa}^{-1}(k)$ vanishes at $k=0$ and is bounded on bounded 
intervals we have $||g||_{L_2}<\infty$, thus 
$||\psi||^2=2^{-1} ||g||^2_{L_2}<\infty$.
Since we have already shown $||C(u)\Omega||<\infty$ it follows from the 
estimate 
\be \label{5.30}
||C(u)w(f)\Omega\||=||w(f)^{-1}C(u)w(f)\Omega\||\le ||C(u)\Omega||+
||\psi||
\ee
where we used unitarity of $w(f)$. Thus $C(u)$ is densely defined on the Fock
space.

\section{Constructive QFT methods for the SDA}
\label{s6}

As Fock representations of the non-volume preserving SDA meet obstacles as 
detailed in the previous section, the question arises whether one 
can construct the SDA in other representations. We therefore investigate 
in this section the possibility to employ
constructive QFT (CQFT) methods to construct $C(u)$ in a new representation
which is not a Fock representation. This method was invented 
in \cite{15} and has lead to a successful construction of $\Phi^4_3$ 
theory where one first works in finite volume and in a second step 
takes the thermodynamic limit. In the first subsection we sketch 
the method of \cite{15} and in the second we consider its application 
to the SDA.

\subsection{Dressing transformations}
\label{s6.1}

Suppose we are given a classical Hamiltonian $H$ which maybe split as 
$H=H_0+V$ where $H_0$ can be defined as symmetric 
operator on a Hilbert space ${\cal H}_0$ with dense domain 
${\cal D}_0\subset {\cal H}_0$ while 
$V$ is ill-defined as an operator but is at least a symmetric quadratic 
form with dense domain ${\cal D}_0$. 
Thus while $H_0\;D_0\subset {\cal H}_0$ (often even $H_0 {\cal D}_0
\subset {\cal D}_0$ is an invariant domain) we 
have $V\;{\cal D}_0\not\subset {\cal H}_0$. However, we have that both
$<\psi_0,H_0^2\psi_0'>_0,\;<\psi_0,V\psi_0'>_0$ exist for 
$\psi_0,\psi_0'\in {\cal D}_0$ where $<.,.>_0$ denotes the scalar product 
on ${\cal H}_0$. We intend to construct a new Hilbert space $\cal H$ with 
new dense domain ${\cal D}\subset {\cal H}$ such that $H$ is a 
symmetric operator on $\cal H$ with dense domain $\cal D$, i.e. 
$H {\cal D}\subset {\cal H}$ i.e $<\psi,H^2\psi'>$ exists for 
all $\psi,\psi'\in {\cal D}$ where $<.,.>$ denotes the inner product 
on ${\cal H}$.    

To have a concrete picture in mind consider the important case of 
practical importance that $H_0$ (``free part'')
can be defined on a Fock space ${\cal H}_0$
with annihilation operators $<f,A>$ for $f$ a smearing 
function in the 1-particle Hilbert space $L_2(\sigma)$ annihilating 
the Fock vacuum $\Omega_0$. We introduce an ONB $e_I$ of the 1-particle HS 
$L_2(\sigma)$ where I is a ``mode label'' 
(e.g. for 
$\sigma=\mathbb{R}^D$ we may pick the Hermite basis labelled by 
$I=(I_1,..,I_D)\in \mathbb{N}_0^D$ 
and on $\sigma=T^D$ we may use the Fourier 
basis labelled by $I=(I_1,..,I_D)\in \mathbb{Z}^D$). The corresponding
orthonormal Fock basis of the Fock space ${\cal H}_0$ is then labelled
by occupation numbers $n=\{n_I\}_I$ 
\be \label{6.1a}
b_n=\prod_I\; \frac{[A_I^\ast]^{n_I}}{\sqrt{n_I!}}\;\Omega,\;n_I\in 
\mathbb{N}_0,\;
||n||:=\sum_I\; n_I<\infty
\ee
where $A_I=<e_I,A>_{L_2}$. The space ${\cal D}_0$ is the finite linear span
of the $b_n$.

Let $S_M,\;M\in \mathbb{N}$ be finite sets of modes $I$ which are nested 
in the sense $S_M \subset S_{M+1}$ and such that $S_\infty$ is the complete 
index set (e.g. in the above examples on $\sigma=\mathbb{R}^D, T^D$
we could consider $||I||:=|I_1|+..+|I_D|$ and $S_M=\{I:\;||I||\le M\}$). 
Thus $M$ serves as a 
``mode cut-off''. Note that elements of the Fock basis 
have a built-in ``occupation number cut-off'', 
that is $||n||<\infty$ for every 
$n$. 

Let ${\cal H}_{0,M}$ be the subspace obtained as the completion of the span 
${\cal D}_{0,M}$
of those $b_n$ with $n_I=0$ for all $I\not\in S_M$. We expand the classical
expression for $H$ into annihilation and creation operators
and normal order it. Thus we obtain for both $H_0,V$ linear 
combinations of expressions of the form with $0\le k\le N\in \mathbb{N}_0$. 
\be \label{6.1}
B_{N,k}:=\sum_{I_1,..,I_N} \; B_{N,k}\;^{I_1..I_N}\; 
A_{I_1}^\ast\;..\;A_{I_{N-k}}^\ast\;A_{I_{N-k+1}}\;..\;A_{I_N}
+c.c.
\ee
which is a monomial with $N-k$ creators and $k$ annihilators with
certain, finite complex coefficients $B_{N,k}^{I_1..I_N}$. It depends on 
the 
index support and decay of those as $I_1,..,I_N\in S_M$ with $M\to \infty$
that decides about whether $B_{N,k}$ is either 
an operator or merely a quadratic form. We can enforce the 
operator qualifier by restricting 
all $I_1,..,I_N$ to $S_M$ (mode cut-off)
thereby obtaining operators $B^M_{N,k}$ and thus
$H^M$ is densely defined on all of ${\cal D}_0$. Restricted to ${\cal D}_{0,M}$ 
the operator $H_M$ has range in ${\cal H}_{0,M}$. 

Suppose now that we find an invertible operator $T_M:\; {\cal D}_0\to 
{\cal D}_0$ such that 
\be \label{6.2}
\hat{H}_M:=T_M^{-1}\; H_M\; T_M
\ee
has a strong limit $s-\lim_{M\to \infty}\hat{H}_M=\hat{H}$ with dense {\it 
and invariant} domain ${\cal D}_0$. To construct such $T_M$ we make the Ansatz 
\be \label{6.2b}
T_M=\exp(G_M)
\ee
where $G_M$ itself can be expanded as in (\ref{6.1}) with expansion 
coefficients $E_{N,k}^{I_1 .. I_N}$ and all indices 
are resricted to $S_M$. Then we work out $e^{-G_M} A_I e^{G_M}$ using 
the commutation relations and try to adjust the coefficients $E$ to the 
coefficients $B$ such that $\hat{H}$ has the desired property. Typically
one tries to achieve that after applying $T_M$, contributions 
(\ref{6.1}) with large $N$ 
are either deleted or have much improved decay behaviour in all 
indices, preferrably 
of compact fixed index support $M_0$ which is 
eventually independent of $M$ as $M$ increases or that the index range 
of each creator is compactly supported around one of the indices of at least 
one creator in an $M$-independent manner. 
This step may require renormalisation, i.e. we may in fact 
modify the classically given coefficients $B_{N,k}$ to $g_{N,k} B_{N,k}$ 
for some numbers $g_{N,k}$, i.e. the renormalised Hamiltonian has the same 
type of terms $B_{N,k}$ as the ``classical'' Hamiltonian but renormalised 
``coupling constants'' $g_{N,k}$.    

If this is possible, consider for $\psi_0,\psi_0'\in {\cal D}_0$
\be \label{6.3}
<T_M \psi,\; T_M\psi'>_M:=
\frac{<T_M \psi_0,T_M \psi'>}{<T_M \Omega_0,T_M \Omega_0>}
\ee
and we check for which $\psi_0,\psi_0'\in {\cal D}_0$ the limit
\be \label{6.4}
<T \psi_0, T\psi_0'>:=\lim_{M\to\infty}\;<T_M \psi_0,\; T_M\psi_0'>_M
\ee
exists. Denote that subspace by $\hat{{\cal D}}_0$ and define ${\cal D}$ as 
the space of vectors 
$T\psi_0:=\lim_{M\to \infty} T_M\psi_0,\; \psi_0\in \hat{{\cal D}}_0$. 
Suppose that 
$\hat{H}$, which has dense invariant domain ${\cal D}_0$,
even has dense invariant domain 
$\hat{{\cal D}}_0$. Then for $\psi_0\in \hat{{\cal D}}_0$
\be \label{6.5}
H\; T\psi_0:=\lim_{M\to \infty} H_M\; T_M\psi_0
=\lim_{M\to \infty} T_M \hat{H}_M\;\psi_0=T \hat{H} \;\psi_0
\ee
defines $H$ as an operator ${\cal D}\to {\cal H}$ where $\cal H$ is the 
completion of $\cal D$ with respect to (\ref{6.4}). Note that if the rate
of divergence of the numbers $||T_M \psi||_0,\; \psi\in {\cal D}_0$ 
is not the same as that of $||T_M \Omega_0||_0$ up to a number, then 
the resulting space $\cal H$ will be finite dimensional and thus 
not acceptable.

The construction of a possible $T_M$ can of course be performed 
succesively in steps,
i.e. $T_M=T_M^n\;..\;T_M^1$ where each of the $T_M^k$ is designed to remove 
a certain type of term. Another possible approach is a perturbative
construction: We formally write $H=H_0+g V$ where $g$ is an organisation 
parameter of the perturbation. Then we write 
$T_M=\exp(-\sum_{n=1}^\infty\; g^n\; t_{M,n})$ and find 
\be \label{6.5a}
T_M^{-1} H_M T_M=H_{0,M}+g\;[V_M+[t_{M,1}]+g^2([t_{M,1},V_M]+[t_{M,2},
H_{0,M}]+\frac{1}{2}[t_{M,1},[t_{M,1},H_{0,M}]]+O(g^3)
\ee
and try to construct $t_{M,n}$ such that each coefficient of $g^n$ defines 
an operator with range in ${\cal D}_{0}$.

\subsection{Application to the SDA}
\label{s6.2}

We limit ourselves to the scalar case, the higher tensor rank case being 
similar but more complex and even less likely to succeed, although a strict 
proof must be (and can be given the tools provided here) carried out
to confirm. We write the normal ordered 
$C(u)$ for a general Fock representation parametrised 
by $\kappa$ in the suggestive form
\be \label{6.6}
C(u)=\frac{i}{2}\;d^Dx\;
u^a\{[\kappa^{-1}(A+A^\ast)]_{,a}\;[\kappa A]  
-[\kappa A^\ast]\;[\kappa^{-1}(A+A^\ast)]_{,a}  
\ee
As $C(u)$ would be well defined on ${\cal D}_0$ if the $(A^\ast)^2$ was 
not present we try to remove it using a dressing transformation $T$. Consider 
for a symmetric, real valued kernel $c$ to be specified the {\it squeezing
operator}
\be \label{6.7}
T:=e^{\frac{1}{2}\; c(A^\ast,A^\ast)},\;
c(A^\ast,A^\ast):=\int\; d^Dx\; d^Dy \; c(x,y)\; A^\ast(x)\; A^\ast(y)
\ee
We find with $[c\cdot f](x):=\int\; d^Dy \; c(x,y) f(y)$
\ba \label{6.8}
T^{-1}\;C(u)\; T
&=& 
\frac{i}{2}\;\int\; d^Dx\;
u^a
\{[\kappa^{-1}(A+[1+c\cdot] A^\ast)]_{,a}\;[\kappa A+c\cdot A^\ast]  
-[\kappa A^\ast]\;[\kappa^{-1}(A+[1+c\cdot] A^\ast)]_{,a}  
\}
\nonumber\\
&=& 
\frac{i}{2}\;\int\; d^Dx\;
u^a
\{
[\kappa^{-1}([1+c\cdot] A^\ast)]_{,a}\;[\kappa A+c\cdot A^\ast]  
[\kappa A+c\cdot A^\ast]\;[\kappa^{-1} A]_{,a}
\nonumber\\
&& +[[\kappa^{-1} A]_{,a},[c\cdot A^\ast]]  
-[\kappa A^\ast]\;[\kappa^{-1}(A+[1+c\cdot] A^\ast)]_{,a}  
\}
\nonumber\\
&=& 
\frac{i}{2}\;{\rm Tr}(u^a\;[\kappa^{-1}]_{,a} \;c)\cdot 1
+
\frac{i}{2}\;\int\; d^Dx\;u^a
\{
[\kappa^{-1}([1+c\cdot] A^\ast)]_{,a}\;[\kappa A]  
+[\kappa A+c\cdot A^\ast]\;[\kappa^{-1} A]_{,a}
\nonumber\\
&& -[\kappa A^\ast]\;[\kappa^{-1} A]_{,a}  
-[\kappa^{-1}([1+c\cdot] A^\ast)]_{,a}\;[(1-c\cdot) A^\ast]  
\}
\ea
We see that a non-trivial renormalisation is necessary in order to use
(\ref{6.7}) namely we must replace the definition (\ref{6.6}) of $C(u)$ 
by adding a central term that originates from normal ordering of (\ref{6.8})
\be \label{6.9}
C(u)- \frac{i}{2}\;{\rm Tr}(u^a\;[\kappa^{-1}]_{,a} \;c)\cdot 1
\ee
which is unaffected by conjugating with $T$. Denoting the modified $C(u)$ 
by $C(u)$ again we see that $T^{-1}\; C(u)\; T$ is given by the 
integral in the last line of (\ref{6.8}). 
That expression would define an operator on the Fock space if 
the $(A^\ast)^2$ contribution to that integral would vanish. That 
would be possible for the choice of integral kernel $c=\pm 1_{L_2}$. 
However, the squeezing operator $T$ for $c=1$ does not allow a 
controlled limit of the inner product (\ref{6.4}). The reason for this 
is $c=\pm 1$ defines infinite squeezing and turns either $p$ or $q$
into an annihilation operator as for $c=1$ we have 
$T^{-1} p T\propto A+c\cdot A^\ast-A^\ast=A$ and for $c=-1$ we have
$T^{-1} q T\propto A+c\cdot A^\ast+A^\ast=A$. This means that 
either $p T\Omega=0$ or $q T\Omega=0$ so that the squeezed vacuum $T \Omega$
is that of the Narnhofer Thirring representation. But in this representation
the conugate operator $q$ and $p$ respectively does not exist.

A formal way to see this is to work out the cut-off dressing operator
$T_M$ for $c=\pm 1$ and the corresponding inner product $<.>_M$ at finite $M$.
One can see that (\ref{6.4}) diverges at every finite $M$. In order to 
control the limit $M\to \infty$ one must ensure that $0<c_M^2<1$ at finite 
$M$. We will prove this in appendix \ref{sa} which is of interest by itself 
because it provides tools to perform complicated normal ordering manipulations
of exponentials such as those that appear in 
\be \label{6.10}
<T_M\Omega_0,T_M\Omega_0>_0=<\Omega_0,\;
e^{<\frac{1}{2}c_M(A_M,A_M)>}\; e^{<\frac{1}{2}c_M(A_M^\ast,A_M^\ast)>}   
\ee
which is surprisingly hard to compute. The result is
\be \label{6.11}
<T_M\Omega_0,T_M\Omega_0>_0=<\Omega_0,\;
\exp(-\frac{1}{2}{\rm Tr}_{L_M}(\ln(1-c_M^2)))
\ee
where $L_M$ is the subspace of the 1-particle Hilbert space $L_2$ spanned by 
the modes $I\in V_M$. The logarithm is computed using the spectral theorem.
For a translation invariant kernel $c_M$ the trace exists only for compact 
$\sigma$ and indeed in CQFT one usually starts in finite volume and only later
takes the thermodynamic limit. Thus for $\sigma=T^D$ with coordinate 
volume unity and a translation 
invariant kernel $c$ which is diagonal in 1-particle Hilbert space basis 
$c\cdot e_I=c_I e_I$
\be \label{6.12}
<T_M\Omega_0,T_M\Omega_0>_0=
\prod_{I\in V_M} [1-c_I^2]^{-1/2}
\ee
This shows that we must not pick $c=1$ i.e. $c_I=1$ for any $I$, rather 
we want that $c_I\to \pm 1$ for $I\in S_M,\; M\to \infty$ in a controlled way.
However, if $|c_I|<1$ for any $I\in S_M,\; M<\infty$ it follows that 
(\ref{6.8}) when expanded into modes $I$ and then truncated to $I\in S_M$,
let us denote it by $\hat{C}_M(u)$, will not have range in ${\cal D}_0$ 
for some fixed $M_0$ as $M\to\infty$ 
Therefore $\hat{C}_M(u) {\cal D}_0$ will have at best 
range in ${\cal H}_0$ as $M\to \infty$ rather than ${\cal D}_0$ 
and thus it is not clear that the programme of section \ref{s6.1}
will work.   
  
To investigate this we evaluate the norm of 
$T_M \hat{C}_M(u)\Omega=C_M(u) T_M\Omega$ 
(all indices restricted to $S_M$) say for $\sigma=T^D$ ($T_M$ can be moved 
to the vacuum past the creation operators that survive in 
$\hat{C}_M(u)\Omega$)
\ba \label{6.13}
&& \frac{||T_M \hat{C}_M(u) \Omega||^2}{||T_M\Omega||^2}
=\sum{I,J,I',J'\in S_M} \; 
[\hat{u}^a_{I+J} \frac{1+\hat{c}_I}{\kappa_I}\; k_a^I \hat{\kappa}_J 
(1-\hat{c}_J]^\ast\;\;
[\hat{u}^a_{I'+J'} \frac{1+\hat{c}_{I'}}{\kappa_{I'}}\; k_a^{I'} 
\hat{\kappa}_{J'} 
(1-\hat{c}_{J'}]\;\;
\nonumber\\
&& \frac{<T_M \Omega, A_I A_J A_{I'}^\ast A_{J'} T_M\Omega>}
{<T_M \Omega, T_M\Omega>}
\ea
The matrix elements in the last line
receives non-vanishing contributions only when 
either $I=I',J=J',\; I\not=J$ or $I=J',J=I',\;I\not=J$ or 
or $I=J, I'=J',\;I\not=I'$ or
$I=J=I'=J'$
because $T_M\Omega$ contains only even numbers of excitations for each mode.
These matrix elements factorise into factors of the 
form 
\be \label{6.14}
<T_M \Omega, A_I A_I^\ast T_M \Omega>,\; 
<T_M \Omega, [A_I]^2 T_M \Omega>,\; 
<T_M \Omega, [A_I^\ast]^2 T_M \Omega>,\; 
<T_M \Omega, [A_I]^2 [A_I^\ast]^2 T_M \Omega> 
\ee
After dividing by $<T_M\Omega,T_M\Omega>$ these become with 
$\Omega_I=e^{c_I [A_I^\ast]^2/2}\Omega$ and $N_I=||\Omega_I||^2$
\be \label{6.14a}
N_I^{-1}\;<\Omega_I, A_I A_I^\ast \Omega_I>,\; 
N_I^{-1}\; <\Omega_I, [A_I]^2 \Omega_I>,\; 
N_I^{-1}\;<\Omega_I, [A_I^\ast]^2 \Omega_I>,\; 
N_I^{-1}\;<\Omega_I, [A_I]^2 [A_I^\ast]^2 \Omega_I> 
\ee
These numbers can be computed by elementary quantum mechanics of the 
harmonic oscillator by epanding the exponential
(we drop the label $I$ and $c$ is now just a number, not a kernel):
\ba \label{6.15}
&& N=||e^{c[A^\ast]^2/2}\Omega||^2=
\sum_{m,n}\;\frac{[c/2]^{m+n}}{m!\;n!}
<\Omega,A^{2m}\;[A^\ast]^{2n}\Omega>=    
\nonumber\\
&=&\sum_n\;\frac{[c/2]^{2n}\;[2n]!}{[n!]^2}=[1-c^2]^{-1/2}
\nonumber\\
&& ||A^\ast\; e^{c[A^\ast]^2/2}\Omega||^2
=\sum_n\;\frac{[c/2]^{2n}\;[2n+1]!}{[n!]^2}
=\frac{d}{dc} c\;
\sum_n\;\frac{[c/2]^{2n}\;[2n]!}{[n!]^2}
\nonumber\\
&=& \frac{d}{dc} c\;[1-c^2]^{-1/2}=[1-c^2]^{-3/2}
\nonumber\\
&& <e^{c[A^\ast]^2/2}\Omega,[A^\ast]^2\; e^{c[A^\ast]^2/2}\Omega>
=\sum_n\;\frac{[c/2]^{2n+1}\;[2n+2]!}{[n!]\;[n+1]!}
\nonumber\\
&=& c\sum_n\;\frac{[c/2]^{2n}\;[2n+1]!}{[n!]^2}
=c[1-c^2]^{-3/2}
\nonumber\\
&&||[A^\ast]^2\; e^{c[A^\ast]^2/2}\Omega||^2
=\sum_n\;\frac{[c/2]^{2n}\;[2(n+1)]!}{[n!]^2}
\nonumber\\
&=& \frac{d^2}{dc^2} c^2\;[1-c^2]^{-1/2}
=(2-c^2+c^4)[1-c^2]^{-5/2}
\ea
In the first step recognised the Taylor series of the function 
$\mapsto (1-x)^{-1/2}$ with radius of convergence $|x|<1$. In the second 
we used that the number is given by $d/dc\cdot c$ applied to the same series.
The third series turns out to be just $c$ times the second. In the fourth step
we applied an identity similar to the one in the second step. 

These 
computations reveal that 
the norms of excitations of the squeezed state
diverge at a faster rate than itself: In fact we can compute the rate of 
divergence by the same computation
\be \label{6.16}
\kappa_{I,n}^2:=||[A_I^\ast]^n\; e^{c_I[A_I^\ast]^2/2}\Omega||^2
=[\frac{d}{dc_I}]^n\; c_I^n\;[1-c_I^2]^{-1/2} 
\ee
which is of order $[1-c_I^2]^{-(2n+1)/2}$ as $c_I\to 1$ (for $c_I\to 0$ 
this becomes $(n)!$). Accordingly, if we 
want to use the squeezing operator as a dressing transformation we 
must use the span of the vectors 
\be \label{6.17}
b'_n=\prod_I\; \frac{[A_I^\ast]^{n_I}}{\kappa_{I,n_I}} \Omega,\;
||n||=\sum_I n_I<\infty
\ee
which is now orthonormal with respect to $<.,.>_M=<T_M.,T_M.>/||T_M\Omega^2||$
for $M$ sufficiently large. The span of the $b'_n$ is the same as that 
of the $b_n$, that is, ${\cal D}_0$.  

Returning to (\ref{6.13}) we focus first on the term with the strongest 
divergence corresponding to the $I=J=I'=K'$ term which is given by
\be \label{6.18}
\sum{I\in S_M} \; 
|\hat{u}^a_{2I}\; k_a^I (1-\hat{c}_I^2)|^2 \kappa_{I,2}^2\; 
[1-\hat{c}_I^2]^{1/2}
\ee
which can be made finite if $\hat{u}^a_I$ has finite mode support
because $\kappa_{I,2}^2$ diverges only as $[1-c_I^2]^{-5/2}$. Here 
we identify the mode labels $I$ on $T^D$ with $I\in \mathbb{Z}^D$ and 
$\hat{u}^a_I=<e_I,u^a>_{L_2}$ and $k^I$ is the momentum 
of model label $I$ (proportional to $I$). Next consider the configuration 
$I'=I\not=J=J'$ which contributes the term
\ba \label{6.19}
&&\sum_{I\not=J\in S_M} \; 
|\hat{u}^a_{I+J} \frac{1+\hat{c}_I}{\kappa_I}\; k_a^I \hat{\kappa}_J 
(1-\hat{c}_J)|^2\;\;\kappa_{I,1}^2 \kappa_{J,1}^2 N_I^{-1} N_J^{-1}
\nonumber\\
&=& \sum_{I\not=J\in S_M} \; 
|\hat{u}^a_{I+J} \frac{1+\hat{c}_I}{\kappa_I}\; k_a^I \hat{\kappa}_J 
(1-\hat{c}_J|^2\;\;(1-c_I^2)^{-1} (1-c_J^2)^{-1}
\nonumber\\
&=& \sum_{I\not=J\in S_M} \; 
|\hat{u}^a_{I+J} k_a^I|^2\;
\frac{[1+\hat{c}_I][1-c_J]}{[1-c_I][1+c_J]}\; 
\frac{\hat{\kappa}^2_J}{\kappa^2_I} 
\ea
which no longer converges as $M\to \infty$ because $\hat{u}_{I+J}$ cannot
suppress $I,J$ individually.\\
\\
\\
We conclude that the apparanetly natural choice of a squeezing operator as 
a dressing transformation does not quite work although it comes 
rather close. This does not exclude the existence
of better behaved dressing operator. The perturbative 
approach based on (\ref{6.5a}) 
however is unlikely to work because in our case $H_0,V$ depend on $u$ 
and so would the $T_M$ constructed from (\ref{6.5a})
while 
the needed transformation $T_M$ must not depend on $u$ so that for the 
commutator algebra we have the property 
$[C_M(u),C_M(v)]T_M=T_M \;[\hat{C}_M(u),\hat{C}_M(v)]$.

\section{S-Matrix techniques for the SDG}
\label{s7}

Given the fact that one can find Fock representations of $D(u)$ 
satisfying $[D(u),D(v)]=iD(-[u,v])$ and the split $C(u)=D(u)+S(u)$ we are 
reminded of the usual problem of QFT in which we are given a Hamiltonian
operator $H$ which can be split as $H=H_0+V$ where $H_0$ is well defined 
on a suitable Fock space and $V$ is the interaction which in general 
is no longer a well defined operator on that Fock space but only a 
quadratic form when normal ordered. 
Thus we have a correspondence between the roles 
of $C(u),D(u),S(u)$ and $H,H_0,V$ as in the previous section. 

In QFT one is often content to define the S-matrix $S$ which itself also is 
only a quadratic form. The S-matrix is constructed from the Moeller 
operator $M(t)=e^{it H} \; e^{-it H_0}$ perturbatively using 
Gell-Mann Low formula which expresses $V$ in terms of the free time 
evolution of the fields. In our case 
this would correspond to construct a 
quadratic form corresponding to $U(u):=e^{iC(u)}$ which can be deduced from 
the analogous Moeller operator expression 
$e^{i C(u)}\; e^{-iD(u)}$ as $e^{i D(u)}$
is a well defined unitary operator on the chosen Fock space. The
so constructed $U(u)$ is then a quadratic form representation of the 
SDG element $\varphi^u_{s=1}$ and one can analyse unitarity properties 
in the form $U(u)^\ast=U(-u)$ which does not require multiplication of 
quadratic forms. 

Accordingly we proceed to to construct 
\be \label{7.1}
M_u(s):=e^{is C(u)}\; e^{-is D(u)}
\ee
for simplicity in the Fock representation with $\kappa=1$ for scalars. 
It obeys the ODE
\be \label{7.2}
-i\frac{d}{ds}\; M_u(s)
=e^{is C(u)}\;[C(u)-D(u)]\;e^{-is\;D(u)}
=M(s)\; e^{is D(u)}\;S(u)\;e^{-is\;D(u)}
=:=M_u(s)\; S_s(u)
\ee
where $S_s(u)$ is the ``free evolution'' of $S(u)$ which treats $\pi,\phi$ 
or equivalently $A,A^\ast$
as scalar densities of weight 1/2. We have explicitly
\be \label{7.3}
S(u)=-\frac{i}{4}\int\; d^Dx\; u^a_{,a}\; [A^2-(A^\ast)^2]
\ee
and thus with the abbreviation $v=u^a_{,a}$ 
\ba \label{7.4}
S_s(u)
&=& =-\frac{i}{4}\int\;d^Dx\; v(x)\; 
|\det(\partial \varphi^u_s(x)/\partial x)|\; [A^2-(A^\ast)^2](\varphi^u_s(x))
\nonumber\\
&=& -\frac{i}{4}\int\;d^Dx\; v(\varphi^u_{-s}(x))\; 
[A^2-(A^\ast)^2](x)
\nonumber\\
&=:& S(u_s)
\ea
Here $s\mapsto \varphi^u_s$ is the one parameter group of diffeomorphisms
generated by the vector field $u$, that is, the integral curve 
$\varphi^u_s(x)=c^u_x(s)$ satisfying $c^u_x(0)=x,\;\frac{d}{ds} 
c^u_s(x)=u(c^u_x(s))$. 
Explicitly  $v(\varphi^u_{-s}(x))=[e^{-s L_u}\cdot v](x)$ where $L_u$ denotes
the Lie derivative on scalars.

The solution of (\ref{7.2}) is given by 
\be \label{7.5}
M_u(s)={\cal P}_r\;\exp(i\int_0^s\; ds'\; S_{s'}(u))
\ee
where the path ordering symbol orders the latest parameter dependence to 
the right. What now simplifies (\ref{7.5}) daramatically is the fact 
that the quadratic forms $S(u),S(u')$ are formally mutually commuting.
Strictly speaking, this statement is ill defined as it stands because 
quadratic forms cannot be multiplied. However, it makes sense in the presence 
of a mode cut-off when computing the commutator and if after computing 
the commutator the cut-off can be removed. This will be made precise 
in the next section. For the moment a formal computation will suffice
and we abbreviate $v=u^a_{,a},\;v'=u^{\prime a}_{,a}$
\ba \label{7.6}
[S(u),S(u')]
&=& \frac{1}{16}\int\; d^Dx\; d^Dy\; 
[u(x)\;u'(y)-u(y)\;u'(x)]\;[A^2(x),A^\ast(y)^2]
\nonumber\\
&=& \frac{1}{8}\int\; d^Dx\; d^Dy\;\delta(x,y) 
[u(x)\;u'(y)-u(y)\;u'(x)]\;\{A(x) A^\ast(y)+A^\ast(y) A(x)\}
\ea
The r.h.s. is no longer normal ordered, bringing it into normal ordered
form yields the constant 
\be \label{7.7}
\int\;=\frac{1}{8}\int\; d^Dx\; d^Dy\;\delta^2(x,y)\;
[u(x)\;u'(y)-u(y)\;u'(x)]
\ee
which is of the type $0\cdot \infty$. It is given the value zero in the above
mentioned precise sense.

Since the $S_s(u)$ are mutually commuting we can drop the normal ordering 
symbol and simplify 
\be \label{7.8}
M_u(s)=\exp(i\int_0^s\; ds'\; S_{s'}(u))=\exp(iS(u_s)),\;\;
u_s(x)=\int_0^s\; dr\; [e^{-r L_u} u^a_{,a}](x)
\ee
This shows that we can define $e^{iC(u)}=M_u(1) e^{i D(u)}$ as a quadratic
form if and only if we can define $e^{iS(v)}$ as a quadratic form for any 
$v$. 

We could now do this using the {\it counter-term method}. That is, 
we try to give meaning to the Taylor expansion 
\be \label{7.9}
e^{i S(v)}=\sum_{n=0}^\infty\; \frac{i^n}{n!}\; S(v)^n
\ee
but already $S(v)^2$ is ill-defined because quadratic forms cannot be 
multiplied. Accordingly we modify $S(v)$ by a counter-term 
$S(v)\to S(v)+Z(v)$. The only counter term allowed is a constant 
$Z(v)$ because the counter terms must be of the same form as they 
appear in $S(v)$ modulo constants. In fact, only constants are needed 
in order to bring powers of $S(v)$ into normal ordered form.
To determine $Z(v)$ we 
define inductively (note that $S(v)=:S(v):$ is already normal ordered)
\ba \label{7.10}
&& S(v)^2= :S(v)^2:+Z_{2,0}(v)\; 1,\;
S(v)^3= :S(v)^3:+Z_{3,0}(v)\; S(v),\;
\\
&& S(v)^4= :S(v)^4:+Z_{4,0}(v)\; :S(v)^2: + Z_{4,1}(v)\; 1,\;
S(v)^n= \sum_{k=0}^{[n/2]}\; Z_{n,k}(v)\;:S(v)^{n-2k}:
\nonumber\\
\ea
where $[.]$ is the Gauss bracket. The constants $Z_{n,k}(v)$ can be 
inductively computed by this method but it is very difficult to 
find the building scheme in this way. Fortunately, for the particular
form of 
$S(v)$ we can apply proposition \ref{prop.a1} in the appendix to immediately 
write $e^{i S(v)}$ in normal ordered form. To apply the proposition
we note that 
\be \label{7.11}
iS(v)=a(A^\ast,A^\ast)/2+c(A,A)/2,\;
c(x,y)=\frac{1}{2}v(x)\;\delta(x,y),\; a(x,y)=-c(x,y)
\ee
thus we consider in the terminology of that proposition the 
case $b=0, a=-c$ so that we can apply formula (\ref{a.25})
\be \label{7.12}
e^{iS(v)}=e^{k/2}\;
e^{l(A^\ast,A^\ast)/2}\;e^{m(A^\ast,A)}\;e^{r(A,A)/2},\;
k=-{\rm Tr}(c {\rm th}(c)),\;
l=-r={\rm th}(c),\;
m=-\ln({\rm ch}(c))
\ee
This means that 
\be \label{7.13}
e^{iS(v)-k/2\cdot 1}=[e^{i S(v)}]_{{\rm q.f.}}
\ee
is a well defined quadratic form with form domain given by the span of 
Fock states which can be interpreted as renormalised {\it scattering matrix}.

This is quite surprising: From (\ref{7.10}) one 
would not have guessed that the explicit expression 
of the normal ordering constants $Z_{n,k(v)}$ allows for a 
reorganisation of the series to the effect that the exponential 
of the {\it renormalised scaling transformation quadratic form}
\be \label{7.14}
S_{{\rm ren}(v)}=S(v)+ik(v)/2\cdot 1
\ee
results in a well defined quadratic form. Note that $S_{{\rm ren}}(v)$ 
by itself is not a symmetric quadratic form. This is because $e^{iS(v)}$ is 
formally 
unitary, thus $[e^{iS(v)}]^\ast=e^{-iS(v)}=e^{iS(-v)}$. This is 
implemented by (\ref{7.12}) because $l,r$ are odd under $c\to -c$ caused 
by $v\to -v$ while $k,m$ are even. Due to 
$c^2(x,y)=\int\; d^Dz\; v(x)\delta(x,z)\;v(z)\;\delta(z,y)
=v^2(x) \delta(x,y)$ we have explicitly
\be \label{7.15}
k(v)=\int\; d^Dx\; v(x)\;{\rm th}(v(x))\; \delta(x,x)
\ee
which must be interpreted again to the effect that one first works at finite 
mode cut-off upon which $\delta(x,x)$ is replaced by the finite number
$p_M(x,x),\; p_M(x,y)=\sum_{I\in S_M}\; e_I(x)\; e_I(y)^\ast$, then computes 
$S_{{\rm ren}}(v)$ and then takes the limit.    

Of course $e^{iS_{{\rm ren}}(v)}$ still cannot be multiplied by 
$e^{iS_{{\rm ren}}(v')}$ since these are just quadratic forms. However we have 
due to $k(v)=k(-v)$
\be \label{7.16}
e^{iS_{{\rm ren}}(v)}\;e^{iS_{{\rm ren}}(v')}\;  
e^{-iS_{{\rm ren}}(v)}\;e^{-iS_{{\rm ren}}(v')}  
=e^{iS(v)}\;e^{iS(v')}\;e^{-iS(v)}\;e^{-iS(v')}  
=e^{i[S(v)+S(v')-iS(v)-iS(v')]}  
=1
\ee
where we used (\ref{7.6}). In this sense the SDA holds for the renormalised 
scaling transformations.

\section{Quadratic form constraint algebra}
\label{s8}

In the first subsection we review the procedure to construct 
commutators of constraint {\it operators}. We use it to motivate 
the second subsection to define commuators of {\it qudratic forms} 
which seems contradictory as quadratic forms cannot be multiplied.
In the third subsection we apply this concept to the SDA. 

\subsection{Operator constraint algebras}
\label{s8.1}

The ususal setup for a Poisson algebra of first class constraints 
$C(u)$ with smearing functions $u\in V$ and with 
structure constants $\{C(u),C(v)\}=C(f(u,v))$ where $f(u,v)=-f(v,u)$ are
constant functions on the phase space is to ask for a Hilbert space 
$\cal H$ with dense subspace ${\cal D}$ on which the constraints $C(u)$ 
are implemented as {\it operators}, i.e. $C(u)\;{\cal D}\subset {\cal H}$.
In order that also commutators be well-defined we ask for an 
{\it invariant operator domain}, i.e. we require even that 
$C(u)\;{\cal D}\subset {\cal D}$. The joint kernel of all the $C(u)$ are 
common generalised zero eigenvectors, i.e. linear functionals 
$l:\;{\cal D}\to \mathbb{C}$ such that 
\be \label{8.1a}
l[C(u)\psi]=0\;\forall u\in U,\;\forall \psi\in {\cal D}
\ee
Suppose first that the classical $C(u)$ generate a Lie algebra 
$\{C(u),C(v)\}=C(f(u,v))$ with structure constants 
$f:\;V\times V\to V$. If we have a representation of the $C(u)$
in the quantum theory, that is,
\be \label{8.2a}
[C(u),C(v)]=iC(g(u,v))
\ee
with structure constants 
$g:\;V\times V\to V$ then we are granted that the space of solutions 
$l$ is not too small in the 
following sense: If $[C(u),C(v)]$ would not be of the form $C(g(u,v))$ 
then $l$ automatically satisfies more conditions than just (\ref{8.1a}). 
This allows that $g(u,v)\not=f(u,v)$ i.e. the constraint algebra is 
is {\it mathematically non-anomalous}. If in addition $g=f$ we call it
{\it physically non-anomalous}. 

Now it may happen that in the classical theory $f$ is not a constant function 
on the phase space $\Gamma$, then $f:\; V\times V\to V\times 
{\rm Fun}(\Gamma)$ takes values in the phase space function valued smearing
functions (in short: structure functions). In the quantum theory then 
$f:\; V\times V\to V\times 
{\cal L}({\cal H})$ will be an operator valued smearing function. In this 
case a relation like (\ref{8.2}) faces ordering issues. In case that  
\be \label{8.3a}
[C(u),C(v)]=i\int\; d^D\; C(x)\;[g(u,v)](x)
\ee
with $g(u,v)$ is ordered to the right and such $g(u,v)(x)\; {\cal D}\to
{\cal D}$ will still lead to no etra condition (the same if $g=f$).

Now a solution $l$ is entirely determined by its values $l(\psi),\; \psi\in 
{\cal D}$ and thus given an orthonormal basis $b_n\in {\cal D}, n\in {\cal N}$
any solution 
$l$ can be written in the form
\be \label{8.4a}
l=\sum_n l(b_n)\; <b_n,.>_{{\cal H}}
\ee
Thus, to construct a solution we consider arbitrary coefficients $l_n$ which 
must solve 
\be \label{8.5a}
\sum_n\; l_n\; <b_n,C(u)\;b_{n'}>=0\; \forall \; u\in V,\;n'\in {\cal N} 
\ee
The point of repeating this well known formalism is to note that the 
condition (\ref{8.5a}) {\it never needs that $C(u)$ is an operator}! 
To state (\ref{8.5a}) it is sufficient to require that $C(u)$ is a quadratic
form with form domain given by ${\cal D}$. The solutions $l$ typically are not 
elements in $\cal H$ but merely distributions (linear functinals) on 
$\cal D$. We therefore advertise the point of view that {\it it is sufficient
to quantise constraints as quadratic forms}.

This however appears to meet a caveat: Quadratic forms cannot be multiplied
(even operators cannot unless e.g. 
they have a common dense invariant domain) 
thus it appears impossible to check whether (\ref{8.2a}) or (\ref{8.3a})
hold. In the first subsection 
we will therefore develop a formalism that could 
be called {\it quadratic form commutator algebra}. It depends on an 
additional structure called a {\it regularisation of the quadratic form
product} and which is inspired by the CQFT procedure laid out in section 
\ref{s6}. In the second subsection we will then apply this formalism
to the SDA.

\subsection{Truncation structures and quadratic form commutators}
\label{s8.2}

\begin{Definition} \label{def8.1} ~\\
i.\\
Fix some system of coordinates on $\sigma$.
Let $\{e_I\}_{I\in S}$ be a real valued ONB of the 1-particle Hilbert space 
$L_2:=L_2(\sigma,d^Dx)$ of scalar fields on $\sigma$. 
Here the mode labels $I$ take values in the mode set $S$. 
A truncation structure on $S$ is a nested 
sequence $\mathbb{N}_0\to S;\;M\mapsto S_M$ of finite subsets $S_M\subset S$
such that $S_0=\emptyset,\; S_M\subset S_{M+1},\; 
\cup_{M=0}^\infty \;S_M=S$.\\
ii.\\
Given an ONB with mode set $S$ and a truncation structure on $S$, we define
the truncation kernel $k_M,\; M\in \mathbb{N}_0$ on $f\in L_2$ by
\be \label{8.1}
[k_M\; f](x)=\sum_{I\in S_M}\; e_I(x)\; <e_I,f>_{L_2}    
\ee
iii.\\
Suppose that $H=H[q,p]$ is a real valued polynomial functional of canonically 
conjugate tensor fields $q,p$ of type $(A,B,w),(B,A,1-w)$
with a symmetric ordering chosen (e.g. the one obtained from 
Weyl quantisation). The canonical truncation 
$H^c_M$ of $H$ is obtained by replacing $H$ by $H_M:=H[k_M q,k_M p]$ where 
$k_M q, k_M p$ denotes the tensor component-wise application of 
(\ref{8.1}).\\
iv.\\
Let $\kappa$ be a real valued, positive, scalar kernel and 
$A^\mu=2^{-1/2}[\kappa q^\mu-i \delta^{\mu\nu}\;\kappa^{-1} p_\nu]$
where $\mu$ is a compound index i.e. 
$q^\mu=q^{a_1..a_A}_{b_1..b_B},\; p_\mu=p_{a_1..a_A}^{b_1..b_B}$.
Let $H[A,A^\ast]$ be the rewriting of $H$ in terms of $A,A^\ast$ with 
normal ordering prescription applied. The normal truncation 
$H^n_M$ of $H$ corresponding to $\kappa$
is given by replacing $H$ by $H[k_M A,k_M A^\ast]$ where
$k_M A, k_M A^\ast$ denotes the component-wise application of (\ref{8.1}).
\end{Definition}
This could be generalised to more general, positive definite 
tangent space metrics $g^{\mu\nu}$ replacing $\delta^{\mu\nu}$. We 
note that in Hilbert space representations $q,p, A, A^\ast$ are 
typically operator valued distributions, while  
the truncations $k_m q, k_M p, k_M A, k_M A^\ast$ are operators
if the functions $e_I$ are chosen from a suitable space of test functions 
(e.g. Hermite functions for $\sigma=\mathbb{R}^D$ or momentum eigenfunctions 
for $\sigma=T^D$) because the sum in (\ref{8.1}) is finite. For the same 
reason, $H^c_M, H^n_M$ are then well defined operators thanks to the 
assumed polynomial dependence on $q,p$.      
\begin{Definition} \label{def8.2} ~\\
Let $\omega$ be a regular state (i.e. on the Weyl algebra $\mathfrak{A}$
generated 
by the Weyl elements $W^\mu_\nu(e_I,e_J)=\exp(i[<e_I,q^\mu>+<e_J,p_\nu>)$. 
Let $({\cal H}, \Omega, \rho)$ be its GNS data given by a Hilbert space 
$\cal H$, a representation $\rho$ of $\mathfrak{A}$ on $\cal H$ and a 
cyclic vector $\Omega$. Let ${\cal D}=\rho(\mathfrak{A})\Omega$ be 
the dense subspace of $\cal H$ invariant for $\rho(\mathfrak{A})$.
We will drop the letter $\rho$ for simplicity and do not distinguish 
between the algebra element and its operator representation. The regularity 
of the state means that we have also access to the 
$q_I=<e_I,q>, p_I=<e_I,p>, A_I=<e_I,A>, A_I^\ast=<e_I,A>^\ast$, not only their 
exponentials.\\ 
i.\\
The canonical quadratic form corressponding to $H$, if it exists, is defined 
as the weak limit 
\be \label{8.2} 
<\psi,\;H^c\;\psi'>_{{\cal H}}
:=\lim_{M\to\infty}\; <\psi,\;H^c_M\; \psi'>_{{\cal H}}   
\ee
with dense form domain $\psi,\psi'\in {\cal D}$.\\
ii.\\
The normal quadratic form corresponding to $H$ and $\kappa$, if it exists, 
is defined as the weak limit 
\be \label{8.3} 
<\psi,\;H^n\;\psi'>_{{\cal H}}
:=\lim_{M\to\infty}\; <\psi,\;H^n_M\; \psi'>_{{\cal H}}   
\ee
with dense form domain $\psi,\psi'\in {\cal D}$.
\end{Definition}
These definitions can be applied to classical constraint functions and 
yield canonical or normal truncations $C_M(u)$ as operators with 
dense and invariant domain $\cal D$ and, if existent, quadratic forms 
$C(u)$ with form domain $\cal D$. 
The smearing functions $u$ of the constraints belong to a certain space
of tensor fields 
which we take to be component wise spanned by the $e_I$ as well.

Quadratic forms cannot be multiplied.
However, one may be able to define their commutator in the following sense:
\begin{Definition} \label{def8.3} ~\\
Suppose that operator 
truncations $C_M(u),\; C_M(v)$ have quadratic form  
limits $C(u), C(v)$ as above. Their commutator, if existent, is defined 
as the quadratic form given by the weak limit  
\be \label{8.4}
<\psi,\;[C(u),C(v)]\;\psi'>_{{\cal H}}:=
\lim_{M\to\infty}\; <\psi,\;[C_M(u),C_M(v)]\;\psi'>_{{\cal H}}
\ee
with dense form domain $\psi,\psi'\in {\cal D}$.
\end{Definition}
To see why (\ref{8.4}) has a chance to converge even though neither 
$C(u) C(v)$ or $C(v) C(u)$ have any meaning, we note first that 
$[C_M(u),C_M(v)]$ is a well defined operator with dense invariant 
domain $\cal D$. Now due to the minus sign involved in 
$[C_M(u),C_M(v)]=C_M(u) C_M(v)-C_M(v) C_M(u)$ it is possible that the 
terms that diverge when we consider 
$<\psi, C_M(u) C_m(v)\psi'>$ and 
$<\psi,C_M(v) C_m(u)\psi'>$ separately in fact cancel out. Whether or not 
this happens may depend on all the choices $S, e_I, g^{\mu\nu}, \kappa, 
\omega$. It may also happen that one finds that (\ref{8.4}) does not 
exist with the definition of $C_M(u)$ as given in definition \ref{def8.1}
but that it exists if one subtracts from $C_M(u)$ certain counter terms 
which are of the same form as the terms that define $C_M(u)$ but with 
different coupling constants, see the previous section. In this 
case one will work with these subtracted, renormalised quadratic forms
which we denote by $C(u)$ again.
\begin{Definition} \label{def8.4} ~\\ 
Let $C(u)$ be classical first class constraints with structure constants 
$f(u,v)=-f(v,u)$ defined by the Poisson bracket $C(f(u,v)):=\{C(u),C(v)\}$. 
Here 
$f:V\times V\to V$ where $V$ is the space of smearing functions as defined 
above. Suppose that corresponding 
constraint quadratic forms $C(u)$ with quadratic form valued commutators 
can be defined as above. Then the quadratic form commutator algebra is 
called mathematically non-anomalous
\be \label{8.5}
[C(u),C(v)]=i\;C(g(u,v))
\ee
for any $g:\; V\times V\to V, \; g(v,u)=-g(u,v)$. It is called physically
non-anomalous if even $g=f$.
\end{Definition}
This definition encompasses constraints which form a Lie algebra but this 
is not the most general situation that one encounters in practice. For 
instance in GR
the algebra of SDC $C$ and HC $H$ 
closes in the sense that their Poisson brackets 
are linear combinations of SDC and HC but no longer with structure  
constants but rather structure functions on the classical phase space. 
Thus while we can still write $C(f(u,v)):=\{H(u),H(v)\}$ now 
$f$ depends on the point in the phase space at which one computes the 
Poisson bracket and therefore the quadratic form correspondent 
$C(f(u,v))$ is ill defined: It lacks an ordering prescription for the 
operator or quadratic form corresponding to $f(u,v)$. In particular,
the classical expression for $C(f(u,v)$ when normal ordered and which 
one expects to obtain from the commutator calculation if $H(u), H(v)$ are 
normal ordered, 
is not given by $\int \; d^Dx\; C(x)\; f(u,v)(x)$ with both $C(x),f(u,v)(x)$
normal ordered and $C(x)$ to the left. 
Note that normal ordering is crucial in order that the weak 
quadratic form limit exists. Thus while one could {\it define} 
the ordering of $i C(f(u,v))$ by $[H(u),H(v)]$, in this case this 
case $C(f(u,v))$ does not have the following property:
For any linear functional $l:\cal D\to \mathbb{C}$ such that 
$l(C(u)\psi)=0$ for all $\psi,u$ we automatically have 
$l(C(f(u,v))\psi)\equiv 0$ for all $u,v,\psi$. 
Similarly, a symmetric ordering of $H(u)$ induces that 
$C(f(u,v))=-i\;[H(u),H(v)]$ is symmetrically ordered which therefore
is not of the form of $\int \; d^Dx\; C(x)\; f(u,v)(x)$. See the discussion 
in \cite{29} for more details. 

The reason for why this is not problematic for commutators of 
quadratic forms in contrast to the commutator of operators is the 
following: If $C(u)$ is an 
actual operator with dense invariant domain $\cal D$ then indeed 
$l(C(u)\; C(v)\psi)$ is well defined for $l:\; {\cal D}\to \mathbb{C}$ 
for any $\psi\in {\cal D}$ since also $C(u)\;C(v)\psi\in {\cal D}$. 
However, if $C(u)$ is just a quadratic form then the formal object 
$C(u)\psi$ which one may try to define by $C(u)\psi:=\sum_n <b_n,C(u)\psi> 
b_n$ in terms of an ONB is not even an element of $\cal H$ any more! 
Therefore one cannot define $C(u)\psi$ and even less $C(u)C(v)\psi$. Thus 
in contrast to the case of operators, where $l(C(u)\psi)=0$ implies 
$l([C(u),C(v]\psi)=l(C(u)C(v)\psi)-l(C(v)C(u)\psi)=0$ such a conclusion 
is simply not possible for the case of quadratic form valued constraints 
because only $[C(u),C(v)]$ is defined as a quadratic form not the individual
products $C(u)C(v), C(v)C(u)$. One may be tempted to introduce a resolution 
of unity $l([C(u),C(v)]\psi)=\sum_n\; 
[l(C(u)b_n)\;<b_n,C(v)\psi>-l(C(v)b_n)\;<b_n,C(u)\psi>]$
but this ``identity'' is none because $C(u)\psi)\not\in {\cal H}$ 
and thus cannot 
be expanded into the $b_n$ (formally the interchange of limits $M\to \infty$
and $n\le N,\; N\to\infty$ implicit in this identity is not allowed).
Therefore the anomaly conclusion that holds for operators does not 
apply for quadratic forms.

\subsection{Application to the SDA}
\label{s8.3}
 
Still one may wish to avoid these complications and  
either work with an Abelianised version of the constraints or the 
corresponding master constraint \cite{9,30} which avoids the structure 
functions. More concretely, write the HC in density weight zero form $H(x)$ by 
multiplying by suitable powers of the determinant $Q$ 
of the spatial D-metric, and consider the normal ordered version of the 
master constraint
\be \label{8.6}
M:=\int \; d^Dx \; H(x)\; Q(x)\; H(x)
\ee
and try to perform above steps to construct a quadratic form. This is 
much harder than for $C(u)$ because the integrand of (\ref{8.6}) 
is not a polynomial. The reason for why we use a non-polynomial form for 
$M$ is that then classically $\{C(u),M\}=0$. In a polynomial form we 
could not close the Poisson algebra of the $M, C(u)$ not even with structure 
functions. Thus in the form (\ref{8.6}) we are left to to check $[C(u),M]=0$ in 
the quantum theory which is now a Lie algebra situation. 
 
In this paper we just focus on the SDA. We have the astonishing result:
\begin{Proposition} \label{prop8.1} ~\\
For any choice of $e_I, S, \kappa$ the normal ordered constraint quadratic 
form $C(u)$ exists in the Fock representation $\omega$ dictated by $\kappa$
and the algebra of quadratic form valued constraints not only exists but 
is also physically non-anomalous.
\end{Proposition}
Before we go into the details, let us explain the reason for why this 
works: Any normal ordered polynomial $C_M(u)$ is automatically a 
well defined quadratic form in the Fock representation. Therefore 
in the computation of $C_M(u),C_M(v)$ the possible source of trouble results 
from a term that is not automatically normal ordered. This is a term 
of the symbolic form $[A_M^2,(A_M^\ast)^2]$. In contrast to the untruncated 
case $M=\infty$ this can be meaningfully computed at any finite $M$.
When reordering into normal order, one picks up a normal ordering constant.
However that constant depends on a symmetric combination of labels 
$I,J\in S_M$ and we have to sum it over $S_M$ 
against an antisymmetric matrix built from a truncated version of 
$[u,v]$. As all sums are finite, the unwanted term cancels.\\
\begin{proof}:\\
Let ${\cal D}_M$ be the subspace of ${\cal D}$ which is spanned by 
Fock states $b_n$ with occupation number $n_I=0$ for all $I\not\in S_M$.
Then given $\psi,\psi'\in {\cal D}$ we find a finite, sufficiently large 
$M$ such that $\psi,\psi'\in {\cal D}_M$. It is sufficient to consider 
the case of tensor fields $q^a_b, p_a^b$ which displays all possible features 
of the calculation. We consider the annihilation operator 
\be \label{8.7}
A_{ab}:=2^{-1/2}[\kappa\; q^a\;_b-i \kappa^{-1}\;p_a\;^b]
\ee
and Fock vacuum $A_{ab}\Omega=0$. Note that 
\be \label{8.8}
[A_{ab}(x),A_{cd}(y)]=\delta_{ac}\; \delta_{bd}\;\delta(x,y)
\ee
For the truncated version $A_{ab,M}=k_M\cdot A_{ab}$ correspondingly
\be \label{8.9}
[A_{ab,M}(x),A_{cd,m}(y)]=\delta_{ac}\; \delta_{bd}\;k_M(x,y),\;
k_M(x,y)=\sum_{I\in S_M}\; e_I(x)\; e_I^\ast(y)
\ee
which is a smooth function, of rapid decrease on $\mathbb{R}^D$ 
or periodic on $T^D$, for above choice of $e_I$. The classical 
SDC is given by 
\ba \label{8.10}
C(u) &=&
\int\;d^Dx\; p_a^b\;(u^c\; q^a_{b,c}-u^a_{,c} q^c_b+u^c_{,b} q^a_c)
\\
&=& \frac{i}{2}\sum_{a,b,c}\;\int\; d^Dx\;
[\kappa(A_{ab}-A_{ab}^\ast)] \;
\{
u^c\;[\kappa^{-1}(A_{ab}+A_{ab}^\ast]_{,c}
-u^a_{,c}\; [\kappa^{-1}(A_{cb}+A_{cb}^\ast]
+u^c_{,b}\; [\kappa^{-1}(A_{ac}+A_{ac}^\ast]
\} 
\nonumber
\ea
Its truncation is, still not in normal ordered form, with the abbreviations 
\be \label{8.11}
A_{ab}^I:=<e_I,A_{ab}>,\; 
F^I_{ab}:=A_{ab}^I-[A_{ab}^I]^\ast,\;\;
G^I_{ab}:=A_{ab}^I+[A_{ab}^I]^\ast
\ee
given by
\ba \label{8.12}
C_M(u) &=&\frac{i}{2}\sum_{I,J\in S_M}\;\sum_{a,b,c}\;
F^I_{ab}\;
\{<(\kappa\cdot e_I),u^c (\kappa^{-1}\cdot e_J)_{,c}>\; G^J_{ab}
-<(\kappa\cdot e_I),u^a_{,c} (\kappa^{-1}\cdot e_J)>\; G^J_{cb}
\nonumber\\
&& +<(\kappa\cdot e_I),u^c_{,b} (\kappa^{-1}\cdot e_J)>\; G^J_{ac}
\ea
This suggests to define 
\be \label{8.13}
U^{ab,cd}_{IJ}:=
<(\kappa\cdot e_I),u^e (\kappa^{-1}\cdot e_J)_{,e}>\; \delta^{ac} \;\delta^{bd}
-<(\kappa\cdot e_I),u^a_{,c} (\kappa^{-1}\cdot e_J)>\; \delta^{bd}
+<(\kappa\cdot e_I),u^d_{,b} (\kappa^{-1}\cdot e_J)>\; \delta^{ac}
\ee
to rewrite (\ref{8.12}) in the compact form
\be \label{8.14}
C_M(u)=\frac{i}{2}\sum_{I,J\in S_M}\; U^{ab,cd}_{IJ}\;
F^I_{ab}\; G^J_{cd}
\ee
We note the non-vanishing commutators
\be \label{8.15}
[A_{ab}^I,(A_{cd}^J)^\ast]=\delta^{IJ}\delta_{ac}\;\delta_{bd}
\ee
At this point we normal order 
\be \label{8.16}
C_M(u)=\frac{i}{2}\sum_{I,J\in S_M}\; U^{ab,cd}_{IJ}\;
(F^I_{ab}\; A^J_{cd}+[A^J_{cd}]^\ast F^I_{ab})
\ee
Consider $\psi,\psi'\in {\cal D}$. We find $M_0\in \mathbb{N}_0$ such 
that $\psi,\psi'\in {\cal H}_{M_0}$. Then for any $M\ge M_0$ we have 
\be \label{8.17}
<\psi,C_M(u)\psi'>=<\psi,C_{M_0}(u)\psi'>
\ee
To see this note that the round bracket is given by 
\be \label{8.18}
A^I_{ab}\; A^J_{cd}
-[A^I_{ab}]^\ast\; A^J_{cd}
+[A^J_{cd}]^\ast A^I_{ab}
-[A^J_{cd}]^\ast [A^I_{ab}]^\ast
\ee
If $I\in S_{M}-S_{M_0}$ we can move in the first and third term $A^I_{ab}$ to
the right hitting $\psi$ which vanishes while
we can move in the second and fourth term $[A^I_{ab}]^\ast$ to
the left where its adjoint hits $\psi'$ which vanishes.    
If $J\in S_{M}-S_{M_0}$ we can move in the first and second term $A^J_{cd}$ to
the right hitting $\psi$ which vanishes while
we can move in the third and fourth term $[A^J_{cd}]^\ast$ to
the left where its adjoint hits $\psi'$ which vanishes.    
It follows that the limit $M\to \infty$ stabilises at $M=M_0$ and thus exists
pointwise in $\psi,\psi'\in {\cal D}$ because $C_{M_0}$ is just a polynomial
in annihilation and creation operators. The same argument reveals that 
if in (\ref{8.16}) we would allow $I,J$ to take full range then the result 
of computing matrix elements of that object coincides with the result
of (\ref{8.17}). Therefore the limit quadratic form is given by
\be \label{8.18a}
C(u)=\frac{i}{2}\sum_{I,J}\; U^{ab,cd}_{IJ}\;
(F^I_{ab}\; A^J_{cd}+[A^J_{cd}]^\ast F^I_{ab})
\ee
~\\
Turning to the commutator we have 
\ba \label{8.19}
&& -4\;[C_M(u),C_M(v)]
=\sum_{I,J,K,L\in S_M}\; 
U^{ab,cd}_{IJ}\; V^{ef,gh}_{KL}\;
[F^I_{ab}\; A^J_{cd}+[A^J_{cd}]^\ast \; F^I_{ab},
[F^K_{ef}\; A^L_{gh}+[A^L_{gh}]^\ast \;F^K_{ef}]
\nonumber\\
&=&
\sum_{I,J,K,L\in S_M}\; 
U^{ab,cd}_{IJ}\; V^{ef,gh}_{KL}\;
\{
[F^I_{ab}\; A^J_{cd}, F^K_{ef}\; A^L_{gh}]
+
[F^I_{ab}\; A^J_{cd},[A^L_{gh}]^\ast \; F^K_{ef})]
\nonumber\\
&& +
[[A^J_{cd}]^\ast \; F^I_{ab},[F^K_{ef}\; A^L_{gh}]
+
[[A^J_{cd}]^\ast \; F^I_{ab},[A^L_{gh}]^\ast \; F^K_{ef}]
\}
\ea
In what follows we often perform the operation of ``exchanging indicices'',
by this we mean to relabel, in this precise order, $I,J,a,b,c,d$ by 
$K,L,e,f,g,h$. We can apply this operation to the first and fourth term 
which results in exchanging $U,V$ with a minus sign, so we can replace the 
first and fourth term by half of it and the result of the exchange operation.
We can apply it also to the third term which results in minus the second term
upon exchanging $U,V$. Thus (\ref{8.19}) becomes
\ba \label{8.20}
&& -8\;[C_M(u),C_M(v)]
=
\sum_{I,J,K,L\in S_M}\; 
[U^{ab,cd}_{IJ}\; V^{ef,gh}_{KL}\;
-V^{ab,cd}_{IJ}\; U^{ef,gh}_{KL}]\;
\nonumber\\
&& \{
[F^I_{ab}\; A^J_{cd}, F^K_{ef}\; A^L_{gh}]
+
2\;[F^I_{ab}\; A^J_{cd},[A^L_{gh}]^\ast \; F^K_{ef}]
+
[[A^J_{cd}]^\ast \; F^I_{ab}),[A^L_{gh}]^\ast \; F^K_{ef}]
\}
\ea
We need the commutators
\be \label{8.21}  
[A^I_{ab},F^J_{cd}]=-\delta^{IJ}\;\delta_{ac}\;\delta_{bd}
=[(A^I_{ab})^\ast,F^J_{cd}],\;\;
[F^I_{ab},F^J_{cd}]=0
\ee
We note that $(F^I_{ab})^\ast=-F^I_{ab}$ so that (\ref{8.21}) is consistent 
with the identity $[A,B]^\ast=[B^\ast,A^\ast]$ for general operators $A,B$.
We can continue using $[A,BC]=[A,B]C+B[A,C]$ repeatedly
\ba \label{8.22}
&& -8\;[C_M(u),C_M(v)]
=
\sum_{I,J,K,L\in S_M}\; 
[U^{ab,cd}_{IJ}\; V^{ef,gh}_{KL}\;
-V^{ab,cd}_{IJ}\; U^{ef,gh}_{KL}]\;
\nonumber\\
&& \{
F^I_{ab}\;[A^J_{cd}, F^K_{ef}]\; A^L_{gh}
+F^K_{ef}\; [F^I_{ab}, \; A^L_{gh}]\; A^J_{cd}
\nonumber\\
&& +2\;F^I_{ab}\; [A^J_{cd},[A^L_{gh}]^\ast] \; F^K_{ef}
+2\;F^I_{ab}\; [A^L_{gh}]^\ast \;[A^J_{cd}, F^K_{ef}]
+2\;[F^I_{ab},A^L_{gh}]^\ast] \; F^K_{ef}] \; A^J_{cd}
\nonumber\\
&& +
[A^J_{cd}]^\ast \; [F^I_{ab}),[A^L_{gh}]^\ast] \; F^K_{ef}]
+[A^L_{gh}]^\ast \;[[A^J_{cd}]^\ast , F^K_{ef})] \; F^I_{ab}
\}
\ea
Again we perform the exchange operation in the second and seventh term
respectively which results in minus the first and sixth term respectively
times minus the antisymmetric combination of $U,V$, hence we can drop the 
second and seventh term and multiply the first and sixth with a factor of two.
Thus (\ref{8.22}) simplifies to 
\ba \label{8.23}
&& -4\;[C_M(u),C_M(v)]
=
\sum_{I,J,K,L\in S_M}\; 
[U^{ab,cd}_{IJ}\; V^{ef,gh}_{KL}\;
-V^{ab,cd}_{IJ}\; U^{ef,gh}_{KL}]\;
\nonumber\\
&& \{
F^I_{ab}\;[A^J_{cd}, F^K_{ef}]\; A^L_{gh}
+
F^I_{ab}\; [A^J_{cd},[A^L_{gh}]^\ast] \; F^K_{ef})
\nonumber\\
&&
+F^I_{ab}\; [A^L_{gh}]^\ast \;[A^J_{cd}, F^K_{ef})]
+[F^I_{ab},[A^L_{gh}]^\ast] \; F^K_{ef})] \; A^J_{cd}
+
[A^J_{cd}]^\ast \; [F^I_{ab}),[A^L_{gh}]^\ast] \; F^K_{ef})
\}
\nonumber\\
&=&
\sum_{I,J,K,L\in S_M}\; 
[U^{ab,cd}_{IJ}\; V^{ef,gh}_{KL}\;
-V^{ab,cd}_{IJ}\; U^{ef,gh}_{KL}]\;
\nonumber\\
&& \{
- \delta^{JK}\delta_{ce}\delta_{df} F^I_{ab}\; A^L_{gh}
+
\delta^{JL}\;\delta_{cg}\delta_{bh}\; F^I_{ab}\; F^K_{ef})
-\delta^{JK}\delta_{ce}\delta_{df}\;F^I_{ab}\; [A^L_{gh}]^\ast 
\nonumber\\
&& +\delta^{LI}\delta_{ga}\delta_{hb}\; F^K_{ef})] \; A^J_{cd}
+
\delta^{LI}\delta_{ga}\delta_{hb}\; A^J_{cd}]^\ast \; F^K_{ef})
\}
\ea
The second term in the curly bracket is symmetric under the exchange operation
because the $F^I_{ab}$ commute while the antisymmetric combination of $U,V$
is antisymmetric under the exchange operation, hence the second term is 
identically zero. The first, fourth and fifth term are already normal ordered 
but the third term is not, hence we perform an additional 
commutator to bring it into normal ordered form which generates a normal
ordering constant
\ba \label{8.24}
&& -4\;[C_M(u),C_M(v)]
=
\sum_{I,J,K,L\in S_M}\; 
[U^{ab,cd}_{IJ}\; V^{ef,gh}_{KL}\;
-V^{ab,cd}_{IJ}\; U^{ef,gh}_{KL}]\;
\\
&& 
\{
-\delta^{JK}\delta_{ce}\delta_{df}\;
(F^I_{ab}\; A^L_{gh}+[A^L_{gh}]^\ast \;F^I_{ab}) 
+
\delta^{JK}\delta_{ce}\delta_{df}\;\delta^{LI}\delta_{ga}\delta_{hb}\cdot 1
+
\delta^{LI}\delta_{ga}\delta_{hb}\; 
(F^K_{ef}] \; A^J_{cd}+[A^J_{cd}]^\ast \; F^K_{ef})
\}
\nonumber
\ea
We perform the sums over $K,L,e,f,g,h$ in the normal ordering term which 
becomes 
\be \label{8.25}
\sum_{I,J\in S_M}\;\sum_{a,b,c,d}\;
\{U^{ab,cd}_{IJ}\;V^{cd,ab}_{JI}-V^{ab,cd}_{IJ}\;U^{cd,ab}_{JI}\}
\ee
We perform the exchange of indices $I,a,b$ by $J,c,d$ in this precise order
and see that (\ref{8.25}) vanishes exactly. If we introduce the compound 
indices $\alpha=(Iab),\beta=(Jcd)$ and set $U^{ab,cd}_{IJ}=U^{\alpha\beta}$
and similar for $V$ then (\ref{8.25}) reduces to the statement
\be \label{8.26}
{\rm Tr}(U\cdot V-V\cdot U)=0
\ee
i.e. the cyclicity of the trace for {\it finite dimensional matrices}.
In a sense, (\ref{8.26}) is the most important part of the calculation 
and the {\it whole purpose of the truncation was to give a well defined 
meaning to the trace operation}. 
    
We can now finish the calculation and perform the remaining sums
\ba \label{8.27}
&& -4\;[C_M(u),C_M(v)]=
-\sum_{I,J,L\in S_M}\;\sum_{a,b,c,d,g,h}\; 
[U^{ab,cd}_{IJ}\; V^{cd,gh}_{JL}\;
-V^{ab,cd}_{IJ}\; U^{cd,gh}_{JL}]\;
(F^I_{ab}\; A^L_{gh}+[A^L_{gh}]^\ast \;F^I_{ab}) 
\nonumber\\
&& +
\sum_{I,J,K\in S_M}\;\sum_{a,b,c,d,e,f} 
[U^{ab,cd}_{IJ}\; V^{ef,ab}_{KI}\;
-V^{ab,cd}_{IJ}\; U^{ef,ab}_{KI}]\;
(F^K_{ef} \; A^J_{cd}+[A^J_{cd}]^\ast \; F^K_{ef})
\}
\ea
We relabel $L,g,h$ to $K,e,f$ in the first term to obtain 
\ba \label{8.27a}
&& -4\;[C_M(u),C_M(v)]=\sum_{I,J,K\in S_M}\;\sum_{a,b,c,d,e,f}\; 
\nonumber\\
&&\{
-[U^{ab,cd}_{IJ}\; V^{cd,ef}_{JK}-V^{ab,cd}_{IJ}\; U^{cd,ef}_{JK}]\;
(F^I_{ab}\; A^K_{ef}+[A^K_{ef}]^\ast \;F^I_{ab}) 
\nonumber\\
&& +
[U^{ab,cd}_{IJ}\; V^{ef,ab}_{KI}\;-V^{ab,cd}_{IJ}\; U^{ef,ab}_{KI}]\;
(F^K_{ef} \; A^J_{cd}+[A^J_{cd}]^\ast \; F^K_{ef})
\}
\ea
In the second term we relabel $K,e,f$ to $I,a,b$ in this precise order 
\ba \label{8.27b}
&& -4\;[C_M(u),C_M(v)]=\sum_{I,J,K\in S_M}\;\sum_{a,b,c,d,e,f}\; 
\nonumber\\
&& \{
-[U^{ab,cd}_{IJ}\; V^{cd,ef}_{JK}-V^{ab,cd}_{IJ}\; U^{cd,ef}_{JK}]\;
(F^I_{ab}\; A^K_{ef}+[A^K_{ef}]^\ast \;F^I_{ab}) 
\nonumber\\
&& +
[U^{ef,cd}_{KJ}\; V^{ab,ef}_{IK}\;-V^{ef,cd}_{KJ}\; U^{ab,ef}_{IK}]\;
(F^I_{ab} \; A^J_{cd}+[A^J_{cd}]^\ast \; F^I_{ab})
\}
\ea
and once more in the second term we relabel $J,c,d$ by $K,e,f$
\ba \label{8.28}
&& -4\;[C_M(u),C_M(v)]
=
\sum_{I,J,K\in S_M}\;\sum_{a,b,c,d,e,f}\; 
\nonumber\\
&& \{
-[U^{ab,cd}_{IJ}\; V^{cd,ef}_{JK}-V^{ab,cd}_{IJ}\; U^{cd,ef}_{JK}]\;
(F^I_{ab}\; A^K_{ef}+[A^K_{ef}]^\ast \;F^I_{ab}) 
\nonumber\\
&& +
[U^{cd,ef}_{JK}\; V^{ab,cd}_{IJ}\;-V^{cd,ef}_{JK}\; U^{ab,cd}_{IK}]\;
(F^I_{ab} \; A^K_{ef}+[A^K_{ef}]^\ast \; F^I_{ab})
\}
\nonumber\\
&=&
-2\;\sum_{I,J,K\in S_M}\;\sum_{a,b,c,d,e,f}\; 
[U^{ab,cd}_{IJ}\; V^{cd,ef}_{JK}-V^{ab,cd}_{IJ}\; U^{cd,ef}_{JK}]\;
\nonumber\\
&& (F^I_{ab}\; A^K_{ef}+[A^K_{ef}]^\ast \;F^I_{ab}) 
\ea
Let now $\psi,\psi'\in {\cal D}$ then we find $M_0\in \mathbb{N}_0$ such that 
$\psi,\psi'\in {\cal H}_{M_0}$ and by the same argument as before 
we see that for $M\ge M_0$
in the matrix element of (\ref{8.28}) the indices $I,K$ are 
constrained to $S_{M_0}$ while $J\in S_M$ still
\be \label{8.29}
<\psi,[C_M(u),C_M(v)]\psi'>=\frac{1}{2}\;
\sum_{I,K\in S_{M_0},\;J\in S_M}\;\sum_{a,b,c,d,e,f}\; 
[U^{ab,cd}_{IJ}\; V^{cd,ef}_{JK}-V^{ab,cd}_{IJ}\; U^{cd,ef}_{JK}]\;
<\psi,(F^I_{ab}\; A^K_{ef}+[A^K_{ef}]^\ast \;F^I_{ab})\psi'> 
\ee
because (\ref{8.28}) is normal ordered. We now perform the limit $M\to \infty$
in the sum
\be \label{8.30}
\lim_{M\to\infty} \sum_{J\in S_M}\sum_{c,d}\;
[U^{ab,cd}_{IJ}\; V^{cd,ef}_{JK} -U\;\leftrightarrow\; V]
\ee
To do this we rewrite, using that $\kappa$ is a positive, symmetric, real 
valued, invertible kernel
on $L_2(\sigma)$ 
\ba \label{8.31}
U^{ab,cd}_{IJ} &=&
-<\kappa^{-1}[u^g(\kappa\cdot e_I)]_{,g},e_J>\; \delta^{ac} \;\delta^{bd}
-<\kappa^{-1}u^a_{,c}(\kappa\cdot e_I), e_J>\; \delta^{bd}
+<\kappa^{-1} u^d_{,b}(\kappa\cdot e_I), e_J>\; \delta^{ac}
\nonumber\\    
V^{cd,ef}_{IJ} &=&
<e_J,\kappa\;v^h (\kappa^{-1}\cdot e_K)_{,h}>\; \delta^{ce} \;\delta^{df}
-<e_J,\kappa v^c_{,e} (\kappa^{-1}\cdot e_K)>\; \delta^{df}
+<e_J,\kappa v^f_{,d} (\kappa^{-1}\cdot e_K)>\; \delta^{ce}
\nonumber\\
&&
\ea
As $M\to\infty$ we can now sum over $J$ unconstrained and use the completeness
relation $\sum_J e_J\; <e_J,.>_{L_2}=1_{L_2}$. This yields for (\ref{8.30})
altogether nine terms which we list in the sequence: 
first term from $U$ in (\ref{8.31}) times first to third term from 
V in (\ref{8.31}), 
second term from $U$ times first to third term from V, 
third term from $U$ times first to third term from V minus the roles of 
$u,v$ interchanged. 
\ba \label{8.32}
1. && 
-<\kappa^{-1}[u^g(\kappa\cdot e_I)]_{,g},
\kappa\;v^h (\kappa^{-1}\cdot e_K)_{,h}>\; \delta^{ae} \;\delta^{bf}
\;\; - \;\;u\;\;\leftrightarrow\;\; v 
\nonumber\\
&=&
<\kappa\cdot e_I,
u^g [v^h (\kappa^{-1}\cdot e_K)_{,h}]_{,g}>\; \delta^{ae} \;\delta^{bf}
\;\; - \;\;u\;\;\leftrightarrow\;\; v 
\nonumber\\
&=&
<\kappa\cdot e_I,[u,v]^g (\kappa^{-1}\cdot e_K)_{,g}>\; \delta^{ae} \;\delta^{bf}
\nonumber\\
2. && 
<\kappa^{-1}[u^g(\kappa\cdot e_I)]_{,g},\kappa v^a_{,e} 
(\kappa^{-1}\cdot e_K)>\; \delta^{bf}
\;\; - \;\;u\;\;\leftrightarrow\;\; v 
\nonumber\\
&=& 
-<\kappa\cdot e_I, u^g (v^a_{,e} 
(\kappa^{-1}\cdot e_K))_{,g}>\; \delta^{bf}
\;\; - \;\;u\;\;\leftrightarrow\;\; v 
\nonumber\\
3. && 
-<\kappa^{-1}[u^g(\kappa\cdot e_I)]_{,g},
\kappa v^f_{,b} (\kappa^{-1}\cdot e_K)>\; \delta^{ae}
\;\; - \;\;u\;\;\leftrightarrow\;\; v 
\nonumber\\
&=&
<\kappa\cdot e_I,
u^g (v^f_{,b} (\kappa^{-1}\cdot e_K)_{,g}>\; \delta^{ae}
\;\; - \;\;u\;\;\leftrightarrow\;\; v 
\nonumber\\
4. && 
-<\kappa^{-1}u^a_{,e}(\kappa\cdot e_I),
v^g (\kappa^{-1}\cdot e_K)_{,g}> \;\delta^{bf}
\;\; - \;\;u\;\;\leftrightarrow\;\; v 
\nonumber\\
&=& 
-<\kappa\cdot e_I, u^a,e\; v^g (\kappa^{-1}\cdot e_K)_{,g}> \;\delta^{bf}
\;\; - \;\;u\;\;\leftrightarrow\;\; v 
\nonumber\\
5. && 
+<\kappa^{-1}u^a_{,g}(\kappa\cdot e_I),
\kappa v^g_{,e} (\kappa^{-1}\cdot e_K)>\; \delta^{bf}
\;\; - \;\;u\;\;\leftrightarrow\;\; v 
\nonumber\\
&=& 
<\kappa\cdot e_I, u^a_{,g}\; v^g_{,e} (\kappa^{-1}\cdot e_K)>\; \delta^{bf}
\;\; - \;\;u\;\;\leftrightarrow\;\; v 
\nonumber\\
6. && 
-<\kappa^{-1}u^a_{,e}(\kappa\cdot e_I),\kappa v^f_{,b} (\kappa^{-1}\cdot e_K)>
\;\; - \;\;u\;\;\leftrightarrow\;\; v 
\nonumber\\
&=& 
-<\kappa\cdot e_I, u^a_{,e} v^f_{,b} (\kappa^{-1}\cdot e_K)>
\;\; - \;\;u\;\;\leftrightarrow\;\; v 
\nonumber\\
7. && 
<\kappa^{-1} u^f_{,b}(\kappa\cdot e_I),
\kappa\;v^h (\kappa^{-1}\cdot e_K)_{,h}>\; \delta^{ae}
\;\; - \;\;u\;\;\leftrightarrow\;\; v 
\nonumber\\
&=& 
+<\kappa\cdot e_I, u^f_{,b}
\;v^g (\kappa^{-1}\cdot e_K)_{,g}>\; \delta^{ae}
\;\; - \;\;u\;\;\leftrightarrow\;\; v 
\nonumber\\
8. && 
-<\kappa^{-1} u^f_{,b}(\kappa\cdot e_I),
\kappa v^a_{,e} (\kappa^{-1}\cdot e_K)>\;
\;\; - \;\;u\;\;\leftrightarrow\;\; v 
\nonumber\\
&=& 
-<\kappa\cdot e_I, u^f_{,b}\;v^a_{,e} (\kappa^{-1}\cdot e_K)>\;
\;\; - \;\;u\;\;\leftrightarrow\;\; v 
\nonumber\\
9. &&
<\kappa^{-1} u^d_{,b}(\kappa\cdot e_I),
\kappa v^f_{,d} (\kappa^{-1}\cdot e_K)>\; \delta^{ae}
\;\; - \;\;u\;\;\leftrightarrow\;\; v 
\nonumber\\
&=&
<\kappa\cdot e_I, u^g_{,b}\;v^f_{,g} (\kappa^{-1}\cdot e_K)>\; \delta^{ae}
\;\; - \;\;u\;\;\leftrightarrow\;\; v 
\ea
We see that the sixth term cancels against the eighth (both without Kronecker 
$\delta$), the three terms proportional to $\delta^{bf}$ combine to 
\be \label{8.33}
-<\kappa\cdot e_I, [u,v]^a_{,e} (\kappa^{-1} e_K)> \; \delta^{bf}
\ee
using the product rule 
and the three terms proportional to $\delta^{ae}$ combine to 
\be \label{8.34}
+<\kappa\cdot e_I, [u,v]^f_{,b} (\kappa^{-1} e_K)> \; \delta^{ae}
\ee
Altogether (\ref{8.30}) becomes
\be \label{8.35}
<\kappa \cdot e_I,[u,v]^g (\kappa^{-1}\cdot e_J)_{,g}>\;
\delta^{ae}\;\delta^{bf}
-<\kappa \cdot e_I,[u,v]^a_{,e} (\kappa^{-1}\cdot e_J)>\;\delta^{bf}
+<\kappa \cdot e_I,[u,v]^f_{,b} (\kappa^{-1}\cdot e_J)>\;\delta^{ae}
\equiv W^{ab,ef}_{IK}
\ee
where $W$ is obtained from $U$ in (\ref{8.13}) by replacing $u$ by $[u,v]$.

We conclude 
\be \label{8.36}
<\psi,[C(u),C(v)]\psi'>=i\;\frac{i}{2}\;
\sum_{I,K\in S_{M_0}}\;\sum_{a,b,e,f}\; (-W^{ab,ef}_{IK})\; 
<\psi,(F^I_{ab}\; A^K_{ef}+[A^K_{ef}]^\ast \;F^I_{ab})\psi'> 
=i\;<\psi,\;C(-[u,v])\psi'>
\ee
i.e. the quadratic form commutator exists and is physically non-anomalous.
The choice of $\kappa, e_J, S_M$ has played no role.\\
\end{proof}
~\\
The result is quite surprising in the following sense: The Fock 
representations that we have chosen are {\it heavily background dependent}:
We had to choose a coordinate system, a background metric (here $\delta_{ab}$)
a kernel $\kappa$, a basis of the $L_2$ spaces equipped with that background 
metric. One would therefore expect that such a background dependent 
representation could not lead to a non-anomalous representation of the 
SDA. That one arrives at a non-anomalous representation (although only 
as quadratic forms) nevertheless shows that this intuition is {\it completely
false}. The situation reminds of the phenomenon of symmetry breaking:
While a Hamiltonian may be invariant under a given symmetry, any of its 
ground states may not display this symmetry \cite{31}. Here we have 
an analogous situation: The algebra of quadratic forms is exactly 
diffeomorphism covariant, however none of the states in their common form
domain $\cal D$ are diffeomorphism invariant.

The second aspect that is surprising is that this works for an 
infinite number of Fock representations which we choose to be 
mutually unitarily inequivalent. For this to be the case it suffices 
to consider two kernels $\kappa_1,\kappa_2$ such that the operator
$\kappa_1\kappa_2^{-1}-\kappa_2\kappa_1^{-1}$ on $L_2(\sigma)$
fails to be of Hilbert 
Schmidt type (e.g. in the case of a scalar field 
two representations for the KG of mass 
$m_1\not=m_2$). This shows that the uniqueness result of \cite{32} which 
analyses the class of diffeomorphism invariant states for a non Abelian 
version of the Weyl algebra for a gauge field theory fails to hold if 
one does not insist on invariance of the state (in this section
the state is 
the vacuum epectation functional with respect to the chosen Fock 
vacuum vector). 
 
The invariant vector states are no longer elements of the Hilbert space 
but rather
linear functionals on $\cal D$. We could in principle solve the quantum 
constraint equations $l[C(u)\psi]=0$ for all $u\in V, \psi\in {\cal D}$
as follows: As $l$ is a linear functional on $\cal D$ we write it as 
$l=\sum_n\; l_n\; <b_n,.>_{{\cal H}}$ where $b_n$ is the ONB of Fock states.
Given that we have already introduced the structure $e_I$ we can define 
a basis of vector fields $(v_{Ia})^b:=\delta^a_b\; e_I$. Then the constraint 
equations are equivalent to the infinite system of linear equations 
\be \label{8.37}
0=\sum_n\; l_n \; C^{Ia}_{n,n'}=0,\; C^{Ia}_{n,n'}=<b_n,C(v_{Ia})\;b_{n'}>
\ee
for all $n',I,a$. Naively counting there are D times infinite many square 
matrices $C^{Ia}$ on the space of coefficients $l_n$ and thus the system 
(\ref{8.37}) appears to be heavily over determined. However, the 
matrices are sparsely populated as $C(u)$ can change individual 
occupation numbers only by $0,\pm 1, \pm 2$ and total particle number only 
by $0,\pm 2$. So $C^{Ia}_{n,n'}=0$ unless $||n||=||n'||,||n'||\pm 2$ where 
$||n||=\sum_{I,\mu} n^\mu_I$ is the particle number ($\mu$ again labels 
the components fof the given tensor field). In particular all matrices 
are block diagonal where the blocks correspond to even and odd particle 
number. Next, the coefficients $U^{\mu,\nu}_{IJ}(u)$ in (\ref{8.13})
evaluated at $u=v_{Ka}$ may themselves vanish, at fixed $I,K$, 
for $J$ outside a certain 
neighbourhood in $S$ depending on $I,K$. Consider for instance $\sigma=T^D$
and $\kappa$ diagonal on the momentum eigenfunction basis $e_I$
given up to normalisation by 
\be \label{8.39}
e_I(x)=\prod_a e_{\sigma_a, n_a}(x^a),\sigma_a=\pm, 
n_a\in \mathbb{N}_0,\; 
e_{+,n}(x)=\cos(nx),\;e_{-,n}(x)=\sin(nx),\;
\ee
Then the 
$U^{\mu,\nu}_{IJ}(v_{ka})$ are proportional to integrals of the form
$<e_K,e_I \; e_J>$ and $e_I\; e_J$ is a linear combination of $2^D$ functions 
$e_L$ such that if $e_I,e_J$ carry labels $n^a,m^a$ then $e_L$ carries labels 
$l^a=|n^a-m^a|$. Thus in this case, for given $K$ only those pairs of labels 
$I,J$ contribute 
to the sum over $I,J$ in $C^{Ka}$ for which ``$I,J$ differ by $K$''.\\
~\\
The space of solutions $l$ also requires a new inner product as 
$||l||=\sum_n \; |l_n|^2=\infty$. As one cannot average over 
the diffeomorphism group if we have only have access to the quadratic form
generators one may try to define a renormalised inner product by 
\be \label{8.40}
<l^1,l^2>_{{\rm ren}}:=\lim_{N\to \infty}\; 
\frac{\sum_{||n||\le N}\; (l^1_n)^\ast\; l^2_n}
{\sum_{||n||\le N}\; (l^0_n)^\ast\; l^0_n}
\ee
where $l^0$ is a reference solution.\\
\\
We leave the solution of the constraint equations as a future project.

\section{Conclusions and Outlook}
\label{s9}

We have shown in this paper that the representation theory of the SDA 
and SDG in represnentations of the Weyl algebra for bososn or the CAR 
algebra for fermions 
has many fascets. One could organise them into three classes:\\
1. Discontinuous representations of the SDG as unitary operators and thus
not of the SDA at all.\\
2. Strongly continuous representations of the SDG as unitary operators 
and thus of the SDA as self-adjoint operators.\\
3. Representations of the SDA as symmetric quadratic forms and thus not 
of the SDG at all. \\
\\
Eamples for the class 1. arise in Narnofer-Thirring type of represenations.\\
Eamples for the class 2. are half-density representations of CAR and 
Weyl algebra which for fermions are natural but for bosons somewhat 
artificial and require non-trivial canonical transformations on the
bosonic phase space prior to quantisation.\\
Examples for the class 3. are in fact arbitrary Fock representations.
The underlying algebraic state on the CAR or Weyl algebra in case 1. and 2. 
is in fact diffeormorphism invariant while in case 3. it is not. \\
\\
To understand the implications of this result we note that while 
to have a unitary representation of the SDG via discontinuous representations 
of the Weyl algebra is a strong argument in favour of LQG, it comes with some 
inconveniences, most importantly the non-separability of the Hilbert space 
and the related quantisation ambiguities. These come from the fact 
that one must 
define the HC in terms of Weyl operators, i.e. eponentiated field 
operators rather rather than field operators themselves. As the classical 
HC is defined in terms of the non eponentiated fields, this is a difficult 
process which uses non-trivial aspects of the SDG. It would be more convenient
to have a separable Hilbert space at one's disposal and such that the 
non eponentiated fields are defined, which in Fock representations is the 
case. We have shown that such an approach does not meet any obstacles 
with respect to the diffeomorphism group once one realises that it is 
enough to have the generators of the SDA defined as quadratic forms.

In going beyond the SDA and concerning representing the HC as a quadratic 
form we expect to face much harder problems: While it is possible to write 
the HC in polynomial form and thus as a well defined quadratic form on the 
Fock space, to define the HC algebra in terms of quadratic forms will 
be more difficult for two reasons: First, the normal ordering corrections 
picked up in the commutator calculation of the SDA were just constants 
because the SDA are bilinear in creation and annihilation operators, 
but the HC is a higher order polynomial and thus the corrections are now
operators and it is no longer clear that they cancel out as well. 
Second, if a well defined HC commutator resulting from a polynomial HC
exists, it will be a normal ordered polynomial of degree at least four and 
thus not proportional to a SDA generator ordered to the left, i.e. it 
will not automatically annihilate diffeomorphism invariant distributions.
As we have indictated if one wants to avoid structure 
functions one is forced to use the master constraint instead 
of the HC, however in diffeomorphism invariant form and in that 
form it is not polynomial so that it is unclear how to define a 
densely defined quadratic form from it. 

These questions will be adressed
in \cite{32a} where it is shown 1. how to define the polynomial 
HC as a quadratic form and its commutator algebra and 2. how to 
define non-polynomial expresions as quadratic forms 
in Fock representations so that a 
master consstraint or a reduced phase space Hamiltonian becomes available.
As it is demonstrated there, 1. requires mor flexibility in the 
definition for how to define a quadratic form commutator while 2. 
requires an entirely new techniqe and domains in the Fock space,
in particular the notion of quantum non-degeneracy \cite{33} 
plays a prominent role. More in detail, Fock representations
contain semi-classical states with minimal 
fluctuations around any non-degenerate 
background metric $g_1$ which may or may not coincide
with the background metric $g$ that was used to construct the representation.
Moreover, since Fock representations are 
irreducible, every such 
semi-classical state is a cyclic vector and thus 
its excitations lie dense, forming the form domain of say the master 
constraint or the reduced Hamiltonian in a reduced phase space approach 
\cite{34}. We will also have the opportunity to define 
quadratic forms corresponding to geometric observables (lengths, areas,
volumes,..) in Fock representations.
   
In closing, 
note that none of these techniques concerning the HC make use of 
perturbation theory, while they make use of background dependent structures
(here Fock representations). This is because background dependence and 
perturbation theory are logically independent: While a background 
metric $g$ can be used to define a perturbatve metric $h=q-g$ from the 
non-perturbative metric $q$ and perform perturbation theory with respect 
to $h$, one does not need to do it! One can keep the non-perturbative metric 
intact while introducing background dependent representations of the 
Weyl relations of $q$ and its conjugate momentum.

\begin{appendix}

\section{Normal ordering tools for QFT}
\label{sa}

The following tools are an extension to QFT of the techniques developed in 
\cite{28} for quantum mechanics. The idea of proof is similar to that for 
the Baker-Campbell-Hausdorff formula that relates exponentials
of general operators. Here we deal with exponentials of bilinear epressions
in annihilation and creation operators. The purpose is to rewrite products 
of such eponentials in normal orderd form which then allow for a much simpler 
evaluation of its matrix elements between Fock states.\\
\\
We consider integral kernels $a,b,c$ and consider for a scalar theory 
(for higher tensor theories the following still applies for matrix valued 
kernels which are diagonal with respect to the tensor structure or which 
can be simultaenously diagonalised in that sense)
\ba \label{a.1}
&& a(A^\ast,A^\ast):=\int\; d^Dx;d^Dy\; a(x,y)\; A^\ast(x)\; A^\ast(y),\;\;
b(A^\ast,A):=\int\; d^Dx;d^Dy\; b(x,y)\; A^\ast(x)\; A(y),\;\;
\nonumber\\
&& c(A,A):=\int\; d^Dx;d^Dy\; c(x,y)\; A(x)\; A(y)
\ea
and an eponential of the form
\be \label{a.2}
e^B,\; B:=\frac{1}{2}a(A^\ast,A^\ast)+b(A^\ast,A)+\frac{1}{2}c(A,A)
\ee
Note that $a,c$ are automatically symmetric kernels.
It would be quite involved to directly compute matrix elements of (\ref{a.2})
in the Fock basis using the Taylor expansion of the exponential. However, 
using the Taylor expansion and reordering one can see that it can be written
as a sum of monomials of the symbolic form $[A^\ast]^{2K}\; [A^\ast A]^L\;
A^{2M}$ with $K,L,M\ge 0$
where we suppress the spatial dependences and the kernels $a,b,c$ that 
correlate them in all possible permutations of their two spatial varables. 
The question arises whether the sum over $k,l,m$ and those permutations 
can be carried out and organised into a useful
expression.
\begin{Proposition} \label{prop.a1} ~\\
Suppose that\\ 
either 1. $b=0$ and 1i. $c,-a$ are invertible\\
or 2. $b\not=0$ and 2i. $c,\;b\; \cdot c^{-1}\cdot b^T-a$ are invertible
and 2ii. $a,b,c$ mutually commute.\\  
Then there exists a constant $k$ and kernels $l,m,r$ such that 
\be \label{a.3}
e^{\frac{1}{2}\;a(A^\ast,A^\ast)+b(A^\ast,A)+\frac{1}{2}\;c(A,A)}=
E_0\; E_+\; E\; E_-;\;\;
E_0=e^{\frac{k}{2}},\;
E_+=e^{\frac{1}{2}\;l(A^\ast,A^\ast)},\;
E=e^{m(A^\ast,A)},\;
E_-=e^{\frac{1}{2}\;r(A,A)}
\ee
\end{Proposition}
\begin{proof}
We set for a real parameter $s$
\be \label{a.4}
L(s)=e^{sB},\;R(s)=
E_0(s)\;E_+(s)\;E(s)\;    
E_r(s)
\ee
where 
\be \label{a.5}
E_0(s)=e^{\frac{1}{2}\;k_s},\; 
E_+(s)=e^{\frac{1}{2}\;l_s(A^\ast,A^\ast)},\;
E(s)=e^{m_s(A^\ast,A)},\;
E_-(s)=e^{\frac{1}{2}\;r_s(A,A)},\;
\ee
where we have made the constant and the kernels also depend in $s$. We
now require $L(s)=R(s)$ for all $s$. We derive an ODE system
by taking the derivative 
of $L(s)-R(s)=0$ and solve that ODE with initial condition $L(0)=1=R(0)$ 
i.e. $k_0=l_0=m_0=r_0=0$. Then the desired solution (\ref{a.3}) is obtained 
for $s=1$.

Clearly $L'(s)=B L(s)$ while
\be \label{a.6}
R'(s)=
\{
k'(s)+E_+'(s)
+E_+(s)\; E'(s)\; E_+(s)^{-1}+E_+(s)\;E(s)\; E_-'(s)\; E(s)^{-1} E_+(s)^{-1}
\}\;R(s)
\ee
Note that while e.g. the kernels $l_s, l'_s$ are different, not necessarily 
commuting integral kernels, they are
integrated against $A^\ast(x) A^\ast(y)$ and thus these epressions
mutually commute. The same 
applies to $r_s$. For $m_s$ we have 
\ba \label{a.7}
&&[m_s(A^\ast,A),m'_s(A^\ast,A)]=
\int\; d^{2D}x\;  
\int\; d^{2D}y\; 
\times\nonumber\\
&& m_s(x_1,y_1)\; 
m_s'(x_2,y_2)\;\{A^\ast(x_1)\delta(y_1,x_2) A(y_2) 
-A^\ast(x_2)\delta(x_1,y_2) A(y_1)\}
=[m_s,m_s'](A^\ast,A)
\ea
We will assume that $[m_s,m'_s]=0$ and check whether this is the 
case later on. Then after  multiplying $L'=R'$ by 
$L(s)^{-1}=R(s)^{-1}$ from the right 
we get (not displaying $s$ for simplicity) 
\be \label{a.8}
B=\frac{1}{2}\;k'+\frac{1}{2}\;l'+e^{l/2}\;m'\;e^{-l/2}
+\frac{1}{2}\;e^{l/2}\; e^m\; r'\; e^{-m}\; e^{-l/2}
\ee
We have explicity 
\ba \label{a.9}
&& e^{l/2}\;m'\;e^{-l/2}=m'(A^\ast,A-l\cdot A^\ast))=
m'(A^\ast,A)-\frac{1}{2}(m'\cdot l+l\cdot m^{\prime T})(A^\ast,A^\ast)
\\
&& e^{l/2}\; e^m\; r'\; e^{-m}\; e^{-l/2}  
=e^{l/2}\; r'(
e^m\; A\;e^{-m},e^m\; A\;e^{-m})\;e^{-l/2}
=e^{l/2}\; r'(e^{-m}\cdot A,e^{-m}\cdot A)\;e^{-l/2}
\nonumber\\
&=& (e^{-m^T}\cdot r'\cdot e^{-m})(A-l\cdot A^\ast,A-l\cdot A^\ast)
\nonumber\\
&=& 
(e^{-m^T}\cdot r'\cdot e^{-m})(A,A)
+(l\cdot\;e^{-m^T}\cdot r'\cdot e^{-m}\;\cdot l)(A^\ast,A^\ast)
-2(l\cdot\;e^{-m^T}\cdot r'\cdot e^{-m})(A^\ast,A)
\nonumber\\
&& -{\rm Tr}(l\cdot e^{-m^T}\cdot r'\cdot e^{-m})\;1
\nonumber
\ea
where $\cdot$ denotes the kernel product $(l\cdot r)(x,y)=\int\; d^Dz\;
l(x,z) r(z,y)$ and application $(r\cdot A)(x)=\int\; d^Dz\; r(x,z) A(z)$ 
respectively while $m^T(x,y)=m(y,x)$ is the kernel transpose. 
We have used the symmetry of $r,l$ and restored normal order
in the last step. Collecting terms and comparing the independent bilinears 
on both sides we find the four equations
\ba \label{a.10}
0 &=& k' -{\rm Tr}(l\cdot e^{-m^T}\cdot r'\cdot e^{-m})
\nonumber\\
a &=& l'-(m'\cdot l+l\cdot m^{\prime T})
+ l\cdot\;e^{-m^T}\cdot r'\cdot e^{-m}\;\cdot l
\nonumber\\
b &=& m' - l\cdot\;e^{-m^T}\cdot r'\cdot e^{-m}
\nonumber\\
c &=& (e^{-m^T}\cdot r'\cdot e^{-m})
\ea
The last equation enforces that the complicated epression 
$e^{-m^T}\cdot r' \cdot e^{-m}$ is a constant $c$. We can simplify the 
second and third equation
\be \label{a.11}
a=l'-(m'\cdot l+l\cdot m^{\prime T})+l\cdot c \cdot l,\;
b = m' - l\cdot c
\ee
Eliminating $m'$ we get a closed equation for $l$
\be \label{a.12}
a=l'-([b+l\cdot c]\cdot l+l\cdot [b+l\cdot c]^T)+l\cdot c \cdot l,\;
=l'-l\cdot c \cdot l-b\cdot l-l\cdot b^T
=(l+b\cdot c^{-1})'-(l+b\cdot c^{-1})\cdot c \cdot (l+b\cdot c^{-1})^T
+b\cdot c^{-1}\cdot b^T
\ee
where we used symmetry of $c,l$ and invertibility of $c$.
This is a first order inhomogeneous ODE of Riccati type with linear and 
quadratic dependence on $l$ but with constant coefficients $a,b,c$. The 
solution is therefore obtained by quadrature: Let $\hat{l}=l+b\cdot c^{-1},\;
d=b\cdot c^{-1}\cdot b^T-a$ then 
\be \label{a.13}
\hat{l}'+\hat{l}^{\prime T}=2(-d +\hat{l}\cdot c\cdot \hat{l}^T)
\ee
where we used that the non-symmetric pice of $\hat{l}$ drops out of the 
derivative.
If we assume that $d$ is invertible we pick a square root and introduce 
$\tilde{l}=d^{-1/2}\cdot l\cdot d^{-1/2},\; 
\tilde{c}=d^{1/2}\cdot c\cdot d^{1/2}$ 
then 
\be \label{a.14}
\tilde{l}'+\tilde{l}^{\prime T}=2(-1+\tilde{l}\cdot \tilde{c}\cdot \tilde{l})
\ee
If we assume that $b\cdot c^{-1}=(b\cdot c^{-1})^T$ 
and $\tilde{c}$ is invertible then 
we pick a square root and can solve by 
\be \label{a.15}
\tilde{l}(s)=-\tilde{c}^{-1/2}\;{\rm th}(\tilde{c}^{1/2}(s+s_0))  
\ee
where $s_0$ is an integration constant.
Thus $\tilde{l}$ is a function of the symmetric matrix $\tilde{c}$ only and 
therefore $\tilde{l}'$ self-consistently commutes with $\tilde{l}$. We 
express this in terms of $l$ again
\be \label{a.16}
l=d^{1/2}\tilde{l} d^{1/2}-b\cdot c^{-1}
\ee
Now the intial condition $l(0)=0$ requires 
\be \label{a.17}
\tilde{c}^{-1/2}{\rm th}(\tilde{c}^{1/2} s_0)=-d^{-1/2}\cdot b 
\cdot c^{-1}\cdot d^{-1/2}
\ee
As the l.h.s. is symmetric, this again requires that $b\cdot c^{-1}$ 
is symmetric as 
already assumed. If $b=0$ we may pick $s_0=0$ with no further condition on 
$a,b,c$. Otherwise we can pick $s_0=1$ and then (\ref{a.17}) is a 
relation among $a,b,c$ to be obeyed which in general can 
only hold if $a,b,c$ mutually commute, hence are all functions of the 
same symmetric kernel.
This latter condition
is not necessary in quantum mechanics for which $a,b,c$ are simply numbers 
and in which case (\ref{a.17}) can be simply solved for $s_0$. 

Having solved (\ref{a.17}) either way we may integrate
\be \label{a.18}
m'=b+l\cdot c=\hat{l}\cdot c
=-d^{1/2}\cdot\tilde{c}^{-1/2}{\rm th}(\tilde{c}^{1/2}(s+s_0)) 
\cdot d^{1/2}\cdot c
=-d^{1/2}\cdot\tilde{c}^{1/2}{\rm th}(\tilde{c}^{1/2}(s+s_0)) 
\cdot d^{-1/2}
\ee
to obtain
\be \label{a.19}
m=d^{1/2}\cdot\;[m_0-
\ln({\rm ch}(\tilde{c}^{1/2}(s+s_0)))]\cdot d^{-1/2}
\ee
and $m_0=\ln({\rm ch}(\tilde{c} s_0))$. For the case $b=0$
we have $s_0=0$ and $m_0=0$. The condition $[m,m']=0$ assumed above is 
satisfied because $m(s),m'(s)$ are of the form $d^{1/2}\cdot 
f(\tilde{c})\cdot d^{-1/2}$.  
Next we integrate 
\be \label{a.20}
r'=e^{m^T}\; \cdot c\; e^{m}
\ee
The general solution that also meets $r(0)=0$ is 
\be \label{a.21}
r(s)=\int_0^s\; ds'\;e^{m^T(s')}\; c\; e^{m(s')}
\ee
which in general cannot be integrated in closed form. For $b=0$ and thus 
$d=-a$ and if $[a,c]=0$ we have $e^m=[{\rm}(ch(\tilde{c}^{1/2}s)]^{-1}$ 
and thus 
\be \label{a.22}
r(s)=\frac{c}{\tilde{c}^{1/2}}\;{\rm th}(\tilde{c}^{1/2}s)
\ee
Finally
\be \label{a.23}
k'={\rm Tr}(l\cdot c)
=-{\rm Tr}(b)
+{\rm Tr}(d^{1/2}\;\tilde{l}\; d^{1/2}\cdot c)
=-{\rm Tr}(b)
+{\rm Tr}(\tilde{l}\; \tilde{c})
\ee
which is solved by 
\be \label{a.24}
k(s)=
=k_0-{\rm Tr}(b)\; s
-{\rm Tr}(\ln({\rm ch}(\tilde{c}^{1/2}(s+s_0)))
\ee
If $b=0$ we have $s_0=0$ and thus $k_0=0$. 
\end{proof}
~\\
We note the special case $b=0,-a=c=c^\ast$
\be \label{a.25}
l(s)=-{\rm th}(cs),\; m(s)=-\ln({\rm ch}(cs),\;
r(s)={\rm th}(cs),\;
k(s)=-{\rm Tr}(c {\rm th}(cs))
\ee
thus 
\be \label{a.26}
W:=e^{-\frac{1}{2}c(A^\ast,A^\ast)+\frac{1}{2}c(A,A)}
=e^{k(1)/2}\;
e^{l_1(A^\ast,A^\ast)/2}\;  
e^{m_1(A^\ast,A)}\;  
e^{-l_1(A,A)/2}
\ee
Note that $W^\ast=W^{-1}$ which translates into the identity
\be \label{a.27}
e^{k_1/2}\;
e^{-l_1(A^\ast,A^\ast)/2}\;  
e^{m_1(A^\ast,A)}\;  
e^{l_1(A^\ast,A^\ast)/2} 
=
e^{-k_1/2}\;
e^{l_1(A,A)/2}\;  
e^{-m_1(A^\ast,A)}\;  
e^{-l_1(A^\ast,A^\ast)/2}
\ee
While the l.h.s. is normal ordered with respect to the $A^2, (A^\ast)^2$ 
terms, the r.h.s. is anti-normal ordered. 
Indeed if we had decomposed $W^{-1}$ into anti-normal ordering
\be \label{a.28}
L^{-1}=W^{-1}=(E_-)^{-1}\;E^{-1}\;E_+^{-1}\;E_0^{-1}=R^{-1}
\ee
then the same method of proof would give
\be \label{a.29}
-L^{-1}\;B=-R^{-1}[k'+l'+E_+ m E_+^{-1}+ E_+ \; E\; r'\; E^{-1}\; E_+^{-1}]
\ee
which is identical to (\ref{a.8}) and thus yields the same solution for 
$k,l,m,r$. On the other hand $W^{-1}$ is the same as $W$ when replacing 
$c$ by $-c$. This is confirmed by (\ref{a.25}) which displays $l_1=-r_1$ 
as odd under reflection of $c$ while $m_1,k_1$ are even.\\
\\
\begin{Proposition} \label{prop.a2} ~\\
For any real kernel $c$ there exists a real constant $k$ and real kernels 
$l,m$ such that 
\be \label{a.30}
W:=e^{\frac{1}{2}\;c(A,A)}\;e^{\frac{1}{2}\;c(A^\ast,A^\ast)}\;
=e^{\frac{1}{2}\;l(A^\ast,A^\ast)}\; e^{m(A^\ast,A)}\;
e^{\frac{1}{2}\;l(A,A)}
\ee
\end{Proposition}
\begin{proof}
We can deal with one kernel less than in the previous case because this 
time $W=W^\ast$ is symmetric, i.e. we have 
automatically $r=l$. Since $W$ depends only on a single kernel 
$c$, the kernels $l,m$ will be functions of $c$ thus are mutually commuting 
together with their derivatives when we introduce $k_s,l_s,m_s$ as before 
and since $c$ is automatically symmetric it follows that the kernels 
$l_s,m_s, l'_s, m'_s$ are symmetric. We proceed as before and write 
\be \label{a.31}
W(s)=L(s):=e^{\frac{s}{2}\;c(A,A)}\;e^{\frac{s}{2}\;c(A^\ast,A^\ast)}
R(s):=E_0(s)\;E_-(s)\;E(s)\; E_+(s),\;\;
\ee
Then the formula for $R' R^{-1}$ can be copied from (\ref{a.8}), (\ref{a.9}).
On the other hand 
\ba \label{a.32}
L' L^{-1} &=&
\frac{1}{2}[c(A,A)+e^{s c(A,A)/2}\; c(A^\ast,A^\ast) e^{-s c(A,A)}]
=
\frac{1}{2}[c(A,A)+c(A^\ast+sc\cdot A,A^\ast+s c\cdot A)]
\nonumber\\
&=&
\frac{1}{2}[(c+s^2 c^3)(A,A)+c(A^\ast,A^\ast)
+2s\;c^2(A^\ast,A) +s {\rm Tr}(c^2)]
\ea
This yields the ODE system
\ba \label{a.33}
s\; {\rm Tr}(c^2) &=& k' -{\rm Tr}(l\cdot e^{-m}\cdot l'\cdot e^{-m})
\nonumber\\
c &=& l'-2\; m'\cdot l
+ l\cdot\;e^{-m}\cdot l'\cdot e^{-m}\;\cdot l
\nonumber\\
s c^2 &=& m' - l\cdot\;e^{-m}\cdot l'\cdot e^{-m}
\nonumber\\
c+s^2 c^3 &=& (e^{-m}\cdot l'\cdot e^{-m})
\ea
We set $f(s)=c(1+s^2 c^2)$, solve from the third equation $m'=s c^2+l\;f$ and 
from the fourth equation $e^{-m} l' e^{-m}=f$ to find from the second
\be \label{a.34}
c=l'-2 (s c^2+l\cdot f) \cdot l+ l\cdot\;f\cdot l 
=l'-2 s c^2 \cdot l- l\cdot\;f\cdot l 
\ee
which is again a Riccati equation but with non-constant coefficients. 
The general solution of a Riccati equation can be obtained in the form
$l=l_0+l_1$ if one can find a special solution $l_0$ so that $l_1^{-1}$ 
solves then a linear inhomogeneous equation. To find the special solution
$l_0$ we set $l_0(s)=g(cs)$ so that (\ref{a.34}) becomes after factoring out 
$c$ and introducing the kernel variable $v=sc,\; (.)'=d/dv$
\be \label{a.35}
0=g'-2v g-g(1+v^2)g-1
=[g'-g^2]-[2vg+1+(vg)^2]=[g'-g^2]-[1+vg]^2
\ee
The second square bracket vanishes for $g=-v^{-1}$ and the first also 
vanishes in this case. Hence $l_0(s)=-(cs)^{-1}$. It follows for the 
general solution with $h(sc)=l_1(s)$
\ba \label{a.36}
0 &=& g'+h'-2v(g+h)-(g+h)(1+v^2)(g+h)-1
\nonumber\\
&=& \{g'-2v g-g(1+v^2)g-1\}+\{h'-2v h+2g(1+v^2)h-h(1+v^2)h\}
\nonumber\\
&=& -h^2\{(h^{-1})'+2[g(1+v^2)+v]h^{-1}+1+v^2\}
\nonumber\\
&=& -h^2\{(h^{-1})'-2 v^{-1} \;h^{-1}+1+v^2\}
\ea
The vanishing of the curly bracketgives a first order linear ODE for 
$h^{-1}$ with inhomogeneity.
We solve it by $h^{-1}=v^2\; j$ (variation of ``constant'' $j$) which gives 
$j'=-1-v^{-2}$ i.e. $j=i-v+v^{-1}$ where $i$ is an integration constant,
thus $h^{-1}=i\;v^2+v-v^3$. It follows
\be \label{a.37}
l=-v^{-1}+[i v^2+v-v^3]^{-1}=  
v^{-1}\;[i v^2+v-v^3]^{-1}(v-iv^2-v+v^3)=  
[i v+1-v^2]^{-1}(-i+v)
\ee
For $i\not=0$ we get $l(0)=-i$. Thus $l(0)=0$ requires $i=0$ and thus yields 
the unique solution
\be \label{a.38}
l(s)=v\;(1-v^2)^{-1}=s \; c\; [1-s^2 c^2]^{-1}
\ee
If follows using $f=c(1+v^2)$
\be \label{a.39}
m'=s c^2+l\;f=c[v+v(1+v^2)(1-v^2)^{-1}]=2cv(1-v^2)^{-1}
\ee
which is uniquely solved using $m(0)=0$ by 
\be \label{a.40}
m(s)=-\ln(1-s^2 c^2)
\ee
We must check whether (\ref{a.38}) and (\ref{a.39}) are consistent with the 
fourth equation in (\ref{a.33})    
\be \label{a.41}
c(1+v^2)=e^{-m}\cdot l'\cdot e^{-m}
=c\;(1-v^2)[\frac{d}{dv} v(1-v^2)^{-1}](1-v^2)
=c\;(1-v^2)[(1-v^2)^{-1}+2v^2(1-v^2)^{-2}](1-v^2)
=c(1+v^2)
\ee
indeed. Finally
\be \label{a.42}
k'
=s\; {\rm Tr}(c^2)+{\rm Tr}(l\cdot e^{-m}\cdot l'\cdot e^{-m})
=s\; {\rm Tr}(c^2)+{\rm Tr}(l\cdot c(1+v^2))
={\rm Tr}(cv[[1+(1-v^2)^{-1}(1+v^2))
=2{\rm Tr}(cv (1-v^2)^{-1}))
\ee
which is uniquely solved using $k(0)=0$ by 
\be \label{a.43}
k(s)=-{\rm Tr}(\ln(1-s^2 c^2))
\ee
\end{proof}

\end{appendix}

\end{document}